\documentclass[11pt,a4paper]{article}
\pdfoutput=1

\usepackage{jcappub}
\usepackage{mathrsfs}
\usepackage{amsthm}
\usepackage[english]{babel}
\usepackage[latin1]{inputenc}
\usepackage[T1]{fontenc}
\usepackage{textcomp}
\usepackage{epstopdf}
\usepackage{pgfplots}
\pgfplotsset{compat=newest}
\usepackage{booktabs}

\newtheorem{lemma}{Lemma}
\newtheorem{proposition}{Proposition}
\newtheorem{corollary}{Corollary}

\def\a{\alpha}

\def\G{\Gamma}
\def\d{\delta}
\def\D{\Delta}
\def\ep{\epsilon}
\def\vep{\varepsilon}
\def\z{\zeta}

\def\Th{\Theta}
\def\i{\iota}

\def\la{\lambda}

\def\m{\mu}
\def\n{\nu}

\def\r{\rho}

\def\s{\sigma}

\def\t{\tau}
\def\f{\phi}

\def\vf{\varphi}
\def\ch{\chi}

\def\O{\Omega}

\newcommand{\ti}[1]{\tilde{#1}}

\newcommand{\Xd}{X^{\cdot}}
\newcommand{\xid}{\xi^{\cdot}}
\newcommand{\chd}{\chi^{\cdot}}

\newcommand{\vfd}{\varphi^{\cdot}}
\newcommand{\tad}{\tilde{\alpha}^{\cdot}}

\newcommand{\oTh}{\overline{\Theta}}

\newcommand{\oXi}{\overline{\Xi}}
\newcommand{\oV}{\overline{V}}
\newcommand{\oW}{\overline{W}}
\newcommand{\oQ}{\overline{Q}}
\newcommand{\barf}{\bar{\phi}}
\newcommand{\bz}{\bar{\zeta}}
\newcommand{\bars}{\bar{\sigma}}
\newcommand{\barS}{\bar{S}}

\newcommand{\bs}{\bar{s}}
\newcommand{\bk}{\bar{k}}
\newcommand{\bd}{\bar{\d}}
\newcommand{\bg}{\bar{g}}
\newcommand{\bK}{\bar{K}}

\newcommand{\hvf}{\hat{\vf}}
\newcommand{\hvfd}{\hat{\vf}^{\cdot}}
\newcommand{\hT}{\hat{T}}
\newcommand{\hg}{\hat{g}}
\newcommand{\hG}{\hat{G}}
\newcommand{\hK}{\hat{K}}
\newcommand{\hR}{\hat{R}}
\newcommand{\hN}{\hat{N}}

\newcommand{\abs}[1]{\lvert#1\rvert}
\newcommand{\babs}[1]{\big\lvert#1\big\rvert}
\newcommand{\Babs}[1]{\Big\lvert#1\Big\rvert}
\newcommand{\bbabs}[1]{\bigg\lvert#1\bigg\rvert}

\newcommand{\norm}[1]{\lVert#1\rVert}

\newcommand{\de}{\partial}
\newcommand{\half}{\frac{1}{2}}

\newcommand{\p}{\prime}
\newcommand{\beq}{\begin{equation}}
\newcommand{\eeq}{\end{equation}}
\newcommand{\beqnn}{\begin{equation*}}
\newcommand{\eeqnn}{\end{equation*}}
\newcommand{\nn}{\nonumber}
\newcommand{\noi}{\noindent}
\newcommand{\sgn}{\mathrm{sgn}}
\newcommand{\xia}{\xi^{a}}
\newcommand{\xib}{\xi^{b}}

\newcommand{\mbf}[1]{\mathbf{#1}}
\newcommand{\mbfscn}{\text{\textbf{\textsc{n}}}}
\newcommand{\bsm}{\boldsymbol{\mu}}
\newcommand{\bsl}{\boldsymbol{\lambda}}
\newcommand{\tmbf}[1]{\tilde{\mathbf{#1}}}
\newcommand{\mcal}[1]{\mathcal{#1}}
\newcommand{\mcalC}{\mathcal{C}}
\newcommand{\mcalM}{\mathcal{M}}
\newcommand{\mcalA}{\mathcal{A}}
\newcommand{\mcalR}{\mathcal{R}}
\newcommand{\mcalB}{\mathcal{B}}
\newcommand{\mcalE}{\mathcal{E}}
\newcommand{\tmcalC}{\tilde{\mathcal{C}}}
\newcommand{\bmcalC}{\bar{\mathcal{C}}}
\newcommand{\mscr}[1]{\mathscr{#1}}
\newcommand{\mscrM}{\mathscr{M}}
\newcommand{\mscrR}{\mathscr{R}}
\newcommand{\mscrI}{\mathscr{I}}

\newcommand{\mbbZ}{\mathbb{Z}}

\newcommand{\mbbR}{\mathbb{R}}

\newcommand{\dem}{\de_{\m}}
\newcommand{\den}{\de_{\n}}

\newcommand{\detau}{\de_{\t}}
\newcommand{\demden}{\de_{\m} \de_{\n}}
\newcommand{\dhvf}{\dot{\hat{\vf}}}

\newcommand{\Msf}{M_{6}^{4}}
\newcommand{\Mft}{M_{5}^{3}}
\newcommand{\Mfs}{M_{4}^{2}}

\title{Thin limit of the 6D Cascading DGP model}

\author{Fulvio Sbis\`a}

\affiliation{Departamento de F\'isica, Universidade Federal do Esp\'irito Santo,\\
Avenida Fernando Ferrari, 514, CEP 29075-910, Vit\'oria, ES, Brazil}

\emailAdd{fulvio.sbisa@gmail.com}

\abstract{A thin limit description of the 6D Cascading DGP model is derived, starting from a conf\mbox{}iguration where both the codimension-1 and the codimension-2 branes are thick. Postulating that the thicknesses of the two branes obey a hierarchic relation, the thin limit is executed in two steps. First the thin limit of the codimension-1 brane is executed, obtaining a system where a ``ribbon'' codimension-2 brane is embedded inside a thin codimension-1 brane with induced gravity, and then the thin limit of the ribbon brane is considered. By proposing a geometric ansatz on the limit conf\mbox{}iguration, the junction conditions which are to hold at the thin codimension-2 brane are derived. The latters are fully non-perturbative and covariant and, together with the Israel junction conditions at the codimension-1 brane and the Einstein equations in the bulk, constitute the looked-for thin limit formulation of the 6D Cascading DGP model. It is commented on how wide is the class of thin source conf\mbox{}igurations which can be placed on the thin codimension-2 brane.}

\keywords{Modif\mbox{}ied gravity; extra dimensions; cosmic strings, domain walls, monopoles.}

\arxivnumber{1710.00437}

\allowdisplaybreaks[4]

\begin{document}

\maketitle

\flushbottom

\section{Introduction}
\label{section Introduction}

The discovery \cite{Riess98, Perlmutter98} of the late time acceleration (LTA) of the Universe has fuelled a lot of interest into theories which modify gravity at large distances (``modif\mbox{}ied gravity'' theories, see \cite{Clifton:2011jh, Joyce:2014kja} and references therein). The possibility of explaining the mysterious acceleration without invoking dark energy has spurred a lot of activity into several directions, like massive gravity theories, $f(R)$ and more general scalar-tensor theories, and braneworld theories. The development of these type of theories is academically interesting in its own right, and has uncovered very beautiful phenomena, like screening mechanisms. Furthermore, theories which modify gravity in the infrared are also attractive from the point of view of the Cosmological Constant (CC) problem \cite{WeinbergCC}, since they may ``degravitate'' a huge vacuum energy leaving only a tiny remnant \cite{Dvali:2002pe, Dvali:2002fz, ArkaniHamed:2002fu}. 

Unfortunately, modif\mbox{}ied gravity theories are prone to several problems, and despite all the ef\mbox{}fort a convincing modif\mbox{}ied gravity explanation of the LTA and CC problems is still lacking. Among the various unsatisfying aspects we may cite the appearance of perturbative ghosts around physically interesting solutions (cosmological or not), instabilities of other type (inside sources and/or outside), and a dangerously low strong coupling scale. In the context of the modif\mbox{}ied gravity approach, a special role is played by braneworld theories, since the existence of extra dimensions and branes is strongly suggested by string theory and so they are somewhat motivated by fundamental physics. In this framework, the prototypical model to tackle the LTA problem is the DGP model \cite{Dvali:2000hr}, which admits self-accelerating cosmological solutions \cite{Deffayet:2000uy}. The failure of the latter model to f\mbox{}it the observational data satisfactorily \cite{Maartens:2006yt, Rydbeck:2007gy, Fang:2008kc} and the presence of ghost instabilities \cite{Luty:2003vm, Nicolis:2004qq, Koyama:2005tx, Koyama:2007za, Gorbunov:2005zk, Charmousis:2006pn} calls for higher-dimensional generalizations of the DGP model. This is also welcome from the point of view of the CC problem, since for the degravitation mechanism to be ef\mbox{}fective in a braneworld theory it is necessary that the codimension be higher than one \cite{Dvali:2007kt, deRham:2007rw}. On the other hand, branes of codimension two or higher are notoriously very delicate to deal with, since their thin limit is in general not well-def\mbox{}ined \cite{Geroch:1987qn}. Also, they were believed to be even more problematic than codimension one branes from the point of view of ghosts \cite{Dubovsky:2002jm}, although the recent paper \cite{Berkhahn:2012wg} is seriously challenging this belief.

\subsection{The Cascading DGP model}

A compelling higher dimensional generalization of the DGP model is the Cascading DGP model \cite{deRham:2007xp}, whose minimal (6D) set-up is qualitatively described by the action
\beq
\label{CascadingDGP6D}
S = \Msf \int_{\mcal{B}} \!\! d^6 X \, \sqrt{-g} \, R + \Mft \int_{\mcalC_1} \!\! d^5 \xi \, \sqrt{-\ti{g}} \, \ti{R} + \int_{\mcalC_2} \!\! d^4 \ch \, \sqrt{-\bar{g}} \, \Big( \Mfs \bar{R} + \mscr{L}_{M} \Big) \quad ,
\eeq
where $\mcal{B}$ is a 6D spacetime which contains a 5D brane $\mcalC_1$ (the codimension-1 brane), inside which lies a 4D brane $\mcalC_2$ (the codimension-2 brane). Both branes are equipped with induced gravity terms, and the energy-momentum tensor is localized on the codimension-2 brane, which is assumed to describe our universe.

This model has the remarkable feature that the codimension-1 brane with induced gravity seem to act as a gravity regularizer. In fact, f\mbox{}ixing the position of the branes and calculating the brane-to-brane propagator on the 4D brane, the divergence characteristic of codimension-2 branes disappears as long as $\Mft \neq 0$ \cite{deRham:2007xp, deRham:2007rw}. Furthermore, keeping into account also the bending modes of the branes, gravity has indeed shown to remain f\mbox{}inite at f\mbox{}irst order in perturbations \cite{Sbisa:2014gwh}. The theory has also unexpected and interesting properties regarding the presence of ghost modes, such as the presence of a critical tension which discriminates (at least at f\mbox{}irst order in perturbations) between stable and unstable pure tension conf\mbox{}igurations \cite{deRham:2007xp, deRham:2010rw, Sbisa:2014vva}. The behaviour of weak gravity depends on the relation between the two free parameters $m_5 \equiv \Mft/\Mfs$ and $m_6 \equiv \Msf/\Mft$: if $m_5 \gg m_6$, weak gravity ``cascades'' from a 6D behaviour at very large scales to a 5D behaviour at intermediate scales to a 4D behaviour at small scales, while if $m_5 \ll m_6$ there is a direct transition from a 6D behaviour at large scales to a 4D behaviour at small scales \cite{deRham:2007xp, deRham:2007rw}. See \cite{deRham:2009wb} for an analysis of the codimension-3 set-up, \cite{Agarwal:2011mg} for an ef\mbox{}fective approach and \cite{Minamitsuji:2008fz, Afshordi:2008rd, Khoury:2009tk, Agarwal:2009gy, Wyman:2010jp, Moyassari:2011nb} for cosmology-related studies.

It is fair to say that the action (\ref{CascadingDGP6D}) a priori does not single out a unique model, since, as already mentioned, the thin limit of a brane of codimension two or higher is not well-def\mbox{}ined. From the point of view of the variational principle, this is mirrored by the fact that f\mbox{}inding extremals of $S$ is arguably an ambiguous procedure. Since for conf\mbox{}igurations of the cascading type the geometry is not smooth at the codimension-2 brane \cite{Sbisa:2014gwh, Sbisa:2014vva}, we cannot perform the variation assuming the f\mbox{}ields are smooth. Specifying exactly the functional space in which the variation is performed is equivalent to provide a prescription to regularize the branes.

\paragraph{Thick branes and induced gravity terms.}
A safe way to study the model is then to work with thick branes. However, in this case the problem is to f\mbox{}ind a way to equip the codimension-2 brane with a thruthful 4D induced gravity term. For example, if we describe the cod-2 brane as a 5D spherical surface with a 5D induced gravity term (or a 6D sphere with a 6D induced gravity term), then the ef\mbox{}fective description at small radiuses contains a 4D induced gravity term plus an additional degree of freedom \cite{deRham:2007rw}. The latter degree of freedom has a strong inf\mbox{}luence on the phenomenology, which is actually very interesting, but nevertheless dif\mbox{}ferent from the dynamic with a pure 4D induced gravity term (a similar argument holds for the dif\mbox{}ferent results obtained by \cite{Dubovsky:2002jm} and \cite{Kolanovic:2003am} in the pure codimension-2 context). See also \cite{Kaloper:2007ap, Kaloper:2007qh} for a related analysis of pure codimension-2 branes with induced gravity.

The aim of the present paper is to provide a realization of the Cascading DGP model where the branes are thin and equipped with induced gravity terms of the correct dimensionality. To do this, we start from a conf\mbox{}iguration where both branes are thick and then develop a well-def\mbox{}ined thin limit description. This is achieved by making an assumption about the internal structures of the thick branes, namely that there exist a hierarchy between their thicknesses. This program was performed successfully in \cite{Sbisa:2014gwh} at linear order in perturbations around pure tension solutions, so the aim here is to extend that analysis to the non-perturbative level.

The paper is structured as follows: in section \ref{section Nested branes with induced gravity} we introduce the concept of thickness hierarchy between the thick branes, and discuss the thin limit of the codimension-1 brane. In section \ref{section The thin limit of the ribbon brane} we propose an ansatz for the behaviour of the geometry when the thin limit of the ribbon brane is performed, and we implement it mathematically. In section \ref{section The codimension-2 pillbox integration} we describe how the codimension-2 junction conditions should emerge, and perform the pillbox integration across the ribbon brane. In section \ref{section Thin limit of the Cascading DGP model} we obtain the codimension-2 junction conditions in terms of geometrically signif\mbox{}icant quantities, and we derive the thin limit description of the Cascading DGP model. Finally, in section \ref{section Discussion} we discuss our results.
\vspace{2mm}

\textbf{Conventions}: For metric signature (``mostly plus''), connection, covariant derivative, curvature tensors and Lie derivative we follow the conventions of Misner, Thorne and Wheeler \cite{MisnerThorneWheeler}. 6D indices are denoted by capital letters, so run from 0 to 5, 5D indices are denoted by latin letters, and run from 0 to 4, while 4D indices are denoted by greek letters and run from 0 to 3. In general, quantities pertaining to the cod-1 brane are denoted by a tilde $\tilde{\phantom{a}}$, while quantities pertaining to the cod-2 brane are denoted by an overbar $\bar{\phantom{a}}$. Abstract tensors are indicated with bold-face letters, while quantities which have more than one component but are not tensors (such as coordinates $n$-tuples for example) are expressed in an abstract way replacing every index with a dot. We use throughout the text the (Einstein) convention of implicit summation on repeated indices, and we will use unit of measure where the speed of light has unitary value $c=1$.

\section{Nested branes with induced gravity}
\label{section Nested branes with induced gravity}

To develop a thin limit description of a system composed by thick codimension-1 and codi-mension-2 branes (to which we refer also as ``physical branes''), it is very useful to introduce the auxiliary construction of the ``mathematical cod-1 and cod-2 branes''. These are thin branes, therefore devoid of an internal structure, around which the physical branes are localized (in a sense to be made precise below). They play a central role in the thin limit procedure, since the physical branes get identif\mbox{}ied with the mathematical ones in the limit, and this construction allows to def\mbox{}ine precisely the notion of thickness hierarchy. We describe below the mathematical and physical structure of our set-up.

\subsection{The set-up}
\label{The set-up}

\subsubsection{The mathematical set-up}
\label{The mathematical set-up}

We assume that the ambient space is a 6D manifold $(\mcalM, \mcalA)$, equipped with a Lorentzian metric $\mbf{g}$.\footnote{We assume that $\mcalA$, $\ti{\mcalA}$ and $\bar{\mcalA}$ are $C^{\infty}$ atlases, and we use the def\mbox{}inition of submanifold given in \cite{DoCarmoBook}.} The manifold $\mcalM$ contains a 5D submanifold $\mcalC_1 \subset \mcalM$, which we call the mathematical codimension-1 brane. This means that $\mcalC_1$ has a manifold structure on its own, which we indicate with $(\mcalC_1, \ti{\mcalA})$, and that the canonical inclusion $\i_1 : \mcalC_1 \to \mcalM$ is an embedding. We indicate with $T \ti{\mcalC}_1$ the tangent bundle of the manifold $(\mcalC_1, \ti{\mcalA})$, and with $T \mcalC_1 \subset T \mcalM$ its push-forward via the embedding $T \mcalC_1 \equiv \i_1^{\star} (T \ti{\mcalC}_1)$. We assume that  $\mcalC_1$ is orientable, and that the metric manifold $\mcalM$ is $\mbbZ_2$-symmetric with respect to $\mcalC_1$.

The manifold $\mcalC_1$ in turn contains a 4D submanifold $\mcalC_2 \subset \mcalC_1$, which we call the mathematical codimension-2 brane. We indicate with $(\mcalC_2, \bar{\mcalA})$ its manifold structure and with $\ti{\i}_2 : \mcalC_2 \to \mcalC_1$ the associated canonical inclusion. It follows that $\mcalC_2$, seen as a subset of $\mcalM$, is a 4D submanifold with canonical inclusion $\i_2 = \i_1 \circ \ti{\i}_2$. We indicate with $T \bar{\mcalC}_2$ the tangent bundle of the manifold $(\mcalC_2, \bar{\mcalA})$, and respectively with $T \ti{\mcalC}_2 \subset T \ti{\mcalC}_1$ and $T \mcalC_2 \subset T \mcalM$ the push-forwards  $T \ti{\mcalC}_2 \equiv \ti{\i}_2^{\star} (T \bar{\mcalC}_2)$ and $T \mcalC_2 \equiv \i_2^{\star} (T \bar{\mcalC}_2)$. We assume that $\mcalC_2$, seen as a submanifold of $\mcalC_1$, is orientable.

The mathematical branes $\mcalC_1$ and $\mcalC_2$ inherit from $\mcalM$ a Lorentzian metric structure: $\mcalC_1$ is characterized by the codimension-1 induced metric $\ti{\mbf{g}} = \i_{1 \star} (\mbf{g})$, where $\i_{1 \star}$ stands for the pull-back with respect to $\i_{1}$, while $\mcalC_2$ is characterized by the codimension-2 induced metric $\bar{\mbf{g}} = \ti{\i}_{2 \star} (\ti{\mbf{g}}) = \i_{2 \star} (\mbf{g})$. We assume that there exists a smooth, time-like vector f\mbox{}ield (a ``time coordinate'') def\mbox{}ined on $T \bar{\mcalC}_2$. This implies that there exists a smooth and space-like unit vector f\mbox{}ield $\ti{\mbfscn}: \bar{\mcalC}_2 \to T \ti{\mcalC}_1$ normal to $T \ti{\mcalC}_2$ (henceforth called the \emph{codimension-2 normal vector f\mbox{}ield}). Similarly, we assume that there exists a smooth and space-like unit vector f\mbox{}ield $\mbf{n} : \ti{\mcalC}_1 \to T \mcalM$ normal to $T \mcalC_1$ (the \emph{codimension-1 normal vector f\mbox{}ield}). This formalizes the idea that the extra dimensions are space-like.

\paragraph{Curvature tensors}
We assume that the connections on $T \mcalM$, $T \tmcalC_1$ and $T \bmcalC_2$ are the Levi-Civita ones, and we indicate respectively with $\nabla$, $\ti{\nabla}$ and $\bar{\nabla}$ the correspondent covariant derivatives. The intrinsic geometries of $\mcalC_1$ and $\mcalC_2$ are described as usual by the Riemann tensor constructed respectively with $\ti{\nabla}$ and $\bar{\nabla}$. The extrinsic geometry of $\mcalC_1$ is described by the (codimension-1) extrinsic curvature $\ti{\mbf{K}}$, a symmetric type $(0,2)$ tensor f\mbox{}ield def\mbox{}ined by
\beq
\label{c-1 extrinsic curvature}
\ti{\mbf{K}} (\ti{\mbf{u}}, \ti{\mbf{v}}) \equiv \mbf{g} \big( \mbf{n}, \nabla_{\i_{1}^{\star}(\ti{\mbf{u}})} \, \i_{1}^{\star}(\ti{\mbf{v}}) \big) 
\eeq
where $\ti{\mbf{u}}, \ti{\mbf{v}} \in T \ti{\mcalC}_1$.\footnote{This def\mbox{}inition is equivalent to $\ti{\mbf{K}} \equiv  \i_{1 \star} \big( - \half \, \mscr{L}_{\mbf{n}} \, \mbf{g} \big)$, where $\mscr{L}_{\mbf{n}}$ indicates the Lie derivative along $\mbf{n}$.} Likewise, the extrinsic geometry of $\mcalC_2$ as a submanifold of $\mcalC_1$ is described by the (codimension-2) extrinsic curvature tensor $\bar{\mbf{K}}$, def\mbox{}ined by
\beq
\label{c-2 extrinsic curvature}
\bar{\mbf{K}} (\bar{\mbf{u}}, \bar{\mbf{v}}) \equiv \ti{\mbf{g}} \big( \ti{\mbfscn}, \ti{\nabla}_{\ti{\i}_{2}^{\star}(\bar{\mbf{u}})} \, \ti{\i}_{2}^{\star}(\bar{\mbf{v}}) \big)
\eeq
where $\bar{\mbf{u}}, \bar{\mbf{v}} \in T \bar{\mcalC}_2$.

For future reference, it is useful to express the extrinsic curvatures in components. Be $p \in \mcalC_1 \subset \mcalM$, and let $(U, \f) \in \mcalA$ and $(\ti{U}, \ti{\f}) \in \ti{\mcalA}$ be local charts respectively of $\mcalM$ and $\mcalC_1$ in $p$. Let $\xid$ indicate the $5$-tuple of coordinates in $\mbbR^5$, and let $\vfd (\xid) \equiv \f \circ \i_1 \circ \ti{\f}^{-1} $ indicate the components expression of the embedding of $\mcalC_1$. Then $\ti{\mbf{K}}$ in components reads
\beq
\label{c-1 extrinsic curvature ok}
\ti{K}_{ab} = n_{_{L}} \bigg( \dfrac{\de^2 \vf^{_{L}}}{\de \xia \xib} + \G^{^{L}}_{_{AB}}\Big\rvert_{\mcalC_1} \dfrac{\de \vf^{_{A}}}{\de \xia} \dfrac{\de \vf^{_{B}}}{\de \xib} \bigg) \quad ,
\eeq
where $n_{_{L}} = n_{_{L}}(\xid)$ and $\G^{_{L}}_{^{AB}}\big\rvert_{\mcalC_1} = \G^{_{L}}_{^{AB}}\big\rvert_{\vfd(\xid)}$ indicate the 6D connection coef\mbox{}f\mbox{}icients evaluated on $\mcalC_1$. Regarding $\bar{\mbf{K}}$, be $q \in \mcalC_2 \subset \mcalC_1$, and let $(\ti{U}, \ti{\f}) \in \ti{\mcalA}$ and  $(\bar{U}, \barf) \in \bar{\mcalA}$ be local charts respectively of $\mcalC_1$ and $\mcalC_2$ in $q$. Let $\chd$ indicate the $4$-tuple of coordinates in $\mbbR^4$, and let $\tad (\chd) \equiv \ti{\f} \circ \i_1 \circ \bar{\f}^{-1} $ indicate the components expression of the embedding of $\mcalC_2$ into $\mcalC_1$. Then $\bar{\mbf{K}}$ can be expressed in a way analogous to (\ref{c-1 extrinsic curvature ok}) by performing the substitutions $n \to \ti{\textsc{n}}$, $\vf \to \ti{\a}$, $\xi \to \ch$ and $\G\big\rvert_{\mcalC_1} \to \ti{\G}\big\rvert_{\mcalC_2}$.

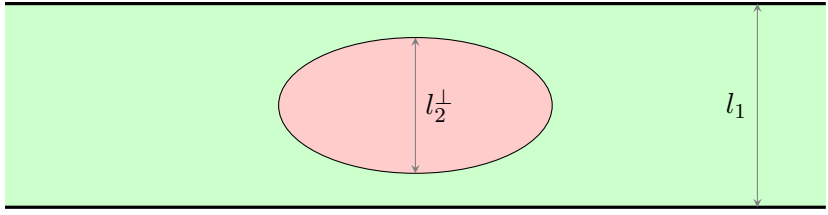
\begin{figure}[t!]
\centering
\begin{tikzpicture}
\begin{scope}[scale=.9,rotate=0,>=stealth]
\fill[color=green!20!white] (-6,-1.5) rectangle (6,1.5);
\filldraw[fill=red!20!white,draw=black] (0,0) ellipse (2cm and 1cm);
\draw[very thick] (-6,-1.5) -- (6,-1.5);
\draw[very thick] (-6,1.5) -- (6,1.5);
\draw[gray,thin,<->] (0,-1) -- node[color=black,anchor=west]{$l_2^{\perp}$} (0,1);
\draw[gray,thin,<->] (5,-1.5) -- node[color=black,anchor=east]{$l_1$} (5,1.5);
\end{scope}
\end{tikzpicture}
\caption{Characteristic scales for the cod-1 brane (green) and the cod-2 brane (ellipse, violet).}
\label{figure 1}
\end{figure}

\subsubsection{The physical set-up}
\label{The physical set-up}

As we mentioned above, the physical cod-1 and cod-2 branes are, naively speaking, localized around the mathematical branes. To make this statement precise, we introduce two reference systems which are adapted to the mathematical branes, the bulk and the codimension-1 Gaussian Normal Coordinates. A pictorial representation of the set-up is given in f\mbox{}igures \ref{figure 1} and \ref{figure 2} .

We recall that, taken any point $p \in \mcalC_1$ and a local coordinate system $(\ti{U}, \tilde{\f}) \in \ti{\mcalA}$ in $p$, we can set-up a reference system on $\mcalM$ in a neighbourhood $U$ of $p$, which is Gaussian Normal to $\mcalC_1$ (see, for example, \cite{CarrollBook}). Any point in $U$ can be reached from $\mcalC_1$ following a unique geodesic of $\mcalM$ normal to $\mcalC_1$, and is identif\mbox{}ied by the value $\z$ of the af\mbox{}f\mbox{}ine parameter along the geodesic (with $\z = 0$ characterizing the points of $\mcalC_1$), together with the coordinates $\xid \in \mbbR^5$ of the starting point on $\mcalC_1$. In the following we refer to this class of reference systems as the \emph{bulk Gaussian Normal Coordinates} (in brief ``bulk GNC'').

Regarding the (physical) codimension-1 brane $\mcalB_{1}$, we assume f\mbox{}irst of all that it can be completely covered by the bulk Gaussian Normal Coordinates, eventually patching several local charts of $\mcalC_1$. Under this assumption, we can associate to every point $p \in \mcalB_{1}$ the proper length $l_1(p)$ of the geodesic interval which results as the intersection of $\mcalB_1$ with the normal geodesic passing through $p$. We formalize the idea that $\mcalB_{1}$ is a thick codimension-1 brane by asking that the supremum $l_1$ of the values $l_1(p)$ as $p$ varies in $\mcalB_{1}$ ($l_1 \equiv \sup_{p \in \mcalB_{1}} \big\{ l_1(p) \big\}$) is f\mbox{}inite, and call $l_1$ the thickness of the physical cod-1 brane.

As anticipated above, we assume that inside $\mcalB_{1}$ there is an additional substructure, the physical codimension-2 brane $\mcalB_{2}$ (see \cite{Morris:1997hj, Edelstein:1997ej} for f\mbox{}ield theory realizations of this type of conf\mbox{}igurations). We can introduce the notions of ``orthogonal'' and ``parallel'' thicknesses of $\mcalB_{2}$ (where orthogonal and parallel are intended relatively to $\mcalC_1$) as follows. The notion of orthogonal thickness can be def\mbox{}ined exactly as we did for the thickness of $\mcalB_{1}$. Taken any $q \in \mcalB_{2}$, there is one and only one geodesic normal to $\mcalC_1$ which passes through $q$ (since the bulk GNC are def\mbox{}ined on all $\mcalB_{1}$ and therefore on all $\mcalB_{2}$). We can then associate to $q$ the proper length $l_{2}^{\perp} (q)$ of the geodesic interval which results as the intersection of $\mcalB_{2}$ with the chosen geodesic. We def\mbox{}ine the orthogonal thickness of $\mcalB_{2}$ as the supremum of this set of thicknesses $l_{2}^{\perp} \equiv \sup_{q \in \mcalB_{2}} \big\{ l_{2}^{\perp}(q) \big\}$.

To def\mbox{}ine the parallel thickness of $\mcalB_{2}$, we consider the reference systems on $\mcalC_1$ which are Gaussian Normal to $\mcalC_2$. In complete analogy with what done above, taken a point $r \in \mcalC_2$ and a local coordinate system $(\bar{U}, \barf) \in \bar{\mcalA}$ in $r$, we can set-up a coordinate system locally in $\mcalC_1$ by following the geodesics of $\mcalC_1$ which are normal to $\mcalC_2$. We call this class of reference systems the \emph{codimension-1 Gaussian Normal Coordinates}, in brief ``cod-1 GNC''. We now call $\mcalR$ the intersection of $\mcalB_{2}$ with $\mcalC_1$, and assume that $\mcalR$ can be completely covered by cod-1 Gaussian Normal Coordinates, eventually patching several local charts of $\mcalC_2$ (therefore any point of $\mcalR$ is individuated by the af\mbox{}f\mbox{}ine parameter $\t$ of a specif\mbox{}ic geodesic of $\mcalC_1$ normal to $\mcalC_2$, and by the coordinates $\chd \in \mbbR^4$ of the starting point on $\mcalC_2$). We can then associate to every point $s \in \mcalR$ the proper length $l_{2}^{\shortparallel}(s)$ of the geodesic interval resulting as the intersection of $\mcalR$ and the geodesic of $\mcalC_1$ normal to $\mcalC_2$ passing through $s$. We def\mbox{}ine the parallel thickness of $\mcalB_{2}$ as the supremum of this set of thicknesses $l_{2}^{\shortparallel} \equiv \sup_{s \in \mcalR} \big\{ l_{2}^{\shortparallel}(s) \big\}$. The idea that $\mcalB_{2}$ is a thick codimension-2 brane is formalized by asking that $l_{2}^{\perp}$ and $l_{2}^{\shortparallel}$ are f\mbox{}inite.

\begin{figure}[t!]
\centering
\begin{tikzpicture}
\begin{scope}[scale=.46,>=stealth]
\fill[color=green!40!white] (1.565,5.83) -- (15.565,5.81) -- (14.521,6) -- (0.521,6) -- (1.565,5.83);
\fill[color=black!05!white] (-0.52,-6) -- (1.565,5.83) -- (15.565,5.83) -- (13.48,-6) -- (-0.52,-6);
\fill[color=green!20!white,rotate=80] (-6,-0.5) rectangle (6,0.5);
\filldraw[fill=red!20!white,draw=black,rotate=80] (0,0) ellipse (2cm and 0.4cm);
\draw[thick,rotate=80] (-6,-0.5) -- (6,-0.5);
\draw[thick,rotate=80] (-6,0.5) -- (6,0.5);
\draw[draw=black,dashed,rotate around={-10:(14,0)}] (14,0) ellipse (0.4cm and 2cm);
\draw[thick,rotate around={-10:(14,0)}] (14.5,-6) -- (14.5,6);
\draw[thin,dashed,rotate around={-10:(14,0)}] (13.5,-6) -- (13.5,6);
\draw[thin,dashed] (0.33,1.97) -- (14.33,1.97);
\draw[thin,dashed] (-0.33,-1.97) -- (13.67,-1.97);
\draw[thin,color=red,<-] (0,0) -- (-0.7,2) node[color=black,anchor=south]{$\mcalB_2$};
\draw[thin,color=red,<-] (0.68,3.5) -- (-0.7,4.5) node[color=black,anchor=south]{$\mcalB_1$};
\end{scope}
\end{tikzpicture}
\begin{tikzpicture}
\begin{scope}[scale=.46,>=stealth]
\fill[color=green!20!white] (0.33,1.97) -- (1.03,5.9) -- (15.03,5.9) node[color=black,anchor=north east]{$\mcalC_1$} -- (14.33,1.97) -- (0.33,1.97);
\fill[color=green!20!white] (-0.33,-1.97) -- (13.67,-1.97) -- (12.97,-5.9) -- (-1.03,-5.9) -- (-0.33,-1.97);
\draw[very thin,rotate around={-10:(14,0)}] (14,0) ellipse (0.4cm and 2cm);
\draw[thick,rotate around={-10:(14,0)}] (14,-6) -- (14,6);
\fill[color=red!20!white] (-0.33,-1.97) -- (0.33,1.97) -- (14.33,1.97)node[color=black,anchor=north east]{$\mcalR$} -- (13.67,-1.97) -- (-0.33,-1.97);
\draw[very thin, rotate=80] (0,0) ellipse (2cm and 0.4cm);
\draw[thick,rotate=80] (-6,0) -- (6,0);
\draw[gray,thin,<->=stealth] (6.65,-1.97) -- (7.35,1.97) node[color=black,anchor=north east]{$l_2^{\shortparallel}$};
\draw[thick,dashed] (10:14);
\draw[very thin] (0.33,1.97) -- (14.33,1.97);
\draw[very thin] (-0.33,-1.97) -- (13.67,-1.97);
\draw[very thick] (0,0) -- (14,0);
\draw[thin,color=red,<-] (10,0) -- (14,-3) node[color=black,anchor=north]{$\mcalC_2$};
\end{scope}
\end{tikzpicture}
\caption{Pictorial representation of the thick (left) and thin (right) nested branes.}
\label{figure 2}
\end{figure}
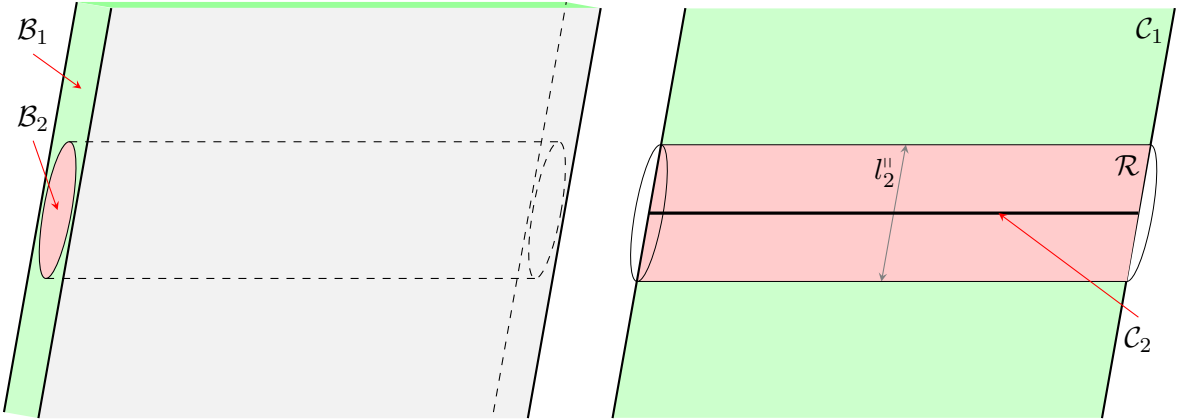

The main physical assumption we make is that there exists a hierarchy between the thicknesses of the physical codimension-1 and codimension-2 branes, namely that $l_{2}^{\perp} \sim l_{1}$ and $l_{2}^{\shortparallel} \gg l_{1}$. See \cite{Sbisa:2014gwh} for a related discussion and a pictorial representation of this assumption. Note that we are making assumptions only on the overall shape of the physical branes, while we are not making any assumption about their internal structures and localization mechanisms (apart from assuming that they can be covered by the Gaussian Normal Coordinates, which can be seen as a mild assumption on their internal structure).

\subsection{The thin limit of the codimension-1 brane}
\label{The thin limit of the codimension-1 brane}

Naively speaking, the thickness hierarchy between $\mcalB_{1}$ and $\mcalB_{2}$ means that, in the direction orthogonal to $\mcalC_1$, the widths of $\mcalB_{1}$ and $\mcalB_{2}$ are of the same order of magnitude, while $\mcalB_2$ is much broader in the direction parallel to $\mcalC_1$ and normal to $\mcalC_2$ (the two branes are both inf\mbox{}inite in the directions parallel to $\mcalC_2$). This implies that, if we study the system focusing on scales much larger than $l_1$ (but not necessarily larger than $l_{2}^{\shortparallel}$), we can ef\mbox{}fectively describe our system as if the physical codimension-1 brane were thin, thereby identifying $\mcalB_{1}$ and $\mcalC_1$. This is possible since the thin limit of a codimension-1 brane is well def\mbox{}ined \cite{Geroch:1987qn}, and permits to equip the codimension-1 brane with a 5D induced gravity term. In this description, $\mcalB_{2}$ becomes identif\mbox{}ied with $\mcalR$, the intersection of $\mcalB_{2}$ and $\mcalC_1$. Since $\mcalR$ is inf\mbox{}inite in the directions parallel to $\mcalC_2$, while it lies inside a strip of thickness $l_{2}^{\shortparallel}$ as far as the coordinate $\t$ of the cod-1 GNC is concerned, it is the 5D analogue of a ribbon. Henceforth, we refer to it as the ``ribbon brane''.

In this ef\mbox{}fective description, the equations of motion are the (source-free) Einstein equations in the bulk and the (sourced) Israel junction conditions \cite{Israeljc} on $\mcalC_1$, namely
\begin{equation}
\label{junctionconditionseq nonZ2}
\Msf \Big[ \tmbf{K} -  \tmbf{g} \, \, tr (\tmbf{K}) \Big]_{\pm} + \Mft \, \tmbf{G} = \,\, \tmbf{T} \quad ,
\end{equation}
where $\mbf{G}$ and $\tmbf{G}$ are the Einstein tensors built respectively from $\mbf{g}$ and $\tmbf{g}$, $tr$ indicate the trace, and $[\phantom{a}]_{_{\pm}} = \vert_{_{+}} - \vert_{_{-}}$. The source term $\tmbf{T}$ is the energy-momentum tensor associated to $\mcalC_1$, which is obtained from the energy-momentum tensor of the physical branes $\mcalB_1$ and $\mcalB_2$ by performing a pillbox integration normally to $\mcalC_1$ (see for example \cite{MisnerThorneWheeler}). We make the further assumption that energy and momentum are localized only inside $\mcalB_2$, while $\mcalB_1$ is characterized only by induced gravity. This means that, expressing $\tmbf{T}$ in cod-1 GNC (henceforth we indicate quantities evaluated in this reference system with an overhat $\!\hat{\phantom{i}}\,$), we have
\beq
\label{Andreja}
\hT_{ab}(\t, \chd) = 0 \qquad \text{for} \qquad \abs{\t} > l \quad ,
\eeq
where we def\mbox{}ined $l \equiv l_{2}^{\shortparallel}/2$. Furthermore, we formalize the idea that momentum does not f\mbox{}low out of the ribbon brane by asking that the pillbox integration across the cod-2 brane of the normal and mixed components of $\hT_{ab}$ vanishes
\beq
\label{angelica}
\int_{-l}^{+l} d \t \,\, \hat{T}_{\t\t}(\t, \chd) = \int_{-l}^{+l} d \t \,\, \hat{T}_{\t\m}(\t, \chd) = 0 \quad .
\eeq

As we already mentioned, we assume that a $\mathbb{Z}_{2}$ ref\mbox{}lection symmetry holds across $\mcalC_1$. When the thin limit for the physical cod-1 brane is taken, this implies that, from a mathematical point of view, $\mcalC_1$ is not anymore a submanifold of a six-dimensional $C^{\infty}$ metric manifold.\footnote{In fact the bulk metric is not smooth at $\mcalC_1$, as the discontinuity of $\tmbf{K} -  \tmbf{g} \, \, tr (\tmbf{K})$ in (\ref{junctionconditionseq nonZ2}) indicates.} Instead, now the ambient manifold $\mcalM$ is composed of two isomorphic $C^{\infty}$ metric manifolds with boundary, which are glued at their common boundary $\mcalC_1$ in such a way that their union is $\mathbb{Z}_{2}$-symmetric with respect to $\mcalC_1$. Therefore, to f\mbox{}ind $\mcalM$ we look for solutions of the system
\begin{align}
\label{Bulkeq}
\mbf{G} =& \,\, 0 \qquad \,\, (\textrm{bulk})\\[2mm]
\label{junctionconditionseq}
2 \Msf \Big( \tmbf{K} -  \tmbf{g} \, \, tr (\tmbf{K}) \Big) + \Mft \, \tmbf{G} =& \,\, \tmbf{T} \qquad (\textrm{boundary})
\end{align}
for a $C^{\infty}$ metric manifold $\mcalM_{_{+}}$ with boundary, such that the normal vector $\mbf{n}$ points into $\mcalM_{_{+}}$. We then create a mirror copy of $\mcalM_{_{+}}$ and glue the two copies at the boundary, which we identify with $\mcalC_1$, in a $\mbbZ_2$-symmetric way. The manifold $\mcalM$ thus obtained is, by construction, $\mathbb{Z}_{2}$-symmetric with respect to $\mcalC_1$, and solves the equation (\ref{junctionconditionseq nonZ2}). To f\mbox{}ind $\mcalM_{_{+}}$, sometimes the following procedure is used: we f\mbox{}irst look for a $C^{\infty}$ metric manifold $\mscrM$ which is \emph{not} $Z_2$-symmetric, and whose metric solves (\ref{Bulkeq}). Then we look for an orientable cod-1 submanifold $\mcalC_1$ of $\mscrM$ whose embedding solves (\ref{junctionconditionseq}). Since $\mcalC_1$ divides $\mscrM$ in two disconnected 6D submanifolds with $\mcalC_1$ as the common boundary, $\mcalM_{_{+}}$ can be recognized as the one which $\mbf{n}$ points into.

This set-up is in fact very familiar, being just a codimension-1 DGP model (in 6D) with the only addition that energy and momentum are conf\mbox{}ined inside a 5D ribbon. It was already introduced in \cite{Sbisa:2014gwh} with the name of \emph{nested branes with induced gravity} set-up (borrowing the term ``nested branes'' from \cite{Gregory:2001xu}). The crucial point in this paper (as well as in the perturbative analysis of \cite{Sbisa:2014gwh}) is if the thin limit of the ribbon brane inside the (already thin) codimension-1 brane is well-def\mbox{}ined or not. We devote the rest of the paper to the analysis of this point.

\section{The thin limit of the ribbon brane}
\label{section The thin limit of the ribbon brane}

We are now interested in developing an ef\mbox{}fective description at length scales much larger than $l_{2}^{\shortparallel}$. Supposing that the codimension-1 energy-momentum tensor $\hT_{ab}(\chd, \t)$ changes very slowly on length scales of the order of $l_{2}^{\shortparallel}$ when we move along the 4D coordinates $\chd$, we can look for a thin limit description of the model def\mbox{}ined by (\ref{Andreja})--(\ref{junctionconditionseq}), where ``thin'' refers to the localization length of the ribbon brane. If feasible, such a description would permit to describe the evolution of the geometry outside the ribbon brane without having to specify the details of its internal structure, and in turn it would provide an ef\mbox{}fective description on scales $\gg l_{2}^{\shortparallel} \gg l_{2}^{\perp} \sim l_{1}$ of the exact conf\mbox{}iguration with both $\mcalB_1$ and $\mcalB_2$ thick, thereby allowing to ignore the details of the internal structures of both physical branes.

Generally speaking, the thin limit of a physical brane can be def\mbox{}ined if, when focusing on length scales much bigger than the brane thickness, the conf\mbox{}iguration of the f\mbox{}ields (in our case, $g_{_{AB}}$ and $\vf^{_{A}}$) outside of the brane (the \emph{external conf\mbox{}iguration}) are not inf\mbox{}luenced by the details of how the source conf\mbox{}iguration (in our case, $\hT_{ab}$) is distributed inside the brane, but are determined only by a suitably def\mbox{}ined ``integrated'' amount of the source. In practice, one associates to the physical brane a mathematical brane, to which the physical brane is to reduce in the limit. Then the formulation of a thin limit description rests on two ingredients. The f\mbox{}irst is a systematic way to associate, to every thick source conf\mbox{}iguration, a \emph{thin source conf\mbox{}iguration} which is def\mbox{}ined on the mathematical brane. The second ingredient is a law which permits to derive the external f\mbox{}ields conf\mbox{}iguration from the thin source conf\mbox{}iguration. Usually this law is found by supplementing the source-free f\mbox{}ield equations with some (``junction'') conditions, which are to hold at the mathematical brane. These conditions relate the thin source conf\mbox{}iguration to the behaviour of suitable quantities which describe the conf\mbox{}iguration of the external f\mbox{}ields in the proximity of the brane.

Unless one is able to solve the equations of motions exactly, in which case it is possible to identify the ingredients mentioned above by studying the behaviour of the solutions, it is necessary to propose an ansatz. That is, one needs to propose how the thin source conf\mbox{}iguration is to be constructed, to guess which are the quantities relevant to describe the geometry in the proximity of the thin brane, and to suggest a way to integrate the equations of motion to f\mbox{}ind conditions which link these objects. If the description one obtains this way turns out to be consistent, then a posteriori one can say that the thin limit is well-def\mbox{}ined, and to have obtained a thin limit description.

\subsection{The geometric ansatz}
\label{The geometric ansatz}

To propose an ansatz for the thin limit of the ribbon brane, we take inspiration from a specif\mbox{}ic subset of source conf\mbox{}igurations, where the equations of motion can be solved exactly, and then put forward an ansatz for the general case. Before doing that, we introduce some geometrical concepts which are central to the discussion to follow. 

\subsubsection{Ridge conf\mbox{}igurations}
\label{Ridge configurations}

Let's consider a point $\chd \in \mcalC_2$, and study the behaviour of the vector $\mbf{n}$ as we move along the geodesic normal to $\mcalC_2$ passing through $\chd$. We say that the brane $\mcalC_1$ has a \emph{ridge} along $\mcalC_2$ at $\chd$ if $\mbf{n}$ displays a f\mbox{}inite jump at $\t = 0$, and if $\mbf{n}(\chd, 0^-)$ and $\mbf{n}(\chd, 0^+)$ are not antiparallel. If $\mbf{n}$ displays a f\mbox{}inite jump at $\t = 0$ and $\mbf{n}(\chd, 0^-)$ and $\mbf{n}(\chd, 0^+)$ are antiparallel, we say that $\mcalC_1$ has a \emph{cusp} along $\mcalC_2$ at $\chd$.\footnote{We therefore change terminology with respect to \cite{Sbisa:2014gwh, Sbisa:2014vva, Sbisa:2014dwa}, where we called a ``cusp'' what we now call a ``ridge''.} We say that $\mcalC_1$ has a ridge along $\mcalC_2$ if it has a ridge along $\mcalC_2$ at every $\chd \in \mcalC_2$, and if $\mbf{n}$ is continous when moving along the directions parallel to $\mcalC_2$. In this case, we say that $\mcalC_2$ lies at the crest of the ridge. We def\mbox{}ine the local \emph{dihedral angle} $\Upsilon(\chd)$ of the ridge as the the geometrical angle (that is, positive and unoriented) between the normal vector on the two sides of $\mcalC_2$
\beq
\label{Sofia Loren}
\Upsilon(\chd) = \arccos \big\langle \mbf{n}(\chd, 0^-), \mbf{n}(\chd, 0^+) \big\rangle \quad ,
\eeq
where $\langle \phantom{i}, \phantom{i} \rangle$ denotes the scalar product. A cusp corresponds to the case $\Upsilon = \pi$.

Note that the def\mbox{}inition of $\mcalC_1$ having a ridge along $\mcalC_2$ is independent of the orientation chosen for the normal vector. In fact, if $\mbf{n}$ displays a f\mbox{}inite jump at $\t = 0$ and if $\mbf{n}(\chd, 0^-)$ and $\mbf{n}(\chd, 0^+)$ are not antiparallel, then the same holds for $-\mbf{n}\,$. To distinguish the two cases, we def\mbox{}ine an \emph{oriented ridge conf\mbox{}iguration} to be given by a cod-1 brane $\mcalC_1$ which has a ridge along $\mcalC_2$, together with a choice for the orientation of $\mcalC_1$ (i.e.~with a def\mbox{}inite choice of $\mbf{n}$). This notion is clearly more apt to describe the manifold $\mcalM_{_{+}}$. As we discuss in appendix \ref{appendix ridge}, the suitable quantity to describe an oriented ridge conf\mbox{}iguration is not the dihedral angle, but is the \emph{oriented dihedral angle} $\O(\chd)$ which can be expressed as
\beq
\label{The Final Countdown}
\O(\chd) = - \, \sgn \, \big\langle \mbfscn_{_{\!-}}(\chd), \mbf{n}_{_{+}}(\chd) \big\rangle \,\, \arccos \big\langle \mbf{n}_{_{-}}(\chd), \mbf{n}_{_{+}}(\chd) \big\rangle \quad ,
\eeq
where $\textsc{n}^{_{\! A}} = \de_{\t} \hvf^{_{A}}$ and
\begin{align}
\label{flauntit}
\mbfscn_{_{\!\pm}}(\chd) &= \lim_{\t \to 0^{\pm}} \mbfscn (\chd, \t) & \mbf{n}_{_{\pm}}(\chd) &= \lim_{\t \to 0^{\pm}} \mbf{n} (\chd, \t) \quad .
\end{align}

\subsubsection{The ansatz}
\label{The ansatz}

Let's consider now a pure tension source conf\mbox{}iguration. In this case (see \cite{Sbisa:2014dwa} for a detailed discussion), the thin limit of the ribbon brane is well-def\mbox{}ined, and the exact solution relative to a tension $\la$ is invariant with respect to translations in the 4D directions. Indeed, when focusing on length scales $\gg l$, the exact solution is well approximated by a thin conf\mbox{}iguration where the bulk is f\mbox{}lat and the cod-1 brane $\mcalC_1$, away from $\mcalC_2$, has f\mbox{}lat intrinsic and extrinsic geometry. Along $\mcalC_2$, the cod-1 brane $\mcalC_1$ of the thin conf\mbox{}iguration displays a ridge of constant dihedral angle, while the cod-2 brane $\mcalC_2$ has f\mbox{}lat intrinsic geometry and f\mbox{}lat extrinsic geometry (seen as a submanifold of $\mcalC_1$). The link between the thick conf\mbox{}iguration and the thin one is provided by performing the pillbox integration of the 4D components of the (thick) Israel junction conditions, more precisely the oriented dihedral angle $\O$ of the thin conf\mbox{}iguration is individuated uniquely by the pillbox integration of the (thick) cod-1 energy momentum tensor across the ribbon brane.

We then propose the following ansatz for the general case. We assume that, when focusing on length scales $\gg l$, the exact conf\mbox{}iguration correspondent to a generic $\tmbf{T}$ (obeying (\ref{Andreja}) and (\ref{angelica})) is well approximated by a conf\mbox{}iguration where the bulk metric is smooth, and the cod-1 brane $\mcalC_1$ has a ridge along $\mcalC_2$, while it is smooth away from it. The geometry is smooth in the 4D directions, and in particular the oriented dihedral angle is a smooth function of $\chd$. The cod-2 brane $\mcalC_2$ has smooth intrinsic geometry but, since it lies on the crest of the ridge, its extrinsic curvature relative to $\mcalC_1$ in general is not. In fact in cod-1 GNC we have
\beq
\Big( \mscr{L}_{\ti{\mbfscn}} \, \ti{\mbf{g}} \Big)_{\!ab} = \, \d_{a}^{\,\, \m} \, \d_{b}^{\,\, \n} \, \Big( \detau \de_{(\m} \hvf^{_{A}} \,\, \de_{\n)} \hvf^{_{B}} \,\, g_{_{AB}}\big\rvert_{\hvfd} + \dem \hvf^{_{A}} \,\, \den \hvf^{_{B}} \,\, \detau \hvf^{_{L}} \,\, \de_{_{\! L}} g_{_{AB}}\big\rvert_{\hvfd} \Big) \quad ,
\eeq
and $\detau \de_{\m} \hvf^{_{A}}$, $\detau \hvf^{_{L}}$ are in general discontinuous in $\t = 0$ as a consequence of the ridge conf\mbox{}iguration. In the pure tension case, $\mscr{L}_{\ti{\mbfscn}} \, \ti{\mbf{g}}$ is continuous across $\t = 0$ only because $\de_{\m} \hvf^{_{A}}$ is independent of $\t$ and $g_{_{AB}}\big\rvert_{\hvfd}$ is constant, which is a fortuitous situation. Essentially, in passing from the pure tension to a generic source conf\mbox{}iguration we allow the dihedral angle to depend smoothly on $\chd$, the cod-1 curvature tensors to be non-zero (remaining smooth away from $\mcalC_2$), and the bulk geometry to be curved (remaining smooth).

We moreover assume that the pillbox integration of the cod-1 energy-momentum tensor across the ribbon brane is the appropriate thin source conf\mbox{}iguration. The relations (\ref{angelica}) imply that
\beq
\label{cristin}
\int_{-l}^{+l} d \t \,\, \hT_{ab}(\chd, \t) = \d_{a}^{\, \, \m} \,\, \d_{b}^{\, \, \n} \,\, \bar{T}_{\m\n}(\chd) \quad ,
\eeq
and we henceforth refer to $\bar{T}_{\m\n}$ as the \emph{codimension-2 energy momentum tensor}, which we assume to be as well a smooth function of $\chd$. Accordingly, we expect that the junction conditions associated to the thin limit of the ribbon brane, henceforth called the \emph{codimension-2 junction conditions}, are to be found by performing a pillbox integration of the 4D components of the Israel junction conditions (\ref{junctionconditionseq}) across the ribbon brane.

\subsubsection{Implementation of the thin limit}
\label{Implementation of the thin limit}

To actually perform the thin limit, instead of focusing on scales much larger than $l$, it is more ef\mbox{}fective to keep the scale of observation f\mbox{}ixed, and send the thickness of the ribbon brane to zero. This corresponds to constructing a suitable family of conf\mbox{}igurations, parametrized by a real number $\r \in (0, l \, ]$, which converge in the limit $\r \to 0^+$ to the conf\mbox{}iguration which, according to our ansatz, describes our system when focusing on length scales $\gg l$. Taking into account the equation (\ref{cristin}), we then construct a family of ribbon branes whose localization length is $\r$, and such that the integrated amount of the energy-momentum tensor is independent of $\r$. Namely, we construct a family of smooth source conf\mbox{}igurations $\hT_{ab}^{_{(\r)}}$ such that
\beq
\label{Andreja rho}
\hT_{ab}^{_{(\r)}}(\chd, \t) = 0 \qquad \text{for} \qquad \abs{\t} > \r \quad ,
\eeq
and such that their pillbox integration across the ribbon brane obeys
\beq
\label{cristina rho}
\int_{-\r}^{+\r} d \t \,\, \hT_{ab}^{_{(\r)}}(\chd, \t) = \d_{a}^{\, \, \m} \,\, \d_{b}^{\, \, \n} \,\, \bar{T}_{\m\n}(\chd) \quad .
\eeq
Likewise, we introduce a family of smooth bulk metrics $g_{_{AB}}^{_{(\r)}}(\Xd)$ and a family of smooth embedding functions $\hvf_{^{\! (\r)}}^{_{A}}(\chd, \t)$, which converge to a conf\mbox{}iguration which displays a ridge along $\mcalC_2$.

There is however a signif\mbox{}icant degree of arbitrarity in how we construct the parametrized families $g_{_{AB}}^{_{(\r)}}$ and $\hvf_{^{(\r)}}^{_{A}}$, due to the freedom of choosing the reference system in the bulk (we remind that we are using the cod-1 GNC on $\mcalC_1$). A conf\mbox{}iguration where the normal vector $\mbf{n}$ is discontinuous can be described by a smooth bulk metric and an embedding which is continuous but not derivable at $\mcalC_2$, or by a discontinuous bulk metric and a smooth embedding (or a mixture of the two things). To perform the pillbox integration and derive the cod-2 junction conditions, it is advisable to tailor the properties of the limit conf\mbox{}igurations to suit best our geometrical ansatz. Although all reference systems are in principle equivalent, some are able to encode our ansatz more naturally than others. This is even more important since, as we see below, to encode our ansatz we need to specify the converge properties of the families $g_{_{AB}}^{_{(\r)}}$ and $\hvf_{^{\! (\r)}}^{_{A}}$ when $\r \to 0^+$.

Keeping the embedding straight (or smooth) has the disadvantage that the bulk metric has to be discontinuous not only at $\mcalC_2$, but also in the bulk away from the branes, where the vacuum Einstein equations hold and the geometry is smooth \cite{deRham:2010rw}. Therefore, such a choice involves a f\mbox{}ictitious (coordinate) discontinuity, which in our opinion hinders the geometrical intuition. For this reason, we believe that the most natural choice is to adopt a gauge where the discontinuity of $\mbf{n}$ at $\mcalC_2$ is supported purely by the embedding of $\mcalC_1$. The importance of choosing a gauge which is geometrically suited to the problem has already been stressed in \cite{Sbisa:2014gwh}, where the advantages of a ``bulk-based'' approach were discussed. Clearly, it is necessary that, a posteriori, the junction conditions we obtain can be cast in a covariant form.

\subsubsection{The limit conf\mbox{}iguration}
\label{The limit configuration}

For the reasons discussed above, we impose that the family $g_{_{AB}}^{_{(\r)}}(\Xd)$ converges to a smooth conf\mbox{}iguration $g_{_{\! AB}}(\Xd)$, and that the family $\hvf_{^{\! (\r)}}^{_{A}}(\chd, \t)$ converges to a conf\mbox{}iguration $\hvf^{_{A}}(\chd, \t)$ which is continuous, smooth with respect to the variables $\chd$, and separately smooth for $\t \geq 0$ and for $ \t \leq 0$. By ``smooth for $\t \geq 0\,$'' we mean that $\hvf^{_{A}}(\chd, \t)$ is smooth for $\t > 0$ and every partial derivative of any order has a f\mbox{}inite limit for $\t \to 0^+$, and an analogue meaning holds for ``smooth for $\t \leq 0\,$''.

To be more specif\mbox{}ic about the domains, $\hvf_{^{\! (\r)}}^{_{A}}$ and $\hvf^{_{A}}$ are functions of the cod-1 GNC, which in general can be def\mbox{}ined only locally. For every $q \in \mcalC_2$, we assume that $\hvf^{_{A}}$ and all the functions of the family $\hvf_{^{\! (\r)}}^{_{A}}$ are def\mbox{}ined in a coordinate system which extends beyond $\t = \pm l$ and contains a 4D neighbourhood of $q \,$. More precisely, we assume that we can f\mbox{}ind a (4D) local reference system $(\bar{U}_{q}, \bar{\f}_{q}) \in \bar{\mcalA}$ such that the cod-1 GNC are well-def\mbox{}ined in an open set containing $\oW_{\!\!q}$, where $W_{\!q} = \bar{U}_{q} \times (-l, l)$. Therefore, once picked a $q \in \mcalC_2$, we consider $\hvf_{^{\! (\r)}}^{_{A}}$ and $\hvf^{_{A}}$ to be def\mbox{}ined on $W_{\!q}$, which implies that they can be extended with continuity (together with their partial derivatives) to $\oW_{\!\!q}$. Regarding $g_{_{AB}}^{_{(\r)}}$ and $g_{_{AB}}$, we assume that we can f\mbox{}ind a reference system $(V_{\!q}, \f_q)$ of the ambient manifold such that the images $\hvf^{\cdot}_{^{\! (\r)}}(\oW_{\!\!q})$ and $\hvf^{\cdot}(\oW_{\!\!q})$ are contained inside $V_{\!q}$ for every $\r$, and that $g_{_{AB}}^{_{(\r)}}$ and $g_{_{AB}}$ can be extended with continuity (together with their partial derivatives) to $\oV_{\!\!q}$. Note that we indicate with $\oV_{\!\!p}$ the closure of $V_{p}$, and with with $\oW_{\!\!q}$ the closure of $W_{\!q}$. The notation is possibly confusing, since a small overbar indicates quantities which pertain to the cod-2 brane while a big overbar indicate the closure of a set.\footnote{In particular, $\oW_{\!\!q} = \overline{\bar{U}}_{\!q} \times [-l, l \, ]$.}

\subsection{The convergence properties}
\label{Convergence properties}

Although necessary, f\mbox{}ixing the properties of the ($\r \to 0^+$) limit conf\mbox{}igurations is not by itself suf\mbox{}f\mbox{}icient to assure that the construction faithfully reproduces our geometric ansatz about the thin limit. This point, which for codimension-1 branes can be regarded as a technicality, is of central importance for higher-codimension branes. It is related to the well-known mathematical fact that it is possible to construct a family of smooth functions which converges (point-wise) to the zero-function, but whose integral over a f\mbox{}ixed interval does not tend to zero. For example this can happen when the family of functions has unbounded peaks with compact support, provided the measure of the supports tends to zero and the peak value diverges in a suitable way, and the position of the peak progressively shifts. For ease of discussion, we (improperly) refer to these types of behaviours as ``disappearing divergence'' behaviours.

This means that we can construct dif\mbox{}ferent realizations of the families $g^{_{(\r)}}_{_{AB}}$ and $\hvf_{^{\! (\r)}}^{_{A}}$, correspondent to the \emph{same} $g_{_{AB}}$ and $\hvf^{_{A}}$, which give dif\mbox{}ferent results when performing the pillbox integration of the Israel junction conditions, and therefore produce dif\mbox{}ferent codimension-2 junction conditions. In physics, often these behaviours are considered pathological and excluded by asking all the f\mbox{}ields to be regular. However, we cannot do this when performing a pillbox integration, since the physically relevant information is contained exactly in a behaviour of this type. We therefore need to clarify, according to our geometrical ansatz, which are the f\mbox{}ields whose divergences are physically meaningful, and which are the ones whose divergences are ``pathological''.

Once done that, we need to implement this information in the behaviour of the families $g^{_{(\r)}}_{_{AB}}$ and $\hvf_{^{\! (\r)}}^{_{A}}$. To understand how to do this, let's note that these pathological behaviours are compatible with the convergence to smooth limit conf\mbox{}igurations because the point-wise convergence is a quite weak condition, which does not link appeciably the regularity properties of the limit conf\mbox{}igurations (such as boundedness, for one) to those of the conf\mbox{}igurations belonging to the families. Stronger types of convergence, on the other hand, do provide such a link. Therefore, to make sure that the construction outlined in section \ref{Implementation of the thin limit} faithfully reproduces our ansatz, we will impose conditions also on the convergence properties of the families, and not only on the properties of the limit functions. This problem has been already highlighted in \cite{Sbisa:2014vva}, where it was shown that it can lead to important physical consequences. The subtlety of the pillbox integration across the ribbon brane is already witnessed by the dif\mbox{}ferent results reached by \cite{Gregory:2001xu, Gregory:2001dn} and \cite{Dvali:2006if} in the pure tension case.

\subsubsection{The convergence properties and the geometric ansatz}
\label{Convergence properties and the geometric ansatz}

In the framework described in section \ref{The geometric ansatz} to implement the thin limit, it is possible to look to the $\r \to 0^+$ limit from a more na\"ive perspective. Inserting the families $g^{_{(\r)}}_{_{AB}}$ and $\hvf_{^{\! (\r)}}^{_{A}}$ and $\hT^{_{(\r)}}_{ab}$ into the Israel junction conditions (\ref{junctionconditionseq}), we obtain a family of junction conditions where, thanks to (\ref{cristina rho}), the right hand side develops a ``Dirac delta'' divergence. It follows that, for the ansatz to work, na\"ively the families $g^{_{(\r)}}_{_{AB}}$ and $\hvf_{^{\! (\r)}}^{_{A}}$ have to be such that they produce a suitable divergence in the left hand side, which counterbalances the divergence due to the source term. That is, they have to produce a divergence localized at $\mcalC_2$ of the cod-1 extrinsic curvature and/or of the induced Einstein tensor, i.e.~of the geometry of the system.

The physical intuition behind our ansatz is that the divergence of the geometry at $\mcalC_2$ is not due to an (essential) singularity of the bulk geometry, and neither to how the cod-2 brane is embedded in the ambient space. Rather, it is due to how the cod-1 brane is embedded in the ambient space, in the direction orthogonal to $\mcalC_2$ (i.e.~to the ridge conf\mbox{}iguration). Therefore, in the context of our geometrical ansatz, when a disappearing divergence behaviour is associated to the bulk metric, we rule it out as pathological. Regarding the cod-1 embedding, we rule out such behaviour when taking place away from $\mcalC_2$, and also when they are associated to the partial derivatives in the directions parallel to $\mcalC_2$. Regarding the behaviour in the direction normal to $\mcalC_2$, the requirement of having a ridge and not a cusp at $\mcalC_2$ is related to the idea that the divergence of the geometry is due to the second normal derivative of the embedding $\de_{\t}^{2} \hvf_{^{(\r)}}^{_{A}}$, and not to the f\mbox{}irst normal derivative $\de_{\t} \hvf_{^{(\r)}}^{_{A}}$. Therefore, we rule out as pathological any disappearing divergence behaviour associated to $\de_{\t} \hvf_{^{(\r)}}^{_{A}}$ (and higher derivatives with respect to $\chd$ thereof).

A natural way to formalize this is to impose that the families $g^{_{(\r)}}_{_{AB}}$ and $\hvf_{^{\! (\r)}}^{_{A}}$ converge in suitable Banach spaces. Since the conf\mbox{}igurations belonging to the families are smooth by hypothesis, when we restrict their domain to a compact closure open set $\Th$ they belong authomatically to the spaces of functions of class $C^n$ on $\oTh$ for every $n \geq 0$, and of Lipschitzian functions on $\oTh$. As we mention in appendix \ref{appendix convergence}, it is possible to introduce natural norms in these spaces which make them Banach. Convergence in such norms constrains, on the one hand, the properties of the limit function, because of the completeness of Banach spaces. On the other hand, it constrains also the way the limit function is approached by the family, because of the specif\mbox{}ic form of the norms $\norm{\phantom{f}}_{C^{n}(\oTh)}$ and $\norm{\phantom{f}}_{Lip (\oTh)}$. For example, the convergence in the norm $\norm{\phantom{f}}_{Lip (\oTh)}$ not only guarantees that the limit function is Lipschitzian, but also that the partial derivatives of f\mbox{}irst order remain bounded, a condition which is clearly very useful to formalize our ansatz.

\subsubsection{The convergence properties of the families}
\label{Convergence properties of the families}

For the reasons mentioned above, we impose the following convergence conditions on the families to ensure that the construction faithfully reproduces our geometric ansatz.

Regarding the bulk metric, since we want to keep under control the behaviour of $g^{_{(\r)}}_{_{AB}}$ and its partial derivatives in $\oV_{\!\!p}$, we assume that the family of smooth functions $\{ g^{_{(\r)}}_{_{AB}} \}_{_{\! \r}}$ converges in the norm $\norm{\phantom{f}}_{C^{n} (\oV_{\!\!p})}$. In particular, this implies that the family $\{ g^{_{(\r)}}_{_{AB}} \}_{_{\! \r}}$ converges uniformly in $\oV_{\!\!p}$ to the limit conf\mbox{}iguration $g_{_{\! AB}}$, and that the family of any choice of partial derivatives of order $\leq n$ of $g^{_{(\r)}}_{_{AB}}$ converges uniformly in $\oV_{\!\!p}$ to the correspondent choice of partial derivatives of $g_{_{\! AB}}$.

Regarding the cod-1 embedding, among the partial derivatives of $\{ \hvf_{^{\! (\r)}}^{_{A}} \}_{_{\! \r}}$ of order $\leq 2$ we want only $\{ \de^{2}_{\t} \, \hvf_{^{\! (\r)}}^{_{A}} \}_{_{\! \r}}$ to be unbounded, and only inside the ribbon brane. Moreover, we want everything to be smooth with respect to the $\chd$ variables. We therefore impose the following conditions. First of all we impose that, indicating
\beq
Q_{q}^{_{(\r)}} = \bar{U}_{q} \times \big\{ \big( -l, -\r \big) \cup \big( \r , l \big) \big\} \quad ,
\eeq
the families $\{ \de^{\a} \hvf_{^{\! (\r)}}^{_{A}} \}_{_{\! \r}}$ with $\abs{\a} \leq n$ converge uniformly in $\{ \oQ_{q}^{_{(\r)}} \}_{_{\! \r}}$ to $\de^{\a} \hvf^{_{A}}$ (this condition is formalized precisely in appendix \ref{appendix convergence}), and that the limits
\begin{align}
\label{Elias}
\lim_{\r \to 0^+} &\de^{\a} \hvf_{^{\! (\r)}}^{_{A}}(\chd, -\r) & \lim_{\r \to 0^+} &\de^{\a} \hvf_{^{\! (\r)}}^{_{A}}(\chd, \r)
\end{align}
exist f\mbox{}inite for $\abs{\a} \leq n$ (with the convention that the partial derivative of order zero is the function itself).\footnote{In the the multi-index notation, we indicate $\a = (\a_1, \ldots, \a_d)$, where the $\a_i$ are non-negative real numbers, and def\mbox{}ine $\abs{\a} \equiv \a_1 + \ldots + \a_d$. Considering a function $f(x_{1}, \ldots , x_{d})$, we then indicate $\de^\a f = \frac{\de^{\abs{\a}} f}{\de x_{1}^{\a_1} \cdots \, \de x_{d}^{\a_d}}$.} These requirements keep under control the behaviour of $\hvf_{^{\! (\r)}}^{_{A}}$ and its partial derivatives outside of the ribbon brane, as well as assuring that the limit function $\hvf^{_{A}}$ is smooth for $\t \geq 0$ and for $\t \leq 0$.

As a second condition, we impose that the family $\hvf_{^{\! (\r)}}^{_{A}}$ converges in the norm
\beq
\label{Hitomi}
\norm{f}_{} = \norm{f}_{Lip (\oW_{\!\! q})} + \sum_{j = 1}^{n} \sum_{\abs{\a} = j} \norm{\de^{\a}_{\chd} f}_{Lip (\oW_{\!\! q})} \quad .\footnote{With $\de^{\a}_{\chd}$ we indicate the partial derivative with respect to the coordinates $\chd$ only, associated to the multi-index $\a$.}
\eeq
This implies, on the one hand, that the families $\{ \de_{\t} \hvf_{^{\! (\r)}}^{_{A}} \}_{_{\! \r}}$ and $\{ \de_{\t} \de^{\a}_{\chd} \hvf_{^{\! (\r)}}^{_{A}} \}_{_{\! \r}}$ with $\abs{\a} \leq n$ remain bounded in $\oW_{\!\! q}$ when $\r \to 0^+$, so the f\mbox{}irst partial derivatives with respect to $\t$ are under control. On the other hand, it implies that the families $\{ \hvf_{^{\! (\r)}}^{_{A}} \}_{_{\! \r}}$ and $\{ \de^{\a}_{\chd} \hvf_{^{\! (\r)}}^{_{A}} \}_{_{\! \r}}$ with $\abs{\a} \leq n$ converge uniformly in $\oW_{\!\! q}$ respectively to $\hvf^{_{A}}$ and $\de^{\a}_{\chd} \hvf^{_{A}}$, since convergence in the norm (\ref{Hitomi}) implies convergence in the norm (\ref{Anri}). This guarantees that the behaviour of $\{ \hvf_{^{\! (\r)}}^{_{A}} \}_{_{\! \r}}$ and $\{ \de^{\a}_{\chd} \hvf_{^{\! (\r)}}^{_{A}} \}_{_{\! \r}}$ with respect to the coordinates $\chd$ is kept under control, and that $\hvf^{_{A}}$ and $\de^{\a}_{\chd} \hvf^{_{A}}$ are continuous both in $\chd$ and in $\t$.

\subsubsection{Converge properties and limit conf\mbox{}igurations}

It is worth commenting on the requirement of the limits (\ref{Elias}) existing and being f\mbox{}inite. By the proposition \ref{proposition Qr} of appendix \ref{appendix convergence}, this implies that the limit function $\hvf^{_{A}}$ and all its partial derivatives of order $\leq n$ have a f\mbox{}inite limit for $\t \to 0^+$ and for $\t \to 0^-$ (although the two limit are in general dif\mbox{}ferent), since we have
\beq
\label{Alinete}
\lim_{\t \to 0^{\pm}} \de^{\a} \hvf^{_{A}}(\chd, \t) = \lim_{\r \to 0^+} \de^{\a} \hvf_{^{\! (\r)}}^{_{A}} (\chd, \pm \r) \quad .
\eeq
This justif\mbox{}ies the assertion that the existence and f\mbox{}initeness of the limits (\ref{Elias}) assures the limit function $\hvf^{_{A}}$ being smooth for $\t \geq 0$ and for $\t \leq 0 \, $.

Besides, note that the expressions $\de^{\a} \hvf_{^{\! (\r)}}^{_{A}}(\chd, -\r)$ and $\de^{\a} \hvf_{^{\! (\r)}}^{_{A}}(\chd, \r)$ indicate the value of $\hvf_{^{\! (\r)}}^{_{A}}$ and its partial derivatives on the sides of the ribbon brane. Therefore, the relation (\ref{Alinete}) implies that the thin limit conf\mbox{}iguration describes \emph{only the external} f\mbox{}ield conf\mbox{}iguration, as it should be in a thin limit description, and not part of the interior of the ribbon brane. This holds also for the normal vector, since the stronger convergence properties of the bulk metric imply
\begin{align}
\label{Bruna}
\lim_{\r \to 0^+} \de^{\a} g^{_{(\r)}}_{_{AB}}\big\vert_{\hvf_{_{\!\! (\r)}}^{\cdot}(\chd, \r)} &= \lim_{\r \to 0^+} \de^{\a} g^{_{(\r)}}_{_{AB}}\big\vert_{\hvf_{_{\!\! (\r)}}^{\cdot}(\chd, -\r)} = \de^{\a} g_{_{AB}}\big\vert_{\hvf^{\cdot}(\chd, 0)} \quad ,
\end{align}
from which it follows that
\begin{align}
\label{Sheila}
\mbf{n}_{_{+}}(\chd) &= \lim_{\r \to 0^+} \mbf{n}_{_{(\r)}}(\chd, \r) & \mbf{n}_{_{-}}(\chd) &= \lim_{\r \to 0^+} \mbf{n}_{_{(\r)}}(\chd, - \r) \quad ,
\end{align}
where $\mbf{n}_{_{+}}$ and $\mbf{n}_{_{-}}$ are def\mbox{}ined in (\ref{flauntit}).

A word about $n$. Ideally we would like to set $n = \infty$, to guarantee $g_{_{\! AB}}$ being of class $C^{\infty}$ and $\hvf^{_{A}}$ being of class $C^{\infty}$ with respect to $\chd$. However this is not possible for technical reasons, since the space of functions of class $C^{\infty}$ is not equipped with a natural norm. One option would be to def\mbox{}ine the convergence of the families in suitable Fr\'echet spaces, instead of in Banach spaces. For simplicity we decide to work with Banach spaces, and assume that $n$ is ``big enough'' so that the functions are suf\mbox{}f\mbox{}iciently regular for any need we may have. Therefore, when we speak about smoothness we actually intend it in the sense of ``of class $C^n$ with $n$ arbitrarily big''. With this precisation, it is clear that the convergence properties imposed in section \ref{Convergence properties of the families} are not only compatible to, but actually imply the properties of the limit conf\mbox{}igurations described in section \ref{The limit configuration}.

\section{The codimension-2 pillbox integration}
\label{section The codimension-2 pillbox integration}

We now turn to the derivation of the junction conditions associated to the thin limit of the ribbon brane. In cod-1 GNC, the Israel junction conditions (\ref{junctionconditionseq}) for the families of metrics and embeddings read
\begin{align}
- 2 \Msf \, \hg^{\la\r}_{^{(\r)}} \, \hK_{\la\r}^{_{(\r)}} + \Mft \, \hG_{\t\t}^{_{(\r)}} \, &= \,\, \hT_{\t\t}^{_{(\r)}} \label{junctionconditionseq tt} \\[4mm]
2 \Msf \, \hK_{\t\m}^{_{(\r)}} + \Mft \, \hR_{\t\m}^{_{(\r)}} \, &= \,\, \hT_{\t\m}^{_{(\r)}} \label{junctionconditionseq tm} \\[3mm]
2 \Msf \Big( \hK_{\m\n}^{_{(\r)}} - \hK^{_{(\r)}}_{\phantom{\m\n}} \, \hg_{\m\n}^{_{(\r)}} \Big) + \Mft \, \hG_{\m\n}^{_{(\r)}} \, &= \,\, \hT_{\m\n}^{_{(\r)}} \quad , \label{junctionconditionseq mn}
\end{align}
where we used the relations (\ref{c1 induced metric GNC trivial k app}) and $\hK^{_{(\r)}}_{\phantom{\m\n}} = \hK_{\t\t}^{_{(\r)}} + \hg^{\la\r}_{^{(\r)}} \hK_{\la\r}^{_{(\r)}}$. Regarding the left hand side of (\ref{junctionconditionseq tt})--(\ref{junctionconditionseq mn}), the results summarized in appendix \ref{appendix Einstein tensor extrinsic curvature} imply that the only terms which can diverge in the $\r \to 0^+$ limit are to be found in the $\m\n$ equation, i.e.~(\ref{junctionconditionseq mn}). This is hoped-for, since by (\ref{cristina rho}) the only components of the cod-1 energy-momentum tensor which give a non-vanishing contribution to the pillbox integration are the $\m\n$ ones. This suggests that we are following a consistent path, and conf\mbox{}irms the expectation that the junction conditions associated to the thin limit of the ribbon brane are to be found by performing the pillbox integration of (\ref{junctionconditionseq mn}), namely
\beq
\label{pillbox integration}
\lim_{\r \to 0^+} \, \int_{-\r}^{\r} d \t \,\, \bigg[ 2 \Msf \Big( \hK^{_{(\r)}}_{\m\n} -  \hg^{_{(\r)}}_{\m\n} \, \, \big( \hK^{_{(\r)}}_{\t\t} + \hg_{^{(\r)}}^{\m\n} \hK^{_{(\r)}}_{\m\n} \big) \Big) + \Mft \, \hG^{_{(\r)}}_{\m\n} \bigg] = \,\, \bar{T}_{\m\n} \quad .
\eeq

For ease of exposition, we consider separately the integration of the induced gravity part and of the extrinsic curvature part, so we def\mbox{}ine
\begin{align}
\mscrI^{G}_{\m\n}(\chd) &\equiv \lim_{\r \to 0^+} \, \int_{-\r}^{\r} d \t \,\, \hG^{_{(\r)}}_{\m\n} \label{ind grav part} \quad , \\[2mm]
\mscrI^{K}_{\m\n}(\chd) &\equiv \lim_{\r \to 0^+} \, \int_{-\r}^{\r} d \t \,\, \Big[ \hK^{_{(\r)}}_{\m\n} -  \hg^{_{(\r)}}_{\m\n} \, \, \big( \hK^{_{(\r)}}_{\t\t} + \hg_{_{(\r)}}^{\m\n} \hK^{_{(\r)}}_{\m\n} \big) \Big] \quad . \label{ext curv part}
\end{align}

\subsection{The induced gravity part}

To perform the pillbox integration (\ref{ind grav part}), it is useful to consider the $4+1$ splitting correspondent to the foliation of $\ti{\mcalC}_1$ where each leaf is a $\t$-constant hypersurface. From (\ref{c1 induced metric GNC trivial k app}), (\ref{c1 induced metric GNC trivial app}) and (\ref{c1 induced metric GNC app}) it is easy to see that $\hg_{\m\n}^{_{(\r)}}(\chd, \t)$ and $\hg_{\m\n}(\chd, \t)$, with $\t$ f\mbox{}ixed, are respectively the family and limit conf\mbox{}iguration of the metric induced on the 4D hypersurface labelled by $\t$.

As we mention in appendix \ref{appendix The induced metric}, $\hg_{\m\n}^{_{(\r)}}$ and $\hg_{\m\n}$ as functions of $\chd$ are of class $C^{n-1}$ in $\oW_{\!\!q}$, and their derivatives with respect to the $\chd$ coordinates change continuously with $\t$ and $\chd$. This implies that the 4D intrinsic geometry of each leaf is smooth, and changes continuously across $\t = 0$ also in the limit $\r \to 0^+$. On the other hand, the extrinsic curvature of each leaf is proportional to $\detau \hg_{\m\n}^{_{(\r)}}$, therefore tends in the limit $\r \to 0^+$ to a conf\mbox{}iguration which is discontinuous at $\t = 0$ and displays a f\mbox{}inite jump when passing from $\t = 0^-$ to $\t = 0^+$. Moreover, the conditions (\ref{Alinete}) and (\ref{Bruna}) imply
\begin{align}
\label{Patty}
\detau \, \hg_{\m\n} \big\vert_{\t = 0^+} &= \lim_{\r \to 0^+} \detau \, \hg_{\m\n}^{_{(\r)}} \big\vert_{\t = \r} \quad & \quad \detau \, \hg_{\m\n} \big\vert_{\t = 0^-} &= \lim_{\r \to 0^+} \detau \, \hg_{\m\n}^{_{(\r)}} \big\vert_{\t = - \r} \quad ,
\end{align}
where we indicated
\begin{align}
\detau \, \hg_{\m\n} \big\vert_{\t = 0^+} &= \lim_{\t \to 0^+} \detau \, \hg_{\m\n} \quad & \quad \detau \, \hg_{\m\n} \big\vert_{\t = 0^-} &= \lim_{\t \to 0^-} \detau \, \hg_{\m\n} \quad .
\end{align}
The situation is in fact completely analogous to that of a cod-1 brane in GR, and therefore the pillbox integration in (\ref{ind grav part}) can be performed exactly as in \cite{MisnerThorneWheeler} giving
\beq
\label{pillbox integral ind grav}
\mscrI^{G}_{\m\n} = \Big[ \bar{K}_{\m\n} - \bar{g}_{\m\n} \, \bar{K} \Big]_{\pm} \quad ,
\eeq
where $\bar{K} = \bar{g}^{\r\s} \, \bar{K}_{\r\s}$.

\subsection{The extrinsic curvature part}
\label{The extrinsic curvature part}

Regarding the pillbox integration (\ref{ext curv part}), the results mentioned in appendix \ref{appendix Einstein tensor extrinsic curvature} imply that the only unbounded term in the integrand is the one that contains $\hK_{\t\t}^{_{(\r)}}$, and that the only divergent term in $\hK_{\t\t}^{_{(\r)}}$ is $n^{_{\! (\r)}}_{_{L}} \, \de_{\t}^2 \, \hvf^{_{L}}_{^{\! (\r)}}$. Since in the limit $\r \to 0^+$ all the bounded terms give a vanishing contribution, because the measure of the integration interval shrinks ot zero, we get
\beq
\mscrI^{K}_{\m\n}(\chd) = - \lim_{\r \to 0^+} \, \int_{-\r}^{\r} \hg^{_{(\r)}}_{\m\n} \, \Big( n^{_{\! (\r)}}_{_{L}} \, \de_{\t}^2 \hvf^{_{L}}_{^{\! (\r)}} \Big) \, d \t \quad . \label{ext curv part 1}
\eeq

This expression can be further simplif\mbox{}ied by noting that only the term in round parenthesis diverges when $\r \to 0^+$, while $\hg^{_{(\r)}}_{\m\n}$ remain bounded and converges uniformly to $\hg_{\m\n}$. As we prove in appendix \ref{appendix integration}, if the following limits exist and are f\mbox{}inite
\begin{align}
&\lim_{\r \to 0^+} \int_{-\r}^{\r} n_{^{L}}^{_{\! (\r)}} \, \de_{\t}^{2} \hvf^{_{L}}_{^{\! (\r)}} \,\, d\t & &\lim_{\r \to 0^+} \int_{-\r}^{\r} \babs{n_{^{L}}^{_{\! (\r)}} \, \de_{\t}^{2} \hvf^{_{L}}_{^{\! (\r)}}} \,\, d\t \quad , \label{absolute value}
\end{align}
then $\mscrI^{K}_{\m\n}(\chd)$ exists and is f\mbox{}inite, and we have
\begin{equation}
\label{pillbox integral extr curv 1}
\mscrI^{K}_{\m\n}(\chd) = - \, \bar{g}_{\m\n} (\chd) \, \mscrI (\chd) \quad ,
\end{equation}
where
\begin{equation}
\label{integral}
\mscrI (\chd) = \lim_{\r \to 0^+} \int_{-\r}^{\r} n_{^{L}}^{_{\! (\r)}} \, \de_{\t}^{2} \hvf^{_{L}}_{^{\! (\r)}} \,\, d\t \quad.
\end{equation}
We prove below the existence and f\mbox{}initeness of the pillbox integral (\ref{integral}), postponing to section \ref{section Discussion} the discussion on the existence and f\mbox{}initeness of the second pillbox integral of (\ref{absolute value}). Note in passing that, as a consistency check, $\mscrI (\chd)$ remains invariant under the ref\mbox{}lection $\t \to -\t$, as it should be since the right hand side of (\ref{pillbox integration}) is not inf\mbox{}luenced by the transformation.

\subsubsection{The auxiliary vector f\mbox{}ields}
\label{The auxiliary vector fields}

The evaluation of the pillbox integral $\mscrI$ is greatly facilitated by the introduction of two suitable families of auxiliary vector f\mbox{}ields. As a f\mbox{}irst step in this direction, let's introduce the families of vector f\mbox{}ields def\mbox{}ined on $\oW_{\!\! q}$
\begin{align}
\textsc{n}_{^{\! (\r)}}^{_{A}}(\chd, \t) &= \de_{\t} \,\hvf^{_{A}}_{^{\! (\r)}} & u_{^{[\m]}}^{_{\! (\r) A}}(\chd, \t) &= \de_{\m} \,\hvf^{_{A}}_{^{\! (\r)}} \quad ,
\end{align}
where $\dem \equiv \de_{\ch^{\m}}$. They are the push-forward with respect to the family of cod-1 embeddings of the vector basis associated to the cod-1 GNC. It follows that they are tangent to the family of cod-1 branes, and moreover that $\mbfscn_{_{\! (\r)}}$ is orthogonal to the set of linearly independent vector f\mbox{}ields $\{ \mbf{u}_{^{[\m]}}^{_{(\r)}} \}_{\m = 0, \ldots , 3}$. So, for every $(\chd, \t)$, the vectors $\mbf{n}_{_{(\r)}}$ and $\mbfscn_{_{\! (\r)}}$ are orthonormal, and are both orthogonal to the set $\{ \mbf{u}_{^{[\m]}}^{_{(\r)}} \}_{\m}$. Analogously, we def\mbox{}ine the vector f\mbox{}ields
\begin{align}
\textsc{n}^{_{\!A}}(\chd, \t) &= \de_{\t} \,\hvf^{_{A}} & u_{^{[\m]}}^{_{A}}(\chd, \t) &= \de_{\m} \,\hvf^{_{A}} \quad , \label{def N and u}
\end{align}
which are respectively def\mbox{}ined on $\oW_{\!\! q}\setminus \{ \t = 0 \}$ and $\oW_{\!\! q}$. For every $(\chd, \t)$ with $\t \neq 0$, $\mbfscn$ and $\{ \mbf{u}_{_{[\m]}} \}_{\m = 0, \ldots , 3}$ are tangent to the limit conf\mbox{}iguration of the cod-1 brane; moreover, the vectors $\mbf{n}$ and $\mbfscn$ are orthonormal, and are both orthogonal to the set $\{ \mbf{u}_{_{[\m]}} \}_{\m}$. The integral (\ref{integral}) can be therefore conveniently expressed as
\begin{equation}
\label{integral N}
\mscrI = \lim_{\r \to 0^+} \int_{-\r}^{\r} n_{^{L}}^{_{\! (\r)}} \, \de_{\t} \textsc{n}_{^{\! (\r)}}^{_{L}} \,\, d\t \quad,
\end{equation}
where dependence on $\chd$ is understood in the left hand side and dependence on $(\chd, \t)$ is understood inside the integral. 
Since the integration is executed at $\chd$ constant, we can f\mbox{}ix $\chd$ and consider the integrand as a function of $\t$ only.

The evaluation of the pillbox integral (\ref{integral N}) is dif\mbox{}f\mbox{}icult because both $n_{^{L}}^{_{\! (\r)}}$ and $\de_{\t} \textsc{n}_{^{\! (\r)}}^{_{L}}$ have a non-trivial behavior at $\t = 0$ in the $\r \to 0^+$ limit. Ideally, we would like to express the integrand in a way such that a unique object embodies the non-trivial behaviour of the geometry at $\t = 0 \,$. Indeed, such an object would be the only one surviving in the pillbox integration. To pursue this goal, it is crucial to notice that, despite the families of smooth vector f\mbox{}ields $n_{^{(\r)}}^{_{L}}$ and $\textsc{n}_{^{\! (\r)}}^{_{L}}$ have a non-trivial behaviour at $\t = 0$ when $\r \to 0^+$, the vector f\mbox{}ields $u_{^{[\m]}}^{_{\! (\r) A}}$ do not. In fact, by our ansatz, the families of smooth functions $u_{^{[\m]}}^{_{\! (\r) A}}$ converge uniformly to continuous conf\mbox{}igurations $u_{^{[\m]}}^{_{A}}$, and the partial derivatives $\de_{\t} u_{^{[\m]}}^{_{\! (\r) A}}$ remain bounded in a neighbourhood of $\t = 0$ when $\r \to 0^+$. Intuitively, this implies that we can complete the set $\{ u_{^{[\m]}}^{_{\! (\r) A}} \}_{_{\m}}$ to a basis by adding two vector f\mbox{}ields $y_{^{(\r)}}^{_{A}}$ and $z_{^{(\r)}}^{_{A}}$ which share the same convergence properties of the $u_{^{[\m]}}^{_{\! (\r) A}}$.

Indeed, we show in appendix \ref{appendix auxiliary fields} that it is possible to f\mbox{}ind two families of orthonormal \emph{auxiliary vector f\mbox{}ields} $y_{^{(\r)}}^{_{A}}$ and $z_{^{(\r)}}^{_{A}}$ which are orthogonal to the vector f\mbox{}ields $u_{^{[\m]}}^{_{\! (\r) A}}$, and which ``behave well'' in $\t = 0$ when $\r \to 0^+$, with a proviso to be discussed shortly. This is not unexpected, since (unlike $n_{^{(\r)}}^{_{L}}$ and $\textsc{n}_{^{\! (\r)}}^{_{L}}$) the auxiliary vector f\mbox{}ields are not forced to be rigidly tangent or normal to the codimension-1 brane. By ``behave well'', we mean that the families of smooth vector f\mbox{}ields $y_{^{(\r)}}^{_{A}}$ and $z_{^{(\r)}}^{_{A}}$ converge uniformly to \emph{continuous} conf\mbox{}igurations $y^{_{A}}$ and $z^{_{A}}$ which are again orthonormal, and that the partial derivatives $\de_{\t} y_{^{(\r)}}^{_{A}}$ and $\de_{\t} z_{^{(\r)}}^{_{A}}$ remain bounded in a neighbourhood of $\t = 0$ when $\r \to 0^+$. In particular, this implies that
\begin{align}
\label{petrushka}
\lim_{\r \to 0^+} y_{_{\! (\r)}}^{_{A}} \big( \pm \r \big) &= y^{_{A}} (0) & \lim_{\r \to 0^+} z_{_{\! (\r)}}^{_{A}} \big( \pm \r \big) &= z^{_{A}} (0) \quad .
\end{align}
The proviso mentioned above relates to the fact that the f\mbox{}ields $y_{^{(\r)}}^{_{A}}(\t)$, $z_{^{(\r)}}^{_{A}}(\t)$, $y^{_{A}}(\t)$ and $z^{_{A}}(\t)$ may not be def\mbox{}ined on $[-l, l \, ]$ and for $\r \in (0, l \, ] \,$, but only on $[-d, d \, ]$ and for $\r \in (0, \bd \, ] \,$, with $0 \! < \! d \! \leq \! l$ and $\bd \! \leq \! l \,$. This however is not a problem, because we use the construction of the auxiliary vector f\mbox{}ields only to evaluate the \emph{limit} (\ref{integral N}), and for $\r$ small enough the integration interval $\big[ \! -\r , \r \, \big]$ is surely contained in $[-d, d \, ] \,$.

\subsubsection{The slope function}
\label{The slope function}

The idea now is to expand the f\mbox{}ields $\big( \mbfscn_{_{\! (\r)}} , \mbf{n}_{_{(\r)}} \big)$ on the basis provided by $\big( \mbf{y}_{_{\! (\r)}} , \mbf{z}_{_{(\r)}} \big)$, and expand the f\mbox{}ields $\big( \mbfscn , \mbf{n} \big)$ on the basis provided by $\big( \mbf{y} , \mbf{z} \big)$. In fact, since the couples $\big( \mbfscn_{_{\! (\r)}} , \mbf{n}_{_{(\r)}} \big)$ and $\big( \mbf{y}_{_{\!\! (\r)}} , \mbf{z}_{_{(\r)}} \big)$ are both orthogonal to the set of vectors $\{ \mbf{u}_{^{[\m]}}^{_{\! (\r)}} \}_{_{\m}}$, they have to be linearly dependent for every $\t$. The same is true of the couples $\big( \mbfscn , \mbf{n} \big)$ and $\big( \mbf{y} , \mbf{z} \big)$, since they are both orthogonal to the set of vectors $\{ \mbf{u}_{_{[\m]}} \}_{_{\m}}$ for every $\t \neq 0 \,$.

Since the two couples of vector f\mbox{}ields $\big( \mbfscn_{_{\! (\r)}} , \mbf{n}_{_{(\r)}} \big)$ and $\big( \mbf{y}_{_{\!\! (\r)}} , \mbf{z}_{_{(\r)}} \big)$ are both orthonormal, they have to be linked by $O(2)$ matrices, which can be chosen to have positive determinant (in case by swapping $\mbf{y}_{_{\!\! (\r)}}$ and $\mbf{z}_{_{(\r)}}$).\footnote{In other words, we can choose the orientation of $\big( \mbf{y}_{_{\!\! (\r)}} , \mbf{z}_{_{(\r)}} \big)$ to be the same of that of $\big( \mbfscn_{_{\! (\r)}} , \mbf{n}_{_{(\r)}} \big) \,$.} It follows that we can write
\begin{align}
\mbfscn_{_{\! (\r)}}(\t) &= \cos \big(S_{_{\! (\r)}} (\t)\big) \,\, \mbf{y}_{_{\!\! (\r)}}(\t) + \sin \big(S_{_{\! (\r)}} (\t)\big) \,\, \mbf{z}_{_{(\r)}}(\t) \label{auxiliary 1} \\[3mm]
\mbf{n}_{_{(\r)}}(\t) &= - \sin \big(S_{_{\! (\r)}} (\t)\big) \, \mbf{y}_{_{\!\! (\r)}}(\t) + \cos \big(S_{_{\! (\r)}} (\t)\big) \, \mbf{z}_{_{(\r)}}(\t) \label{auxiliary 2}
\end{align}
for some family of functions $\big\{ S_{_{\! (\r)}} \big\}_{_{\!\! \r}} : [-d , d \,] \to \mbbR \,$, and likewise for the limit conf\mbox{}igurations we can write
\begin{align}
\mbfscn(\t) &= \cos \big(S (\t)\big) \,\, \mbf{y}(\t) + \sin \big(S (\t)\big) \,\, \mbf{z}(\t) \label{auxiliary 1 lim} \\[3mm]
\mbf{n}(\t) &= - \sin \big(S (\t)\big) \, \mbf{y}(\t) + \cos \big(S (\t)\big) \, \mbf{z}(\t) \quad , \label{auxiliary 2 lim}
\end{align}
for some function $S : [-d , d \,] \setminus \{0\} \to \mbbR \,$. In analogy to \cite{Sbisa:2014gwh}, we call $S_{_{\! (\r)}}$ and $S$ respectively the family and the limit value of the \emph{slope function}. This is motivated by the observation that $\tan S_{_{\! (\r)}}$ is the slope of the embedding function as a function of $\t$ when represented in the system of axis $\big(  \mbf{y}_{_{\!\! (\r)}}, \mbf{z}_{_{(\r)}} \big)$, since $\textsc{n}^{_{\! A}}_{_{\! (\r)}} = \de_{\t} \, \hvf^{_{A}}_{_{\! (\r)}}$ (the same is true for $\tan S$ in the system of axis $\big( \mbf{y}, \mbf{z} \big)$, since $\textsc{n}^{_{\! A}} = \de_{\t} \hvf^{_{A}}$).

Strictly speaking, the relations (\ref{auxiliary 1})--(\ref{auxiliary 2 lim}) do not f\mbox{}ix the functions $S_{_{\! (\r)}}$ and $S$ nor their properties, because of the periodicity of the sine and of the cosine functions. However, as we show in appendix \ref{appendix slope function}, we can always choose the functions $S_{_{\! (\r)}}$ to be smooth on $[-d , d \, ]$, and the function $S$ to be smooth on $[-d , d \, ] \setminus \{ 0 \}$. Moreover, exploiting the requirement of absence of self-intersections of the cod-1 brane, we can choose them in such a way that the family $\big\{ S_{_{\! (\r)}} \big\}_{_{\!\! \r}}$ converges to $S$ uniformly in $\{ [-d , -\r \,] \cup [\, \r , d \,] \}_{_{\! \r}}\,$. This implies that the following relations hold
\begin{align}
\label{Setiba}
\lim_{\r \to 0^+} S_{_{\! (\r)}} (-\r) &= S_{_{\!-}} & \lim_{\r \to 0^+} S_{_{\! (\r)}} (\r) &= S_{_{\!+}} \quad ,
\end{align}
where the limits
\begin{align}
\label{Mona Lisa}
S_{_{\!-}} &= \lim_{\t \to 0^-} S \big( \t \big) & S_{_{\!+}} &= \lim_{\t \to 0^+} S \big( \t \big)
\end{align}
exist and are f\mbox{}inite. The conditions of smoothness and uniform convergence in $\{ [-d , -\r \,] \cup [\, \r , d \,] \}_{_{\! \r}}$ drastically reduces the arbitrarity in the choice of $\big\{ S_{_{\! (\r)}} \big\}_{_{\!\! \r}}$ and $S \,$, and, regarding the function $S$, leave undetermined only an overall additive phase (i.e.~independent of $\t$) which have no relevance at all in the discussion to follow.

Note that $S_{_{\! (\r)}}$ and $S$ inherit from $\de_{\t} \, \hvf^{_{A}}_{^{\! (\r)}}$ and $\de_{\t} \, \hvf^{_{A}}$ a non-trivial behaviour at $\t = 0\,$. In fact, $S$ has to be discontinuous in $\t = 0$ to satisfy (\ref{auxiliary 1 lim}) and (\ref{auxiliary 2 lim}), and $\de_{\t} S_{_{\! (\r)}}$ cannot remain bounded in a neighbourhood of $\t = 0$ to satisfy (\ref{auxiliary 1}) and (\ref{auxiliary 2}). On the other hand, $\{ \mbfscn_{_{\! (\r)}} \}_{_{\! \r}}$, $\{ \mbf{n}_{_{(\r)}} \}_{_{\! \r}}$, $\mbfscn$ and $\mbf{n}$ are completely determined in terms of $S_{_{\! (\r)}}$ and $S$ and of the auxiliary vector f\mbox{}ields, where the latters are regular in $\t = 0$. Therefore the non-trivial behaviour of the geometry in $\t = 0$ is \emph{completely} encoded in the (family and limit conf\mbox{}iguration of the) slope function. We conclude that the introduction of the auxiliary vector f\mbox{}ields indeed permits to express the non-trivial behaviour of $\mbfscn_{_{\! (\r)}}$ and $\mbfscn$ in terms of the non-trivial behaviour of a unique scalar object, the angular coordinate of $\mbfscn_{_{\! (\r)}}$ and $\mbfscn$ in the frame provided by the auxiliary f\mbox{}ields.

\subsection{Evaluation of the pillbox integral}

The ability to express the families $\mbfscn_{_{\! (\r)}}$ and $\mbf{n}_{_{(\r)}}$ in terms of $\mbf{y}_{_{\!\! (\r)}}$, $\mbf{z}_{_{(\r)}}$ and $S_{_{\! (\r)}}$ permits to evaluate the pillbox integral (\ref{integral N}) in a very simple way. In fact, using (\ref{auxiliary 1})--(\ref{auxiliary 2}) to express the integrand of (\ref{integral N}), we get
\begin{multline}
n_{^{L}}^{_{\! (\r)}} \, \de_{\t} \textsc{n}_{^{\! (\r)}}^{_{L}} = \de_{\t} S_{_{\! (\r)}} - \sin S_{_{\! (\r)}} \cos S_{_{\! (\r)}} \, y_{_{L}}^{_{(\r)}} \de_{\t} y_{^{(\r)}}^{_{L}} - \sin^{2}\! S_{_{\! (\r)}} \, y_{_{L}}^{_{(\r)}} \de_{\t} z_{^{(\r)}}^{_{L}} + \\
+ \cos^{2}\! S_{_{\! (\r)}} \, z_{_{L}}^{_{(\r)}} \de_{\t} y_{^{(\r)}}^{_{L}} + \sin S_{_{\! (\r)}} \cos S_{_{\! (\r)}} \, z_{_{L}}^{_{(\r)}} \de_{\t} z_{^{(\r)}}^{_{L}} \label{Heydouga}
\end{multline}
and it is easy to recognize that the only diverging term on the right hand side is $\de_{\t} S_{_{\! (\r)}}$. Plugging (\ref{Heydouga}) into (\ref{integral N}), and using (\ref{Setiba}), we get
\begin{equation}
\label{integral S}
\mscrI = \lim_{\r \to 0^+} \Big[ S_{_{\! (\r)}} (\r) - S_{_{\! (\r)}} (-\r) \Big] = S_{_{\!+}} - S_{_{\!-}} \quad , 
\end{equation}
which is independent of an arbitrary phase common to $S_{_{\!+}}$ and $S_{_{\!-}}$, as we mentioned above. As we explain in the appendix \ref{appendix slope function}, the requirement of absence of self-intersections of the cod-1 brane place the constraint $\abs{S_{_{\!+}} - S_{_{\!-}}} < \pi$ on the possible values of $S_{_{\!+}}$ and $S_{_{\!-}}$.

\subsubsection{Independence from the auxiliary f\mbox{}ields}

The expression (\ref{integral S}) for the pillbox integral $\mscrI$, as it stands, seems to depend on the specif\mbox{}ic choice of the auxiliary vector f\mbox{}ields, which is highly non-unique. We now show that this is not the case. For ease of notation, let's indicate with $\mscrR \big( \theta \big)$ the matrix which executes a 2D rotation of an angle $\theta$ (with the convention that positive angles correspond to rotations in the counterclockwise direction).

Consider a dif\mbox{}ferent choice of auxiliary f\mbox{}ields $\mbf{y}_{^{\! (\r)}}^{\p}$ and $\mbf{z}_{^{(\r)}}^{\p}$ which converge to continuous conf\mbox{}igurations $\mbf{y}^{\p}$ and $\mbf{z}^{\p}$, and such that $\big( \mbf{y}_{^{\! (\r)}}^{\p} , \mbf{z}_{^{(\r)}}^{\p} \big)$ have the same orientation of $\big( \mbfscn_{_{\! (\r)}} , \mbf{n}_{_{(\r)}} \big)$ (respectively, $\big( \mbf{y}^{\p} , \mbf{z}^{\p} \big)$ have the same orientation of $\big( \mbfscn , \mbf{n} \big)$).\footnote{To stave of\mbox{}f possible confusion, we remark that in this section a prime does \emph{not} indicate a derivative.} Then relations completely analogous to (\ref{auxiliary 1})--(\ref{auxiliary 2 lim}) hold, namely
\begin{equation}
\label{Raiane}
\begin{pmatrix}
\mbfscn_{_{\! (\r)}}(\t) \\
\mbf{n}_{_{(\r)}}(\t)
\end{pmatrix}
= \mscrR \Big( S^{\, \p}_{_{\! (\r)}} (\t) \Big)
\begin{pmatrix}
\mbf{y}_{_{\!\! (\r)}}^{\p}(\t) \\
\mbf{z}_{_{\!(\r)}}^{\p}(\t)
\end{pmatrix}
\end{equation}
and
\begin{equation}
\begin{pmatrix}
\mbfscn(\t) \\
\mbf{n}(\t)
\end{pmatrix}
= \mscrR \Big( S^{\, \p} (\t) \Big)
\begin{pmatrix}
\mbf{y}^{\p}(\t) \\
\mbf{z}^{\p}(\t)
\end{pmatrix} \quad ,
\end{equation}
where $S_{^{\! (\r)}}^{\, \p}$ and $S^{\, \p}$ are respectively the family and the limit value of the new slope function. Analogously to (\ref{petrushka}) and (\ref{Setiba}) we have
\begin{align}
\label{petrushka p}
\lim_{\r \to 0^+} \mbf{y}_{_{\!\! (\r)}}^{\p} \big( I_{_{\!(\r)}} \big) &= \mbf{y}^{\p} \big( 0 \big) & \lim_{\r \to 0^+} \mbf{z}_{_{\!(\r)}}^{\p} \big( I_{_{\!(\r)}} \big) &= \mbf{z}^{\p} \big( 0 \big)
\end{align}
and
\begin{align}
\label{Setiba p}
\lim_{\r \to 0^+} S_{_{\! (\r)}}^{\, \p} (\r) &= S^{\, \p}_{_{\!+}} \phantom{iii} & \phantom{iiii} \lim_{\r \to 0^+} S^{\, \p}_{_{\! (\r)}} (-\r) &= S^{\, \p}_{_{\!-}} \quad ,
\end{align}
where $S^{\, \p}_{_{\!+}} = \lim_{\t \to 0^+} S^{\, \p} \big( \t \big)$ and $S^{\, \p}_{_{\!-}} = \lim_{\t \to 0^-} S^{\, \p} \big( \t \big)$. To evaluate the integral (\ref{integral N}), we can follow the same procedure as above using the ``primed'' auxiliary f\mbox{}ields, getting
\begin{equation}
\label{integral Sp}
\mscrI = S_{_{\!+}}^{\, \p} - S_{_{\!-}}^{\, \p}
\end{equation}
with $\abs{S_{_{\!+}}^{\, \p} - S_{_{\!-}}^{\, \p}} < \pi \,$.

Now, recall that the couples $\big( \mbf{y}_{_{\! (\r)}} , \mbf{z}_{_{(\r)}} \big)$ and $\big( \mbf{y}_{_{\!\! (\r)}}^{\p} , \mbf{z}_{_{(\r)}}^{\p} \big)$ are orthonormal by construction and have concordant orientation (since they both have concordant orientation with $\big( \mbfscn_{_{\! (\r)}} , \mbf{n}_{_{(\r)}} \big)$). Therefore they are linked by $SO(2)$ matrices, and explicitly we can write
\begin{equation}
\label{Cintia}
\begin{pmatrix}
\mbf{y}_{_{\! (\r)}}(\t) \\
\mbf{z}_{_{(\r)}}(\t)
\end{pmatrix}
= \mscrR \Big( A_{_{(\r)}} (\t) \Big)
\begin{pmatrix}
\mbf{y}_{_{\!\! (\r)}}^{\p}(\t) \\
\mbf{z}_{_{(\r)}}^{\p}(\t)
\end{pmatrix}
\end{equation}
for some family of functions $\{ A_{_{(\r)}} (\t) \}_{_{\! \r}}$. Likewise, the couples $\big( \mbf{y} , \mbf{z} \big)$ and $\big( \mbf{y}^{\p} , \mbf{z}^{\p} \big)$ are linked by a $SO(2)$ matrix, and explicitly
\begin{equation}
\label{Cintia limit}
\begin{pmatrix}
\mbf{y}(\t) \\
\mbf{z}(\t)
\end{pmatrix}
= \mscrR \Big( A(\t) \Big)
\begin{pmatrix}
\mbf{y}^{\p}(\t) \\
\mbf{z}^{\p}(\t)
\end{pmatrix}
\end{equation}
for some function $A(\t) \,$. We refer to $\{ A_{_{(\r)}} \}_{_{\! \r}}$ and $A$ as the family and limit conf\mbox{}iguration of the \emph{transition function}. Inserting (\ref{Cintia}) into (\ref{auxiliary 1})--(\ref{auxiliary 2}) and comparing with (\ref{Raiane}), we f\mbox{}ind
\begin{equation}
\label{Amanda}
\mscrR \Big( S_{_{\! (\r)}} (\t) \Big) \circ \mscrR \Big( A_{_{(\r)}} (\t) \Big) = \mscrR \Big( S^{\, \p}_{_{\! (\r)}} (\t) \Big) \quad ,
\end{equation}
and analogously
\begin{equation}
\mscrR \Big( S(\t) \Big) \circ \mscrR \Big( A(\t) \Big) = \mscrR \Big( S^{\, \p}(\t) \Big) \quad .
\end{equation}

The value and properties of the transition functions are not f\mbox{}ixed by the relations (\ref{Cintia}) and (\ref{Cintia limit}), but, exactly as for the slope functions, we can choose them to be smooth. Moreover, following the procedure described in appendix \ref{appendix slope function} for the slope functions, we can impose the family $\{ A_{_{(\r)}} \}_{_{\! \r}}$ to converge to $A$ uniformly in $\{ [-d , -\r \,] \cup [\, \r , d \,] \}_{_{\! \r}}\,$. Since the limit conf\mbox{}iguration of the auxiliary f\mbox{}ields (dif\mbox{}ferently from $\mbfscn$) is continuous in $\t = 0 \,$, we have that
\begin{equation}
\lim_{\t \to 0^-} A (\t) = \lim_{\t \to 0^+} A(\t)
\end{equation}
which implies
\begin{equation}
\label{Amandinha}
\lim_{\r \to 0^+} A_{_{(\r)}} (-\r) = \lim_{\r \to 0^+} A_{_{(\r)}} (\r) \quad .
\end{equation}
On the other hand, from (\ref{Amanda}) it follows that
\begin{equation}
\label{Lucelia}
S_{_{\! (\r)}} (\t) + A_{_{(\r)}} (\t) = S^{\, \p}_{_{\! (\r)}} (\t) + 2 \pi \, m_{_{(\r)}}(\t)
\end{equation}
where $m_{_{(\r)}}(\t)$ is a $\mbbZ \,$-valued family of functions, so using (\ref{Amandinha}) we get
\begin{equation}
S_{_{\!+}} - S_{_{\!-}} = S_{_{\!+}}^{\,\p} - S_{_{\!-}}^{\,\p} + 2 \pi \lim_{\r \to 0^+} \Big[ m_{_{(\r)}} (\r) - m_{_{(\r)}} (-\r) \Big] \quad .
\end{equation}
Since the term in square parenthesis is an integer, the conditions $\abs{S_{_{\!+}} - S_{_{\!-}}} < \pi$ and $\abs{S_{_{\!+}}^{\, \p} - S_{_{\!-}}^{\, \p}} < \pi$ imply
\begin{equation}
\lim_{\r \to 0^+} \Big[ m_{_{(\r)}} (\r) - m_{_{(\r)}} (-\r) \Big] = 0 \quad ,
\end{equation}
and therefore
\begin{equation}
S_{_{\!+}} - S_{_{\!-}} = S_{_{\!+}}^{\,\p} - S_{_{\!-}}^{\,\p} \quad .
\end{equation}
It follows that, evaluating the integral (\ref{integral N}) using the ``unprimed'' and ``primed'' auxiliary f\mbox{}ields, we get the same answer.

\section{The thin limit of the Cascading DGP model}
\label{section Thin limit of the Cascading DGP model}

The results of the previous section permit to evaluate exactly the left hand side of (\ref{pillbox integration}), and therefore derive the cod-2 junction conditions. Before turning to the the latter conditions and to the thin limit of the Cascading DGP model, we characterize geometrically the expression (\ref{integral S}) for the pillbox integral $\mscrI$.

\subsection{The geometrical meaning of the pillbox integral}
\label{The geometrical meaning of the pillbox integral}

Despite the expression
\beq
\mscrI = S_{_{\!+}} - S_{_{\!-}} \label{integral S pop}
\eeq
is very simple, it is not always useful in practice. When dealing with the codimension-2 junction conditions, we would like to avoid going each time through the construction of the auxiliary f\mbox{}ields, so we would like to express $\mscrI$ in terms of quantities readily calculable from the bulk metric and the embedding. Moreover, we would like to understand the geometrical meaning of the expression (\ref{integral S pop}). We address these points below.

\subsubsection{Covariant expression for the pillbox integral}
\label{Covariant expression}

To link $S_{_{\!+}} - S_{_{\!-}}$ to readily calculable quantities, let's note that, taking into account that $\mbf{y}$ and $\mbf{z}$ are continuous, the relations (\ref{auxiliary 1 lim})--(\ref{auxiliary 2 lim}) imply
\begin{align}
\mbfscn_{_{\pm}} &= \cos S_{_{\!\pm}} \,\, \mbf{y}(0) + \sin S_{_{\!\pm}} \,\, \mbf{z}(0) \label{auxiliary N} \\[2mm]
\mbf{n}_{_{\pm}} &= - \sin S_{_{\!\pm}} \,\, \mbf{y}(0) + \cos S_{_{\!\pm}} \,\, \mbf{z}(0) \quad , \label{auxiliary n}
\end{align}
where $\mbfscn_{_{\pm}}$ and $\mbf{n}_{_{\pm}}$ are def\mbox{}ined in (\ref{flauntit}). Let's consider the equations (\ref{auxiliary N})--(\ref{auxiliary n}) characterized by the sign ``$-$'': we can invert them to express $\mbf{y}(0)$ and $\mbf{z}(0)$ in terms of $\sin S_{_{\!-}}$, $\cos S_{_{\!-}}$, $\mbfscn_{_{-}}$ and $\mbf{n}_{_{-}}$. Plugging the relations thus obtained into the equations (\ref{auxiliary N})--(\ref{auxiliary n}) characterized by the sign ``+'', we obtain
\begin{align}
\mbfscn_{_{+}} &= \cos \big( S_{_{\!+}} - S_{_{\!-}} \big) \,\, \mbfscn_{_{\!-}} + \sin \big( S_{_{\!+}} - S_{_{\!-}} \big) \,\, \mbf{n}_{_{-}} \label{Jump 1} \\[2mm]
\mbf{n}_{_{+}} &= - \sin \big( S_{_{\!+}} - S_{_{\!-}} \big) \,\, \mbfscn_{_{\!-}} + \cos \big( S_{_{\!+}} - S_{_{\!-}} \big) \,\, \mbf{n}_{_{-}} \label{Jump 2} \quad ,
\end{align}
and, using the orthonormality of $\big(\mbfscn_{_{+}}, \mbf{n}_{_{+}} \big)$ and $\big(\mbfscn_{_{\!-}}, \mbf{n}_{_{-}} \big)$ we get
\begin{align}
\cos \big( S_{_{\!+}} - S_{_{\!-}} \big) &= \big\langle \mbfscn_{_{+}}, \mbfscn_{_{\!-}} \big\rangle  = \big\langle \mbf{n}_{_{+}}, \mbf{n}_{_{-}} \big\rangle  \label{Debora 1} \\[2mm]
\sin \big( S_{_{\!+}} - S_{_{\!-}} \big) &= \big\langle \mbfscn_{_{+}}, \mbf{n}_{_{-}} \big\rangle  = - \, \big\langle \mbf{n}_{_{+}}, \mbfscn_{_{\!-}} \big\rangle  \label{Debora 2} \quad ,
\end{align}
where we used the notation
\beq
\big\langle \mbf{a}, \mbf{b} \big\rangle  = g_{_{AB}}\big\rvert_{\vfd(0)} \,\, a^{_{A}} \,\, b^{_{B}} \quad .
\eeq

We now want to express $S_{_{\!+}} - S_{_{\!-}}$ itself, and not its sine and cosine, in terms of the scalar products of $\mbfscn_{_{+}}$, $\mbfscn_{_{-}}$, $\mbf{n}_{_{+}}$ and $\mbf{n}_{_{-}}$. As is well known, the cosine function is not globally invertible, and to invert it it is necessary to restrict its domain, there being inf\mbox{}inite possible ways to do that. Since, as we already mentioned, the absence of self-intersections of the cod-1 brane implies $\abs{S_{_{\!+}} - S_{_{\!-}}} < \pi$, we choose to perform the inversion on the interval $[0,\pi]$ and def\mbox{}ine
\beq
\arccos \, : \,\, [-1, 1] \to [0 , \pi] \quad .
\eeq
It then follows from (\ref{Debora 1}) and from the property of the cosine function of being even that
\beq
S_{_{\!+}} - S_{_{\!-}} = \sgn \big( S_{_{\!+}} - S_{_{\!-}} \big) \,\, \arccos \Big( \big\langle \mbf{n}_{_{+}}, \mbf{n}_{_{-}} \big\rangle \Big)
\eeq
where $\sgn$ indicates the sign function.\footnote{The sign function is def\mbox{}ined as $\sgn(x) = 1$ for $x > 0$, $\sgn(x) = 0$ for $x = 0$ and $\sgn(x) = -1$ for $x < 0$.} On the other hand, since $\sgn (S_{_{\!+}} - S_{_{\!-}}) = \sgn \big( \sin (S_{_{\!+}} - S_{_{\!-}}) \big)\,$, using (\ref{Debora 2}) we conclude that
\beq
\label{Luiza}
S_{_{\!+}} - S_{_{\!-}} = - \, \sgn \Big( \big\langle \mbfscn_{_{\!-}}, \mbf{n}_{_{+}} \big\rangle \Big) \,\, \arccos \Big( \big\langle \mbf{n}_{_{+}}, \mbf{n}_{_{-}} \big\rangle \Big) \quad .
\eeq
This expression is completely satisfactory, since it is covariant and built from objects directly related to the embedding and the bulk metric.

Of course there are three other equivalent ways to express (\ref{Luiza}), since, using (\ref{Debora 1}) and (\ref{Debora 2}), we can express $\sin ( S_{_{\!+}} - S_{_{\!-}})$ and $\cos (S_{_{\!+}} - S_{_{\!-}})$ in two dif\mbox{}ferent ways each. As a consistency check, it is easy to verify that the right hand side of (\ref{Luiza}) remains invariant under the ref\mbox{}lection $\t \to -\t$, since in this case the normal vectors change as follows
\begin{align}
\label{Sonia}
\mbf{n}_{_{+}} &\to \mbf{n}_{_{-}} & \mbf{n}_{_{-}} &\to \mbf{n}_{_{+}} & \mbfscn_{_{+}} &\to - \mbfscn_{_{\!-}} & \mbfscn_{_{\!-}} &\to - \mbfscn_{_{+}}
\end{align}
and by (\ref{Debora 2}) we have $\big\langle \mbfscn_{_{+}}, \mbf{n}_{_{-}} \big\rangle = - \, \big\langle \mbf{n}_{_{+}}, \mbfscn_{_{\!-}} \big\rangle$.

\subsubsection{The relation with the ridge angles}

Reinstating now the dependence on the coordinates $\chd$, the relation (\ref{Luiza}) implies that we can express $\mscrI(\chd)$ as follows
\beq
\label{Ana Luiza}
\mscrI(\chd) = - \, \sgn \Big( \big\langle \mbfscn_{_{\!-}}, \mbf{n}_{_{+}} \big\rangle (\chd) \Big) \,\, \arccos \Big( \big\langle \mbf{n}_{_{+}}, \mbf{n}_{_{-}} \big\rangle (\chd) \Big) \quad ,
\eeq
which implies that the pillbox integral $\mscrI(\chd)$ coincides with the oriented dihedral angle of the ridge $\O(\chd)$ (see relation \ref{The Final Countdown}). It is very pleasing that the extrinsic curvature part of the pillbox integration turns out to be directly related to the geometrical quantity suited to describe an oriented ridge conf\mbox{}iguration.

The interpretation of $\mscrI(\chd)$ as a closing/def\mbox{}icit angle, along the lines of appendix \ref{appendix ridge}, is however subtler. In fact, for a generic thin source conf\mbox{}iguration the bulk is not f\mbox{}lat and the cod-1 brane is not made up of two f\mbox{}lat semi-hyperplanes; equivalently, the manifolds $\mscrM$, $\mcalM_{_{+}}$ and $\mcalM$ are not globally factorizable into a 4D Lorentzian part and a Riemannian 2D part. It follows that the closing and def\mbox{}icit angles have no global meaning. However, they can be given a local meaning, as we see below. Let's take $q \in \mcalC_2$ and, referring to the discussion of section \ref{The thin limit of the codimension-1 brane} about the def\mbox{}initions of the manifolds $\mscrM$, $\mcalM_{_{+}}$ and $\mcalM \,$, decompose the tangent space $T_{q} \mscrM = T_{q} \, \mcalC_2 \oplus \mcal{E}_{q}$ in the direct sum of the subspace tangent to $\mcalC_2$ and its orthogonal complement $\mcal{E}_{q} = (T_{q} \, \mcalC_2)^{\perp}$. As it is well-known, the exponential map $\exp_{q} : \mcal{F} \subset T_q \mscrM \to \mscrM$ is a local dif\mbox{}feomorphism between a neighbourhood $\mcal{F}$ of the zero tangent vector and a neighbourhood $\mscr{F} \subset \mscrM$ of $q$. We def\mbox{}ine the \emph{local normal section} of $\mscrM$ at $q$ as the set $\mscr{N}_{q} \subset \mscrM$ such that $\mscr{N}_{q} = \exp_{q}(\mcal{E}_{q} \cap \mcal{F})$, where ``normal'' is intended with respect to $\mcalC_2$ and ``local'' means that it is def\mbox{}ined only in a neighbourhood of $q$. Since $\mcalM_{_{+}}$ is identif\mbox{}ied with a subset of $\mscrM$, we def\mbox{}ine the local normal section of $\mcalM_{_{+}}$ at $q$ as the set $\mscr{N}_{q}^{_{+}} = \mscr{N}_{q} \cap \mcalM_{_{+}}$. It follows that the collection $\{ \mscr{N}_{q} \}_{q}$ foliates $\mscrM$ in a neighbourhood of $\mcalC_2$, and $\{ \mscr{N}_{q}^{_{+}} \}_{q}$ foliates $\mcalM_{_{+}}$ in a neighbourhood of $\mcalC_2 \,$.

To characterize geometrically $\mscr{N}_{q}^{_{+}}$, we f\mbox{}irst note that the set of vectors of $\mcal{E}_{q}$ which correspond to curves which from $q$ point into the manifold $\mcalM_{_{+}}$ is the slice $\mcal{W}_{q} \subset \mcal{E}_{q}$ which have $\mbfscn_{_{\!-}}$ and $\mbfscn_{_{+}}$ as borders. The set $\exp_{q}(\mcal{W}_{q} \cap \mcal{F})$ does not coincide in general with $\mscr{N}_{q}^{_{+}}$, since the border of the latter is not necessarily made of geodesics of $\mscrM$. However, the dif\mbox{}ference between $\exp_{q}(\mcal{W}_{q} \cap \mcal{F})$ and $\mscr{N}_{q}^{_{+}}$ becomes more and more negligible as we get closer to $q$, and so we can identify them as long as we are concerned only with their properties at $q$. On the other hand, since the dif\mbox{}ferential of the exponential map evaluated in $0$ is the identity \cite{DoCarmoBook}
\beq
d_{0} \exp_{q} = id \quad ,
\eeq
we can in turn identify $\exp_{q}(\mcal{W}_{q} \, \cap \, \mcal{F})$ with $\mcal{W}_{q} \, \cap \, \mcal{F}$, as far as we consider only their properties at $q$. We conclude that, although $\mscr{N}_{q}^{_{+}}$ is not a slice of the Euclidean plane as in the pure tension case, at $q$ we can indeed approximate it with the slice $\mcal{W}_{q} \,$, which is an oriented (Euclidean) 2D ridge. Using the identif\mbox{}ication $\mscr{N}_{q}^{_{+}} \cong \mcal{W}_{q}$ at $q\,$, we can then def\mbox{}ine the \emph{local closing angle} of $\mcalM_{_{+}}$ at $q$ as the closing angle of $\mscr{N}_{q}^{_{+}} \cong \mcal{W}_{q} \,$. It takes a moment's thought to realize that the oriented dihedral angle of $\mcal{W}_{q}$ coincides with the oriented dihedral angle $\O(\chd)$ of the 6D ridge, evaluated at the coordinate $\chd$ which corresponds to $q\,$. We conclude that we can indeed interpretate the pillbox integral $\mscrI(\chd)$ (and $\O(\chd)$) as the local closing angle of $\mcalM_{_{+}}$ at $q \,$.

Regarding the $\mbbZ_2$-symmetric manifold $\mcalM \,$, as we explained in section \ref{The thin limit of the codimension-1 brane}, it is constructed by creating a copy of $\mcalM_{_{+}}$ and glueing together the two mirror copies in a $\mbbZ_2$-symmetric way with respect to the boundary (the cod-1 brane). It follows from the discussion above that $\mcalM$ is foliated by $\{ \mscr{N}_{q}^{_{+}} \leftrightarrow \mscr{N}_{q}^{_{+}} \}_{q}$, where $\leftrightarrow$ indicates the $\mbbZ_2$-glueing operation, so we def\mbox{}ine the local normal section of $\mcalM$ at q as the set $\mscr{N}_{q}^{_{+}} \leftrightarrow \mscr{N}_{q}^{_{+}}$. Moreover, it follows that $\mscr{N}_{q}^{_{+}} \leftrightarrow \mscr{N}_{q}^{_{+}}$ can be approximated at $q$ with the subset $\mcal{W}_{q} \leftrightarrow \mcal{W}_{q}$ of the two-dimensional Euclidean plane. Therefore, in analogy with the analysis above, we def\mbox{}ine the \emph{local def\mbox{}icit angle} of $\mcalM$ at $q$ as the def\mbox{}icit angle associated to $\mscr{N}_{q}^{_{+}} \cong \mcal{W}_{q} \leftrightarrow \mscr{N}_{q}^{_{+}} \cong \mcal{W}_{q} \,$, and we indicate it with $\D(\chd)$. It is easy to realize that the local def\mbox{}icit angle of $\mcalM$ at $q$ is equal to twice the local closing angle of $\mcalM_{_{+}}$ at $q$, so we have
\beq
\label{Oooohh}
\D(\chd) = 2 \, \O(\chd) \quad .
\eeq

In conclusion, we can give the pillbox integral $\mscrI (\chd)$ two geometrical interpretations, as the local closing angle of $\mcalM_{_{+}}$ or as (half of) the local def\mbox{}icit angle of $\mcalM$
\beq
\label{Glauzia}
\mscrI (\chd) = \O(\chd) = \half \,\, \D(\chd) \quad .
\eeq

\subsection{The cod-2 junction conditions and the Cascading DGP model}
\label{The cod-2 junction conditions and the Cascading DGP model}

Taking into account the expressions (\ref{pillbox integral ind grav}), (\ref{pillbox integral extr curv 1}) and (\ref{Glauzia}) for the pillbox integrals (\ref{ind grav part}) and (\ref{ext curv part}), where $\O(\chd)$ is explicitly given in (\ref{The Final Countdown}), we can express the pillbox integration (\ref{pillbox integration}) of the Israel junction conditions across the ribbon brane as

\begin{equation}
\label{Ana Luiza 2}
- 2 \Msf \, \O \,\, \bar{g}_{\m\n} + \Mft \, \Big[ \bar{K}_{\m\n} - \bar{g}_{\m\n} \, \bar{K} \Big]_{\!\pm} = \,\, \bar{T}_{\m\n}
\end{equation}
or equivalently as

\begin{equation}
\label{Ana Luiza 3}
- \Msf \, \D \,\, \bar{g}_{\m\n} + \Mft \, \Big[ \bar{K}_{\m\n} - \bar{g}_{\m\n} \, \bar{K} \Big]_{\!\pm} = \,\, \bar{T}_{\m\n} \quad ,
\end{equation}
where dependence on the coordinates $\chd$ is understood. These equations are, for a nested branes with induced gravity set-up, the codimension-2 analogous of the Israel junction conditions.

Several comments are in order. The f\mbox{}irst is that the details of the internal structure of the ribbon brane, which are encoded in the behaviour of the families $\hT_{ab}^{_{(\r)}}$, $\hvf_{^{(\r)}}^{_{A}}$ and $g_{_{AB}}^{_{(\r)}}\vert_{\hvfd_{^{(\r)}}}$ for $\t \in [-\r , \r \,]$, do not inf\mbox{}luence how the limit external f\mbox{}ield conf\mbox{}iguration $g_{_{\!AB}}$, $\vf^{_{A}}$ is sourced by the thin source conf\mbox{}iguration $\bar{T}_{\m\n} \,$. In other words, in our implementation of the thin limit the details of the internal structure of the ribbon brane disappear when $\r \to 0^+$. This is a necessary condition for the thin limit of the ribbon brane to be well-def\mbox{}ined. Mathematically, this follows from the fact that the divergent parts of the integrands in (\ref{ind grav part}) and (\ref{integral}) can be written as total derivatives. Therefore, in the $\r \to 0^+$ limit the  pillbox integrals pick up only the value of the (divergent part of the) integrands at the boundary of the ribbon brane. While this is expected for the induced gravity part, since we know this holds for shell conf\mbox{}igurations in GR, the result is non-trivial for the extrinsic curvature part, where the quantity which appears is the total derivative of the slope function.

The second comment is that all the geometrical quantities which appear in the cod-2 junction conditions are well-def\mbox{}ined for a ridge conf\mbox{}iguration, and are sourced by the pillbox integration of the energy-momentum tensor across the ribbon brane. Therefore our geometrical ansatz (see section \ref{The geometric ansatz}) is consistent. This point is tightly linked to the previous one, since it follows from the fact that the quantities which appear as total derivatives in the pillbox integrals are well-def\mbox{}ined for ridge conf\mbox{}igurations. The third comment is that, despite we used a preferred class of gauges to implement the thin limit (we adopted a ``bulk-based'' approach), the junction conditions we arrived at are perfectly covariant, since they are made of tensorial objects on $\mcalC_2 \,$. In particular $\O$ (or equivalently $\D$) are scalars from the point of view of $\mcalC_2 \,$, and they are def\mbox{}ined (see (\ref{The Final Countdown})--(\ref{Oooohh})) in a coordinate-free way in terms of 6D tensorial objects. Therefore, the approach outlined in section \ref{Implementation of the thin limit} to implement the thin limit is, a posteriori, justif\mbox{}ied.

Regarding the discussion at the beginning of section \ref{section The thin limit of the ribbon brane}, to claim that we derived a thin limit description for our system we should prove that the thin source conf\mbox{}iguration f\mbox{}ixes uniquely the external f\mbox{}ield conf\mbox{}iguration. In other words, we should prove that the codimension-2 junction conditions, supplemented by the source-free Israel junction conditions and the bulk Einstein equations
\begin{align}
\label{Bulkeq 2}
\mbf{G} =& \,\, 0 \qquad \,\, (\textrm{bulk})\\[3mm]
\label{junctionconditionseq 2}
2 \Msf \Big( \tmbf{K} -  \tmbf{g} \, \, tr \big( \tmbf{K} \big) \Big) + \Mft \, \tmbf{G} =& \,\, 0 \qquad (\textrm{cod-1 brane}) \\[2mm]
\label{cod2 jceq}
- \Msf \, \D \,\, \bar{\mbf{g}} + \Mft \, \Big[ \bar{\mbf{K}} - \bar{\mbf{g}} \,\, tr \big( \bar{\mbf{K}} \big) \Big]_{\!\pm} =& \,\, \bar{\mbf{T}} \qquad (\textrm{cod-2 brane})
\end{align}
admit a unique solutions once suitable boundary conditions are specif\mbox{}ied. This point is rather technical, and we defer its discussion to section \ref{section Discussion}.

\subsubsection{The 6D Cascading DGP model}

To formulate the thin limit description of the 6D Cascading DGP model, we proceed in two steps. First of all, we notice that in the junction conditions (\ref{Ana Luiza 2})--(\ref{Ana Luiza 3}) the codimension-2 brane is, f\mbox{}inally, thin. This means that we can introduce a genuine 4D induced gravity term on the codimension-2 brane exactly as done in \cite{Dvali:2000hr}, which in practice amounts to perform the substitution $\bar{T}_{\m\n} \to \bar{T}_{\m\n} - \Mfs \, \bar{G}_{\m\n} \,$. We therefore obtain the following codimension-2 junction conditions
\begin{equation}
- \Msf \, \D \,\, \bar{g}_{\m\n} + \Mft \, \Big[ \bar{K}_{\m\n} - \bar{g}_{\m\n} \, \bar{K} \Big]_{\!\pm} + \Mfs \, \bar{G}_{\m\n} = \,\, \bar{T}_{\m\n}
\end{equation}
or the equivalent conditions where $\D$ is replaced by $2 \, \O \,$.

Secondly, we introduce a mirror symmetry with respect to the codimension-2 brane. Considering $\mcalC_2$ as a submanifold of $\mcalC_1$, the absence of self-intersections implies that $\mcalC_2$ divides $\mcalC_1$ into three disjoint pieces $\mcalC_1 = \mcalC_{1}^{_{+}} \cup \mcalC_2 \cup \mcalC_{1}^{_{-}}$, where $\mcalC_{1}^{_{+}}$ and $\mcalC_{1}^{_{-}}$ are open 5D submanifolds of $(\mcalC_1, \ti{\mcalA})$ which have in common the boundary $\mcalC_2 \,$. We assume that $\mcalC_{1}^{_{+}}$ and $\mcalC_{1}^{_{-}}$ are dif\mbox{}feomorphic, and that they are $\mbbZ_2$ symmetric with respect to $\mcalC_2 \,$. In the cod-1 GNC, this implies that the $\t\t$ and $\m\n$ components of $\tmbf{g}$, $\tmbf{G}$, $\tmbf{K}$ and $\bar{\mbf{K}}$ are even with respect to the transformation $\t \to -\t$, while the the $\t\m$ components are odd. In particular, the codimension-2 junction conditions can be equivalently written as
\begin{equation}
- \Msf \, \D \,\, \bar{g}_{\m\n} + 2 \Mft \, \Big[ \bar{K}_{\m\n} - \bar{g}_{\m\n} \, \bar{K} \Big]_{\!+} + \Mfs \, \bar{G}_{\m\n} = \,\, \bar{T}_{\m\n} \quad .
\end{equation}
Let us furthermore stress that, following \cite{Sbisa:2014gwh, Sbisa:2014vva} and dif\mbox{}ferently from \cite{deRham:2007rw, deRham:2007xp, deRham:2010rw}, we do \emph{not} impose the mirror symmetry with respect to $\mcalC_2$ to hold into the bulk, but only on the codimension-1 brane (in other words, we don't impose the manifold $\mcalM$ to satisfy a double mirror symmetry $\mbbZ_2 \times \mbbZ_2$).

Therefore, we obtain the following system of equations for the thin limit of the 6D Cascading DGP model
\begin{align}
\label{Bulkeq Cascading}
\mbf{G} =& \,\, 0 \qquad \,\, (\textrm{bulk})\\[3mm]
\label{junctionconditionseq Cascading}
2 \Msf \Big( \tmbf{K} -  \tmbf{g} \, \, tr \big( \tmbf{K} \big) \Big) + \Mft \, \tmbf{G} =& \,\, 0 \qquad (\textrm{cod-1 brane}) \\[2mm]
\label{cod2 jceq Cascading}
- \Msf \, \D \,\, \bar{\mbf{g}} + 2 \Mft \, \Big[ \bar{\mbf{K}} - \bar{\mbf{g}} \,\, tr \big( \bar{\mbf{K}} \big) \Big]_{\!+} + \Mfs \, \bar{\mbf{G}} =& \,\, \bar{\mbf{T}} \qquad (\textrm{cod-2 brane}) \quad .
\end{align}
Also in this case, we postpone to section \ref{section Discussion} the discussion on the existence and uniqueness of solutions. It is encouraging that the pure tension solutions of \cite{Sbisa:2014gwh, deRham:2010rw, Sbisa:2014vva} and the perturbative solutions of \cite{Sbisa:2014gwh, Sbisa:2014vva}, as well as the exact solution of \cite{Minamitsuji:2008fz}, can be obtained from the system (\ref{Bulkeq Cascading})--(\ref{cod2 jceq Cascading}) (as it is straightforward to check).

Before turning to the general discussion, let's brief\mbox{}ly comment on the covariant conservation of these equations. The covariant divergence of the equation (\ref{Bulkeq Cascading}) gives the condition $\nabla_{_{\!\! B}} G_{^{A}}^{_{\,\,\, B}} = 0$ which is satisf\mbox{}ied identically, i.e.~for any choice of the bulk metric (independently if it satisf\mbox{}ies the equations of motion or not). This is what happens in GR in fact. However, this is \emph{not} the case for the equations (\ref{junctionconditionseq Cascading}) and (\ref{cod2 jceq Cascading}). Taking the covariant divergence of the equation (\ref{junctionconditionseq Cascading}), for example, one gets
\beq
\nabla_{\! b} \, \ti{K}_{a}^{\,\,\, b} - \nabla_{\! a} \, \ti{K} = 0
\eeq
where $\ti{K} = \ti{g}^{rs} \, \ti{K}_{rs}$, which does not vanish for a generic choice of the bulk metric and of the cod-1 embedding. In fact, by the Gauss-Codazzi relation (see e.g.~\cite{MisnerThorneWheeler, DoCarmoBook}) one gets
\beq
\nabla_{\! b} \, \ti{K}_{a}^{\,\,\, b} - \nabla_{\! a} \, \ti{K} = - g^{rs} \, R^{n}_{\,\,\,\, ras}
\eeq
where $R$ is the bulk Riemann tensor and $n$ indicates the normal component. Dif\mbox{}ferently from GR, the vanishing of the covariant divergence of (\ref{junctionconditionseq Cascading}) is just a \emph{consequence} of the equation of motion (\ref{junctionconditionseq Cascading}) being satisf\mbox{}ied, in other words it holds only on-shell. A similar thing happens for the equation (\ref{cod2 jceq Cascading}). Taking its covariant divergence we get
\beq
\label{cachaca}
\Msf \,\, \bar{\nabla}_{\!\! \m} \, \D = 2 \Mft \, \Big[ \bar{\nabla}_{\! \n} \, \bar{K}_{\m}^{\,\,\, \n} - \bar{\nabla}_{\! \m} \, \bar{K} \, \Big]_{\!+} \quad ,
\eeq
which is not identically satisf\mbox{}ied but is again enforced by the equation (\ref{cod2 jceq Cascading}). The fact of $\bar{\nabla}_{\! \n} \, \bar{K}_{\m}^{\,\,\, \n} - \bar{\nabla}_{\! \m} \, \bar{K}$ not being identically zero is actually welcome, because it allows the local def\mbox{}icit angle to change in space and time (otherwise the condition (\ref{cachaca}) would imply $\bar{\nabla}_{\!\! \m} \, \D = 0$, and therefore $\D$ constant).

\section{Discussion}
\label{section Discussion}

The main result of our paper is the system of equations (\ref{Bulkeq Cascading})--(\ref{cod2 jceq Cascading}), which def\mbox{}ines the \emph{nested branes realization of the 6D Cascading DGP model}. Regarding the discussion of section \ref{section Introduction},
it is apparent that the cod-1 and cod-2 branes are equipped with induced gravity terms of the correct dimensionality (respectively, 5D and 4D), as we desired.

To arrive at this system, we started from a conf\mbox{}iguration where the two branes are both thick, and assumed that a hierarchy holds between their thicknesses. Roughly speaking, we assumed that the cod-1 brane is much thinner that the cod-2 brane. This allowed to take the thin limit of the conf\mbox{}iguration, so to say, in two steps. First we considered the approximate description where the cod-1 brane is thin and the cod-2 brane turns into a thick ribbon brane. At this point we consistently introduced a 5D induced gravity term on the cod-1 brane. Then we considered the thin limit of the ribbon brane inside the thin cod-1 brane equipped with induced gravity. This limit turned out to give consistent junction conditions at the thin codimension-2 brane, and permitted to equip the latter brane with a truly 4D induced gravity term.

The systems (\ref{Bulkeq 2})--(\ref{cod2 jceq}) and (\ref{Bulkeq Cascading})--(\ref{cod2 jceq Cascading}) consistently generalize the equations of \cite{Sbisa:2014gwh} and \cite{Sbisa:2014vva} at a background-independent and fully non-perturbative level, in the following sense. Indicating with $g_{_{\! AB}}^{_{\, 0}}\,$, $\vf_{_{0}}^{_{A}}$ and $\ti{\a}_{^{0}}^{a}$ respectively the bulk metric, cod-1 embedding function and cod-2 embedding function of a pure tension conf\mbox{}iguration, we def\mbox{}ine the perturbation f\mbox{}ields
\beq
h_{_{\! AB}} = g_{_{\! AB}} - g_{_{\! AB}}^{_{\, 0}} \qquad \qquad \d \vf^{_{A}} = \vf^{_{A}} - \vf_{_{0}}^{_{A}} \qquad \qquad \d \ti{\a}^{a} = \ti{\a}^{a} - \ti{\a}_{^{0}}^{a} \quad , \label{perturbations}
\eeq
where $\d \vf^{_{A}}\,$ and $\d \ti{\a}^{a}$ are usually referred to as the bending modes of the branes.\footnote{Note that in \cite{Sbisa:2014gwh, Sbisa:2014vva} the background conf\mbox{}iguration f\mbox{}ields are indicated with an overbar, and not with a ``$0$''. We changed notation here because in the present paper the overbar is already used to denote quantities pertaining to the cod-2 brane.} Inserting the def\mbox{}initions (\ref{perturbations}) into (\ref{Bulkeq 2})--(\ref{cod2 jceq}), and imposing the cod-1 GNC, we obtain an expansion of the equations of motion in terms of powers of $h_{_{\! AB}}$ and $\d \vf^{_{A}}$ (and their derivatives). The linear part of this expansion, upon performing the 4D scalar-vector-tensor decomposition of $h_{_{\! AB}}$ and $\d \vf^{_{A}}$ and introducing the master variables, exactly reproduces the results of \cite{Sbisa:2014gwh}. The same holds, starting now from the system (\ref{Bulkeq Cascading})--(\ref{cod2 jceq Cascading}), for the master equations of \cite{Sbisa:2014vva}.

\subsection{On the physical and mathematical assumptions}

It is useful at this stage to reiterate the physical and mathematical assumptions that were made in deriving the thin limit equations.

\paragraph{Thickness hierarchy} The main physical assumption is the hierarchy $l_{2}^{\perp} \sim l_{1}$ and $l_{2}^{\shortparallel} \gg l_{1}$ between the branes thicknesses (see section \ref{The physical set-up}). This choice, which is unrelated to an eventual hierarchy between $m_5$ and $m_6\,$, def\mbox{}ines the class of thick conf\mbox{}igurations for which our procedure to derive the thin limit is valid. This does \emph{not} mean that it is impossible to def\mbox{}ine the thin limit for the thick conf\mbox{}igurations which, although belonging to the cascading branes/nested topological defects category, fall outside this condition. It may be that other classes of conf\mbox{}igurations admit a thin limit description, which may coincide or not with the one we found. A thorough assessment of the admissibility of a thin limit description for a generic thick cascading branes model is, however, outside the scope of this work. Already within the simplifying assumption of the thickness hierarchy, deriving rigorously the thin limit description is challenging.

\paragraph{Gaussian Normal Coordinates} Even before proposing the thickness hierarchy, we made the assumption that the thick cod-1 and cod-2 branes can be completely covered by the appropriate Gaussian Normal Coordinates. This technical assumption is by no means exotic, since it is implicitly used also in the pure codimension-1 case to derive the Israel junction conditions. Without it, it is even hard to def\mbox{}ine the notion of thickness of a physical brane, and it is impossible to introduce the thin source conf\mbox{}iguration as the integral of the localized stress-energy tensor, nor perform pillbox integrations. On the other hand, there are well-known physical situations where such an assumption would be ill-based, for example in the case of collapses and formation of singularities. Since the use of the GNC is not peculiar to the cascading scenarios but is equally relevant for pure codimension-1 branes, let us for a moment consider the case of a single brane.

In that case, the rationale behind this assumption is the following. A thin limit description is an ef\mbox{}fective description valid at length scales much bigger than the thickness of a brane. The assumption of GNC covering completely the brane somehow corresponds to assuming that the thickness is the only relevant characteristic scale of the physical brane. There may however exist (``extreme'') internal conf\mbox{}igurations where additional characteristic scales (such, e.g., the Schwarzschild radius of a black hole inside the brane) emerge in association to the internal dynamics, and which cannot be meaningfully described just by the integrated stress-energy tensor. From this point of view the breakdown of the GNC signals the appearance of these additional characteristic scales in the problem. In these cases, since we are dealing with a multi-scale system, the formulation of the thin limit description needs to be tailored to each specif\mbox{}ic type of conf\mbox{}igurations, and may not follow the general lines which work when the thickness is the only relevant scale of the brane. The assumption of the validity of the Gaussian Normal Coordinates therefore can be interpreted as selecting the internal conf\mbox{}igurations where a general procedure to develop the thin limit decription can be used.

In the case of cascading branes/nested topological defects, analogous considerations apply. In this vein, assuming that the branes are covered by the appropriate GNC corresponds to assuming that the thicknesses are the only relevant scales of the branes. In presence of additional characteristic length scales associated to the internal structures, the development of a thin limit description would have to pass through a discussion of the relation between all the relevant scales, and it is quite likely that even assuming the hierarchy $l_{2}^{\perp} \sim l_{1}$ and $l_{2}^{\shortparallel} \gg l_{1}$ it would not be possible to f\mbox{}ind a unique thin limit description.

Let us remark that the existence of such ``extreme'' internal conf\mbox{}igurations does not mean that the thin limit is ill-def\mbox{}ined in the sense of \cite{Geroch:1987qn}. In fact, proving the existence of the thin limit always needs some sort of regularity conditions, without which the limit never exists. The point of \cite{Geroch:1987qn} is that there exist systems where, even within the assumption of regularity of the internal source conf\mbox{}iguration, it is not possible to f\mbox{}ind a thin limit description. For example, in the case of a pure codimension-2 brane it is not possible to f\mbox{}ind a def\mbox{}inition of ``thin source conf\mbox{}iguration'' such that the latter determines uniquely the external f\mbox{}ield conf\mbox{}iguration, regardless of the regularity assumptions. It is these systems to which one refers to when saying that the thin limit is not well-def\mbox{}ined.

\paragraph{Ridges and convergence properties} In section \ref{section The thin limit of the ribbon brane} we put forward the ansatz that the thin limit conf\mbox{}igurations of our system are characterized by having a ridge at the cod-2 brane, and def\mbox{}ined the convergence properties which the sequence of embeddings and metrics have to enjoy to faithfully respect this geometrical intuition. It is important to clarify that this does not imply any assumption about the properties of the ribbon brane, but just def\mbox{}ines a working hypothesis in order to look for a thin limit description. In fact, when focusing on scales much bigger than its thickness, we can always approximate a thick ribbon conf\mbox{}iguration with a ridge conf\mbox{}iguration, and, likewise, we can always approximate the source conf\mbox{}iguration with a localised one. This is just a consequence of the ribbon having a f\mbox{}inite thickness, and has nothing to do with its internal structure. From the point of view of the physical branes, in other words, it is just a consequence of the thickness hierarchy.

What is \emph{not} in general true is that, when focusing on scales much bigger than the thickness, the equations of motion necessarily become relations which involve exclusively the geometrical quantities which characterise the ridge conf\mbox{}igurations ($\D$, $\bar{\mbf{g}}$, $\bar{\mbf{K}}$, $\bar{\mbf{G}}$) and the localized source conf\mbox{}iguration ($\bar{\mbf{T}}$). The fact that this, as we showed, indeed happens in the Cascading DGP model, strongly suggests that the thin limit is well-def\mbox{}ined. If it didn't happen it would mean that, despite we can approximate the \emph{conf\mbox{}igurations} of our system with sources perfectly localized on the crest of a ridge, following these lines it is not possible to approximate the \emph{model}, i.e.~the way source conf\mbox{}igurations and external f\mbox{}ield conf\mbox{}igurations are linked. To establish this, it is necessary to implement the idea of ``focusing on scales much bigger than the thickness of the ribbon brane'' at the level of the equations of motion. The subtleties discussed in \cite{Sbisa:2014vva} imply that implementing the idea of the system tending to a ridge conf\mbox{}iguration is very subtle, and must involve a careful choice of the convergence properties of the sequences. But again, the latter choice is related to problem of f\mbox{}inding the correct mathematical implementation of the thin limit procedure, and has nothing to do with the internal structure of the ribbon brane.

\paragraph{Cusp conf\mbox{}igurations} A word is in order about the cusp conf\mbox{}igurations. In section \ref{Ridge configurations}, along with introducing the ridge conf\mbox{}igurations, we mentioned that a conf\mbox{}iguration has a cusp when the dihedral angle is $\Upsilon = \pi \,$. The cusp conf\mbox{}igurations were not mentioned ever after. The reason for this is that from (\ref{Sofia Loren}), (\ref{The Final Countdown}) and (\ref{Oooohh}) it follows that a cusp conf\mbox{}iguration corresponds to a def\mbox{}icit angle $\abs{\D} = 2 \pi\,$. In such a case the cod-1 brane becomes degenerate, since the $\t < 0$ and $\t > 0$ sides of the cod-1 brane coincide, and develops a boundary at the cod-2 brane. Furthermore, if $\D = 2 \pi$ the bulk has vanishing volume and coincides with the cod-1 brane (it is, in fact, 5-dimensional). These conf\mbox{}igurations are therefore pathological, and cannot be considered acceptable thin conf\mbox{}igurations.

\paragraph{Lebesgue integrability} The last (technical) assumption we made is that the pillbox integral on the right in (\ref{absolute value}) exists f\mbox{}inite. The pillbox integral in question may diverge only if the embedding function oscillates unbounded an inf\mbox{}inite number of times between $\t = -\r$ and $\t = \r$, in such a way that the integrand is not Lebesgue integrable. Although it is possible to imagine conf\mbox{}igurations where this happens, we believe that such a behaviour is rather innatural for most physical systems, and we deem such conf\mbox{}igurations as unphysical.

\subsubsection{Comments}

It is worthwhile to discuss how restrictive are our assumptions, since it is clear that a thin limit description is the more useful the larger is the class of system which it can describe. From this point of view, it is remarkable that the thickness hierarchy assumption is concerned only with the overall \emph{shape} of the branes, while the only assumption we make about the internal structures and the conf\mbox{}ining mechanisms is that the GNC hold. Therefore, we believe that our thin limit equations describe a fairly rich class of systems. We note in passing that also the papers \cite{deRham:2010rw, Moyassari:2011nb}, although without stating it explicitly, work in the framework of this approximation, since their equations describe a thin cod-1 brane equipped with induced gravity inside which a regularized cod-2 brane is embedded.

It is f\mbox{}inally worth to spend some words on the role of the auxiliary vector f\mbox{}ields. It is clear that they don't enter in the thin limit systems of equations (\ref{Bulkeq 2})--(\ref{cod2 jceq}) and (\ref{Bulkeq Cascading})--(\ref{cod2 jceq Cascading}), since all the objects appearing there are def\mbox{}ined in complete generality in terms of the embedding functions and of the bulk and induced metrics. This is in fact why they are ``auxiliary'': they just help to perform the pillbox integration but play no role in the f\mbox{}inal description. The reason why they are introduced at all is that, as we mention in section \ref{The auxiliary vector fields}, the integrand of (\ref{integral N}) contains several functions which have a non-trivial behaviour and therefore contribute to the pillbox integration. For example, if we worked directly with the embedding functions, more than one component of $\hvf^{_{L} \p}_{^{\! (\r)}}$ would have a non-trivial behaviour.

This is somehow an artifact of the gauge freedom, since the GNC conditions (\ref{c1 induced metric GNC trivial k app}) imply that the $\hvf^{_{L} \p}_{^{\! (\r)}}$ are not independent, and we may express the divergent part of the integral in terms of one of them, say $\hvf^{z \, \p}_{^{\! (\r)}}$. This way to solve the constraints (\ref{c1 induced metric GNC trivial k app}) would however be geometrically contrived, although analytically fair, and as a consequence the integrand of (\ref{integral N}) would have a very complicated expression. This would make it very dif\mbox{}f\mbox{}icult to perform the pillbox integration, not to speak of giving a geometrical interpretation to the result. The interesting message is that, exploiting the gauge freedom, it is possible to express the integrand of (\ref{integral N}) in such a way that only one function embodies the divergence in the geometry. The introduction of the auxiliary vector f\mbox{}ields makes it possible to embody the divergence into a quantity, the slope function, which has a clear geometrical meaning. This allows to perform the pillbox integration easily (see (\ref{integral S})), to express the result in a coordinate independent fashion (cfr.~(\ref{Ana Luiza})) and to give to it a clear geometrical interpretation (cfr.~(\ref{Glauzia})). The latter is of course independent of the auxiliary vector f\mbox{}ields, which disappear from the f\mbox{}inal expression. In some sense, they only serve to solve the constraints (\ref{c1 induced metric GNC trivial k app}) in a geometrically clever way. See also the related discussion in section \ref{De Sitter comments}.

\subsection{On the thin limit and gravity regularization}

Arguably, the most important point to discuss is whether the system (\ref{Bulkeq Cascading})--(\ref{cod2 jceq Cascading}), with appropriate boundary conditions, possesses a unique solution for every choice of the (covariantly conserved) cod-2 energy-momentum tensor. As we commented above, this would imply that the whole thin limit procedure is well-def\mbox{}ined, and that indeed the thickness hierarchy allows to derive a thin limit description of the 6D Cascading DGP model.

Let's consider f\mbox{}irst the system (\ref{Bulkeq 2})--(\ref{cod2 jceq}) (eventually imposing the $\mbbZ_2$ ref\mbox{}lection symmetry with respect to the cod-2 brane). The assessment of the existence and uniqueness of solutions is particularly dif\mbox{}f\mbox{}icult, since the embedding of the cod-1 brane is not smooth at $\mcalC_2$ and usually the existence theorems for systems of PDE rely on choosing a coordinate system where the cod-1 brane is straight (still having a smooth bulk metric). Indeed we are not aware of a general theorem which says a f\mbox{}inal word on the matter, although it may exist. Nonetheless, in analogy with \cite{Sbisa:2014gwh}, we present the following heuristic argument. Note that (\ref{Bulkeq 2}) and (\ref{junctionconditionseq 2}) are exactly the equations of a (source-free) cod-1 DGP model, the only dif\mbox{}ference being that they hold everywhere but at the cod-2 brane. It is usually assumed that the equations for the cod-1 DGP model are well-posed, assuming that the f\mbox{}ields are smooth and imposing boundary conditions where at inf\mbox{}inity in the extra dimensions the gravitational f\mbox{}ield decays and the gravitational waves are purely outgoing. In our case, the cod-1 embedding $\hvf(\chd, \t)$ is continuous but not derivable at $\t = 0 \,$, so (\ref{Bulkeq 2}) and (\ref{junctionconditionseq 2}) do not single out a unique solution any more, because there is freedom in choosing how to patch the solutions at $\t = 0^+$ and $\t = 0^-$. Indeed, one condition is needed to f\mbox{}ix the jump of the normal derivative $\detau \hvf$ across $\t = 0 \,$. We argue that this condition is provided exactly by the cod-2 junction conditions (\ref{cod2 jceq}), so that they constitute with (\ref{Bulkeq 2}) and (\ref{junctionconditionseq 2}) a well def\mbox{}ined system of PDE with the above mentioned boundary conditions.

For what concerns the system (\ref{Bulkeq Cascading})--(\ref{cod2 jceq Cascading}), we expect that the inclusion of the cod-2 induced gravity term does not change qualitatively the discussion above. With the proviso that our heuristic argument needs to be conf\mbox{}irmed, we continue our discussion exploring the consequences of its validity.

\subsubsection{Gravity regularization}

One of the most fashinating features of the Cascading DGP model is the claimed ability of induced gravity on the cod-1 brane to ``regularize gravity'' \cite{deRham:2007rw, deRham:2007xp}. This feature is seen most clearly by neglecting the cod-2 induced gravity term, so we temporarily turn to the system (\ref{Bulkeq 2})--(\ref{cod2 jceq}). It is well-known that, if we localize on a pure cod-2 brane a source term dif\mbox{}ferent from pure tension, the external f\mbox{}ield conf\mbox{}iguration on the brane diverges when we send the brane thickness to zero \cite{Cline:2003ak, Vinet:2004bk} (unless we allow for Gauss-Bonnet terms in the bulk action \cite{Bostock:2003cv}). This behaviour is reproduced by our equations since, in the absence of the cod-1 induced gravity term ($\Mft = 0$), the system (\ref{Bulkeq 2})--(\ref{cod2 jceq}) admits solutions only for a pure tension source conf\mbox{}iguration. With any other source, a ridge conf\mbox{}iguration cannot describe our system (if $\Mft = 0$), because either the cod-1 bending diverges at the cod-2 brane, either the bulk metric becomes singular there, or both. If instead we turn on the cod-1 induced gravity term, it is clear from (\ref{cod2 jceq}) that a wider class of source conf\mbox{}igurations is allowed, since the term $\Mft \, \big[ \bar{\mbf{K}} - \bar{\mbf{g}} \,\, tr \big( \bar{\mbf{K}} \big) \big]_{\!\pm}$ is not necessarily proportional to the cod-2 induced metric (see also the discussion in \cite{Sbisa:2014gwh}). For these source conf\mbox{}igurations gravity is indeed regularized, since gravity on the cod-2 brane is always f\mbox{}inite for a ridge conf\mbox{}iguration. The question if the gravity regularization mechanism is ef\mbox{}fective for \emph{any} (covariantly conserved) source conf\mbox{}iguration is more dif\mbox{}f\mbox{}icult to answer with full rigour, since it is equivalent to ask if the system (\ref{Bulkeq 2})--(\ref{cod2 jceq}) admits solutions for any form of $\bar{\mbf{T}}$. We refer to the discussion above and to our heuristic argument in favour of a positive answer.

\subsection{Conclusions and further comments}

Our results, in the form of the system (\ref{Bulkeq Cascading})--(\ref{cod2 jceq Cascading}), constitute a solid formulation of the (nested branes realization of the) 6D Cascading DGP model. The relevance of our results is best seen by comparing with the situation in \cite{deRham:2010rw, Moyassari:2011nb} where, each time a dif\mbox{}ferent solution is looked for, the pillbox integration is performed anew. In particular the thin limit solution is not sought directly, since to obtain the equations of motion for the thin conf\mbox{}iguration a putative thick ribbon conf\mbox{}iguration is constructed f\mbox{}irst, and then the thin limit is performed. With our formulation all these operations become superf\mbox{}luous, since the pillbox integration across the codimension-2 brane has been performed in complete generality, and therefore once and for all. As a consequence, one can go straight for the thin solution. In particular, it becomes trivial to write the thin limit equations for perturbations around a chosen background, without the need of discussing the singular behaviour of the perturbations and its relation with that of the background solution.

As we mentioned above, the papers \cite{deRham:2010rw, Moyassari:2011nb} implicity consider the same realization of the Cascading DGP model we considered here. Therefore it should be possible to compare their thin limit equations with our f\mbox{}inal equations, and check if they agree. Regarding \cite{deRham:2010rw}, this comparison has already been performed in \cite{Sbisa:2014vva}, where a dif\mbox{}ference between the two results have been uncovered. It was proved in \cite{Sbisa:2014vva} that indeed the pillbox integration in \cite{deRham:2010rw} is done incorrectly, and the reason has been mainly imputed to the poor faithfulness of the brane-based approach to the geometry of the system. The fact that the system (\ref{Bulkeq Cascading})--(\ref{cod2 jceq Cascading}) reproduces the equations of \cite{Sbisa:2014vva} on one hand strengthens the case for the analysis of \cite{Sbisa:2014vva} being correct against that of \cite{deRham:2010rw}, and on the other hand provides a conf\mbox{}irmation of the results of the present paper. Regarding \cite{Moyassari:2011nb}, a thorough discussion will be given in a forthcoming paper, where cosmological solutions are studied solving the equations (\ref{Bulkeq Cascading})--(\ref{cod2 jceq Cascading}).

The necessity of rigour in deriving the codimension-2 junction conditions, mandatory for branes of codimension greater than one, stole space from the discussion of other aspects, such as exemplifying how the junction conditions may be used in practice. To partially amend for this, we show in appendix \ref{appendix de Sitter} how the de Sitter solutions on the cod-2 brane, f\mbox{}irst found in \cite{Minamitsuji:2008fz}, can be straightforwardly obtained from the thin limit equations (\ref{Bulkeq Cascading})--(\ref{cod2 jceq Cascading}).

Let us f\mbox{}inally mention that cod-2 junction conditions were already studied and derived in \cite{Dyer:2009yg} for a system reminiscent of the Cascading DGP model. The system under consideration in that case is a scalar galileon f\mbox{}ield living in a f\mbox{}lat bulk with a ridge-like boundary. In light of the connection between galileons and brane-bending modes \cite{Luty:2003vm, Nicolis:2004qq}, it would be interesting to explore in some detail the relation between the results of \cite{Dyer:2009yg} and ours.

\acknowledgments
We thank CNPq (Brazil) and FAPES (Brazil) for partial f\mbox{}inancial support.

\appendix

\section{Characteristic angles of a ridge conf\mbox{}iguration}
\label{appendix ridge}

We introduce here the notions of dihedral angle and of oriental dihedral angle of a ridge conf\mbox{}iguration. We furthermore provide an interpretation of the latter as an opening angle and as a def\mbox{}icit angle.

\subsection{The dihedral angles}

Let's consider a ridge conf\mbox{}iguration constituted by two f\mbox{}lat semi-hyperplanes $\emph{1}$ and $\emph{2}$ in a 6D Euclidean bulk. Taken a point on the crest of the ridge, i.e.~the common boundary of the semi-hyperplanes, let's call $\mcal{E}$ the (2D) subspace of the tangent space which is orthogonal to the crest. The semi-hyperplanes divide the bulk into two connected regions, which are identif\mbox{}ied by the choice of the vectors $\mbf{n}_{_{1}}, \mbf{n}_{_{2}}$ normal to the semi-hyperplanes (see f\mbox{}igure \ref{fig ridge}). By our conventions, to identify one bulk region we have to choose $\mbf{n}_{_{1}}, \mbf{n}_{_{2}}$ such that they point into the chosen region, so there are only two possible choices for the couple $(\mbf{n}_{_{1}}, \mbf{n}_{_{2}})$, one opposite to the other. Each of the two choices def\mbox{}ines an \emph{orientation} for the ridge, and we call \emph{oriented ridge} a ridge conf\mbox{}iguration for which a choice for $\mbf{n}_{_{1}}, \mbf{n}_{_{2}}$ has been made. To characterize the ridge conf\mbox{}igurations, we def\mbox{}ine the \emph{dihedral angle} $\Upsilon$ of the ridge as the the convex angle (that is, positive and unoriented) between the normal vectors of the two hyperplanes
\beq
\Upsilon(\chd) = \arccos \big\langle \mbf{n}_{_{1}}, \mbf{n}_{_{2}} \big\rangle \quad ,
\eeq
where $\langle \phantom{i}, \phantom{i} \rangle$ denotes the scalar product.\footnote{By consistency, the $\arccos$ is def\mbox{}ined on $[-1, 1] \to [0, \pi] \,$.} The dihedral angle completely characterizes a (unoriented) ridge in the sense of Euclidean geometry, meaning that geometrical objects equivalent up to isometries are identif\mbox{}ied. In other words, the dihedral angle completely characterizes the conf\mbox{}iguration of the two semi-hyperplanes. On the other hand, it is clear that the value of $\Upsilon$ does not carry any information about the orientation of the ridge, since the geometric angle between $\mbf{n}_{_{1}}$ and $\mbf{n}_{_{2}}$ is the same of the one between $-\mbf{n}_{_{1}}$ and $-\mbf{n}_{_{2}} \,$. It follows that, to characterize geometrically the oriented ridge conf\mbox{}igurations, we have to introduce another quantity.

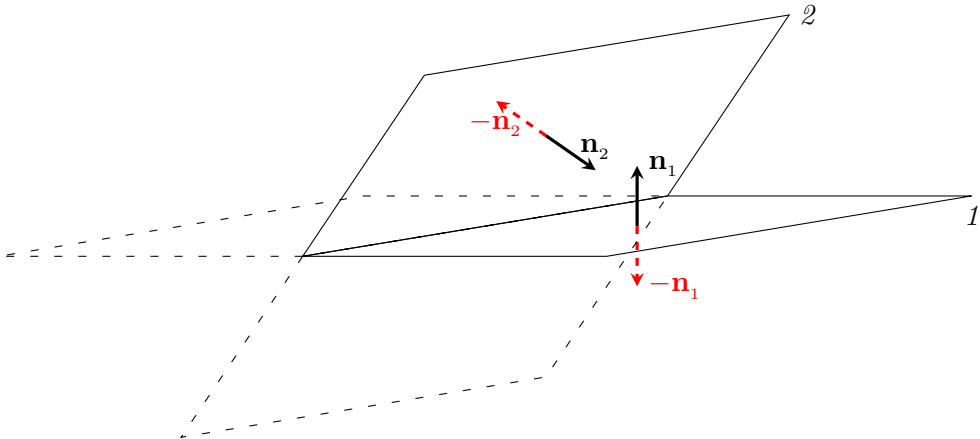
\begin{figure}[t!]
\centering
\begin{tikzpicture}[scale=0.8,>=stealth]
\draw (0,0) -- (5,0) -- (11,1) node[below]{$\emph{1}$} -- (6,1) -- cycle;
\draw (0,0) -- (2,3) -- (8,4) node[right]{$\emph{2}$} -- (6,1) --cycle;
\draw[very thick,->] (5.5,0.5) -- (5.5,1.5) node[right,color=black] {$\mathbf{n}_{_{1}}$};
\draw[very thick,color=red,dashed,->] (5.5,0.5) -- (5.5,-0.5) node[right,color=red] {$-\mathbf{n}_{_{1}}$};
\draw[very thick,->] (4,2) -- +(-35:1) node[above,color=black] {$\mathbf{n}_{_{2}}$};
\draw[very thick,color=red,dashed,->] (4,2) -- +(145:1) node[below,color=red] {$-\mathbf{n}_{_{2}}$};
\draw[thin,loosely dashed] (0,0) -- (-5,0) -- (1,1)  -- (6,1) -- cycle;
\draw[thin,loosely dashed] (0,0) -- (-2,-3) -- (4,-2) -- (6,1) --cycle;
\end{tikzpicture}
\caption{A ridge conf\mbox{}iguration. The dashed lines are the prolongation of the semi-hyperplanes.}
\label{fig ridge}
\end{figure}

To achieve this, let's note that a geometrical way to distinguish the two bulk regions is given by recognizing which region contains the prolongation of the two semi-hyperplanes, and which one doesn't. Let's call $\mbfscn_{_{1}},\mbfscn_{_{2}}$ a choice of unit vectors $\in \mcalE$ which are tangent to the semi-hyperplanes, and such that the orientations of the ordered couples $(\mbf{n}_{_{1}}, \mbfscn_{_{1}})$ and $(\mbf{n}_{_{2}}, \mbfscn_{_{2}})$ coincide. As we show in f\mbox{}igures \ref{fig orientation and rotation 1} and \ref{fig orientation and rotation 2}, the fact that a region contains (respectively, doesn't contain) the prolongation of the semi-hyperplanes is related in a one-to-one way to the orientation of the ordered couple $(\mbf{n}_{_{1}}, \mbf{n}_{_{2}})$ being discordant (respectively, concordant) with the orientation of $(\mbfscn_{_{1}},\mbf{n}_{_{1}})$. It is easy to see that this characterization is independent of the choice of $(\mbfscn_{_{1}},\mbfscn_{_{2}})$, or equivalently of which hyperplane we call ``$\!\emph{1}\,$'' and which we call ``$\!\emph{2}\,$''. Calling $\bsl_{_{1}}$, $\bsl_{_{2}}$, $\bsm_{_{1}}$ and $\bsm_{_{2}}$ the 1-forms respectively dual to $\mbfscn_{_{1}}$, $\mbfscn_{_{2}}$, $\mbf{n}_{_{1}}$ and $\mbf{n}_{_{2}}$, the orientation of $(\mbfscn_{_{1}},\mbf{n}_{_{1}})$ is conveniently expressed by the 2-form $\bsl_{_{1}} \wedge \, \bsm_{_{1}}$, meaning that an ordered couple $(\mbf{a}, \mbf{b})$ of vectors $\in \mcalE$ has concordant (respectively, discordant) orientation with $(\mbfscn_{_{1}}, \mbf{n}_{_{1}})$ if 
\beq
(\bsl_{_{1}} \wedge \, \bsm_{_{1}}) (\mbf{a}, \mbf{b}) > 0
\eeq
(respectively, $< 0$).\footnote{Since the orientations of $(\mbf{n}_{_{1}}, \mbfscn_{_{1}})$ and $(\mbf{n}_{_{2}}, \mbfscn_{_{2}})$ coincide, we have $\bsl_{_{1}} \wedge \, \bsm_{_{1}} = \bsl_{_{2}} \wedge \, \bsm_{_{2}} \,$.} We can then def\mbox{}ine a quantity, the \emph{oriented dihedral angle}
\beq
\O = \sgn \big[ (\bsl_{_{1}} \wedge \, \bsm_{_{1}}) (\mbf{n}_{_{1}}, \mbf{n}_{_{2}}) \big] \, \arccos \langle \mbf{n}_{_{1}}, \mbf{n}_{_{2}} \rangle \quad ,
\eeq
which encodes the value of the dihedral angle \emph{and} at the same time distinguishes between the orientations of the ridge. Indeed, two oriented ridge conf\mbox{}igurations which have the same dihedral angle but opposite orientations, are characterized by oriented dihedral angles whose absolute values coincide but have opposite signs. Note furthermore that by the orthonormality of $\mbfscn_{_{1}}$ and $\mbf{n}_{_{1}}$ we have
\beq
\label{Airu}
(\bsl_{_{1}} \wedge \, \bsm_{_{1}}) (\mbf{n}_{_{1}}, \mbf{n}_{_{2}}) = - \bsl_{_{1}} (\mbf{n}_{_{2}}) = - \big\langle \mbfscn_{_{1}}, \mbf{n}_{_{2}} \big\rangle \quad ,
\eeq
so we get the alternative expression
\beq
\label{Nachi}
\O = - \sgn \big\langle \mbfscn_{_{1}}, \mbf{n}_{_{2}} \big\rangle \, \arccos \langle \mbf{n}_{_{1}}, \mbf{n}_{_{2}} \rangle \quad .
\eeq
\begin{figure}[t!]
\centering
\begin{tikzpicture}[scale=0.8]
\draw[thin] (-4,0) -- (0,0);
\draw[loosely dashed,color=black!50] (0,0) -- (4,0);
\draw[thin] (0,0) -- (3.464,2);
\draw[loosely dashed,color=black!50] (0,0) -- (-3.464,-2);
\draw[->,thick] (-3,0) -- +(0,1.5) node[left,color=black] {$\mathbf{n}_{_{1}}$};
\draw[->,thick] (-3,0) -- +(1.5,0) node[below,color=black] {$\mbfscn_{_{1}}$};
\draw[->,thick] (1.299,0.75) -- +(120:1.5) node[left,color=black] {$\mathbf{n}_{_{2}}$};
\draw[->,thick] (1.299,0.75) -- +(30:1.5) node[anchor=north west,color=black] {$\mbfscn_{_{2}}$};
\end{tikzpicture}
\hspace{14mm}
\begin{tikzpicture}[scale=0.85]
\draw[->,thick] (0,0) -- +(0,3) node[left,color=black] {$\mathbf{n}_{_{1}}$};
\draw[->,thick] (0,0) -- +(3,0) node[below,color=black] {$\mbfscn_{_{1}}$};
\draw[->,>=stealth] (1,0) arc (0:90:1);
\draw[->,thick] (0,0) -- +(120:3) node[left,color=black] {$\mathbf{n}_{_{2}}$};
\draw[->,>=stealth] (0,2) arc (90:120:2);
\end{tikzpicture}
\caption{Concordant orientation of $(\mbf{n}_{_{1}}, \mbf{n}_{_{2}})$ and $(\mbfscn_{_{1}},\mbf{n}_{_{1}})$.}
\label{fig orientation and rotation 1}
\end{figure}
\begin{figure}[t]
\centering
\begin{tikzpicture}[scale=0.8]
\draw[thin] (-4,0) -- (0,0);
\draw[loosely dashed,color=black!50] (0,0) -- (4,0);
\draw[thin] (0,0) -- (3.464,-2);
\draw[loosely dashed,color=black!50] (0,0) -- (-3.464,2);
\draw[->,thick] (-3,0) -- +(0,1.5) node[left,color=black] {$\mathbf{n}_{_{1}}$};
\draw[->,thick] (-3,0) -- +(1.5,0) node[below,color=black] {$\mbfscn_{_{1}}$};
\draw[->,thick] (1.299,-0.75) -- +(60:1.5) node[left,color=black] {$\mathbf{n}_{_{2}}$};
\draw[->,thick] (1.299,-0.75) -- +(-30:1.5) node[anchor=north,color=black] {$\mbfscn_{_{2}}$};
\end{tikzpicture}
\hspace{14mm}
\begin{tikzpicture}[scale=0.85]
\draw[->,thick] (0,0) -- +(0,3) node[left,color=black] {$\mathbf{n}_{_{1}}$};
\draw[->,thick] (0,0) -- +(3,0) node[below,color=black] {$\mbfscn_{_{1}}$};
\draw[->,>=stealth] (1,0) arc (0:90:1);
\draw[->,thick] (0,0) -- +(60:3) node[left,color=black] {$\mathbf{n}_{_{2}}$};
\draw[->,>=stealth] (0,2) arc (90:60:2);
\end{tikzpicture}
\caption{Discordant orientation of $(\mbf{n}_{_{1}}, \mbf{n}_{_{2}})$ and $(\mbfscn_{_{1}},\mbf{n}_{_{1}})$.}
\label{fig orientation and rotation 2}
\end{figure}

\subsection{The closing angle and the def\mbox{}icit angle}

We can give two geometrical interpretations to the oriented dihedral angle, which are more related to the bulk region individuated by the oriented cod-1 brane than to the brane conf\mbox{}iguration itself.

The f\mbox{}irst interpretation emerges by comparing an oriented ridge conf\mbox{}iguration with the ``f\mbox{}lat'' conf\mbox{}iguration $\O = 0 \,$. Let's consider an oriented ridge with $\O > 0\,$. In this case, $\mbf{n}$ points into the bulk region $\mcalM_{_{+}}$ which do not contain the prolongation of the semi-hyperplanes, and therefore $\mcalM_{_{+}}$ is properly contained in the bulk region associated to a f\mbox{}lat ridge. Furthermore, the bigger $\O$ the more the bulk region $\mcalM_{_{+}}$ ``closes in'', and tends to a semi-hyperplane in the limit $\O \to \pi\,$. We can pictorially say that a positive $\O$ measures how much $\mcalM_{_{+}}$ is closed with respect to the case of a f\mbox{}lat ridge conf\mbox{}iguration. In the case $\O < 0\,$, instead, $\mcalM_{_{+}}$ contains the prolongation of the semi-hyperplanes, and so the bulk region associated to a f\mbox{}lat ridge is properly contained in $\mcalM_{_{+}}$. Furthermore, the bigger $\abs{\O}$ the more $\mcalM_{_{+}}$ ``opens up'', and tends to the full 6D space in the limit $\O \to -\pi\,$. Analogously, we can say that a negative $\O$ measures how much $\mcalM_{_{+}}$ is open with respect to the case of a f\mbox{}lat ridge conf\mbox{}iguration. Adhering to the standard convention that a negative closing angle is equivalent to an opening angle, we can then interpret $\O$ as the \emph{closing angle} of the oriented ridge.

The second interpretation is uncovered by constructing a manifold $\mcalM$ which is $\mbbZ_2$ symmetric with respect to the semi-hyperplanes, by copying and pasting the bulk region $\mcalM_{_{+}}$. As is well-known, the resulting manifold $\mcalM$ is the product of the Euclidean 4D space with a 2D Riemannian manifold which is isomorphic to a cone. It can be shown that in the case $\O > 0$ the cone is characterized by a \emph{def\mbox{}icit angle} equal to $2 \, \O$, while in the case $\O < 0$ is characterized by an \emph{excess angle} equal to $2 \, \abs{\O}$ (see \cite{Sbisa:2014dwa} for an extensive analysis). Def\mbox{}ining (as is standard) the def\mbox{}icit angle $\D$ of the cone such that a negative def\mbox{}icit angle corresponds to a (positive) excess angle, we get
\beq
\label{Helena}
\D = 2 \, \O \quad .
\eeq
Therefore we get a second geometrical interpretation of the oriented dihedral angle as half of the def\mbox{}icit angle of the extra-dimensional 2D manifold. The relevance of such an interpretation of course resides in the fact that, in our set-up, the copying and pasting procedure indeed gives the physical 6D manifold.

\section{Convergence properties}
\label{appendix convergence}

We refer to \cite{Rudin:FunctAnalysis} for the details on the mathematical aspects of this and the following appendices. Let $\Th \subset \mbbR^d$ be an open set such that its closure $\oTh$ is compact ($\Th$ is then called a compact closure open set, in brief c.c.o.s.). The space $C^{n}( \oTh )$ of continuous functions $\oTh \to \mbbR$ which are of class $C^n$ in $\Th$, and whose partial derivatives of order $\leq n$ are extended with continuity to $\oTh$, is a Banach space when equipped with the norm
\beq
\label{Cn norm}
\norm{f}_{C^{n}(\oTh)} = \norm{f}_{C^{0}(\oTh)} + \sum_{j = 1}^{n} \sum_{\abs{\a} = j} \norm{\de^\a f}_{C^{0}(\oTh)}
\eeq
where $\norm{f}_{C^{0}(\oTh)} \equiv \sup_{p \in \oTh} \,\, \abs{f(p)} \,$. This norm is also called the $\sup$ norm. On the other hand, the space $Lip (\oTh)$ of Lipschitz functions is the set of functions $f : \oTh \to \mbbR$ such that 
\beq
\abs{f}_{Lip (\oTh)} = \sup_{x,y \, \in \, \oTh, \, x \neq y} \frac{\abs{f(x) - f (y)}}{\abs{x - y}}
\eeq
is f\mbox{}inite, and is a Banach space as well when equipped with the norm $\norm{f}_{Lip (\oTh)} \equiv \norm{f}_{C^{0}(\oTh)} + \abs{f}_{Lip (\oTh)} \,$. These def\mbox{}initions are generalized to functions $\oTh \to \mbbR^{m}$, $f = (f_1 , \ldots , f_m)$ by def\mbox{}ining $\norm{f} = \sum_{i = 1}^{m} \, \norm{f_i} \,$. Regarding the discussion in section \ref{Convergence properties of the families}, it is straightforward to see that the convergence in the norm (\ref{Hitomi}) implies convergence in the following norm
\beq
\norm{f}_{} = \norm{f}_{C^{0}(\oW_{\!\! q})} + \sum_{j = 1}^{n} \sum_{\abs{\a} = j} \norm{\de^{\a}_{\chd} f}_{C^{0}(\oW_{\!\! q})} \label{Anri} \quad ,
\eeq
since $\abs{\phantom{f}}_{Lip}$ is always non-negative.

In our analysis, we work with families of functions labelled by a continuous parameter $\r \in I_{_{+}} \,$, where $I_{_{+}} = ( 0, l \, ]$ (or $I_{_{+}} = ( 0, d \, ]$, with $d$ the positive number def\mbox{}ined in appendix \ref{appendix auxiliary fields}, when the auxiliary vector f\mbox{}ields and/or the slope function are concerned). The thin limit of the ribbon brane is realized by the limit $\r \to 0^+$. Since in the theory of Banach spaces often theorems are expressed in terms of sequences of functions, it is worth clarifying the equivalence between the two approaches. The link between the continuous and the numerable formulation is given by the well-known
\begin{proposition}
\label{prop 1}
Let $(X, d)$ be a metric space, and let us consider a function $F : I_{_{+}} \to X$. Then $F(x) \to \bar{F}$ when $x \to 0^+$ if and only if, for every sequence $\{ x_{n} \}_n$ contained in $I_{_{+}}$ such that $x_{n} \to 0$, we have $F_{n} \equiv F(x_{n}) \to \bar{F} \,$.
\end{proposition}
\noi from which it follows the characterization of Banach spaces with continuous families of functions
\begin{corollary}
Let $(B, \norm{\phantom{f}})$ be a Banach space, and let $F : I_{_{+}} \to B \,$. Then $F$ has a limit in $B$ when $x \to 0^+$ if and only if it satisf\mbox{}ies the Cauchy condition in $0^+$, i.e.~for every $\ep > 0$, there exists a $\d > 0$ such that $0 < x, x^{\p} < \d$ implies $d_{_{B}}\big(F(x), F(x^{\p})\big) = \norm{F(x) - F(x^{\p})} < \ep \,$.
\end{corollary}

\noi For simplicity, we say indif\mbox{}ferently that a function $F :I_{_{+}} \to B$ has a limit for $x \to 0^+$ or that it converges for $x \to 0^+$. In the rest of the paper the r\^ole of $F(x)$ is played by the families $\hvf_{^{(\r)}}^{_{A}}$, $g^{_{(\r)}}_{_{AB}}$ and so on, while the r\^ole of $x$ is played by the continuous parameter $\r \,$. Below we imlicitly assume $\r \in I_{_{+}}$.

In section \ref{Convergence properties} we def\mbox{}ine the convergence properties of the bulk metric and of the cod-1 embedding function. However, the induced metric, the curvature tensors and the auxiliary vector f\mbox{}ields are built from products, compositions and reciprocals of (partial derivatives of) these quantities. In light of the analysis of the appendices \ref{appendix ind metr and ext curv} and \ref{appendix auxiliary fields}, we mention the following propositions which are easy to prove.

\begin{proposition}
\label{proposition 1}
Let $\Th \subset \mbbR^m$ be a c.c.o.s., and let $\{ f_{_{\! (\r)}} \}_{_{\! \r}}, \{ h_{_{(\r)}} \}_{_{\! \r}} : \oTh \to \mbbR$ be two families of continuous functions. Suppose that $f_{_{\! (\r)}} \to f$ and $h_{_{(\r)}} \to h$ in the norm $\norm{\phantom{f}}_{C^{0}(\oTh)}$ when $\r \to 0^+$. Then the family $\{ f_{_{\! (\r)}} h_{_{(\r)}} \}_{_{\! \r}}$ converges to $f h$ in the norm $\norm{\phantom{f}}_{C^{0}(\oTh)}$.
\end{proposition}

\begin{proposition}
\label{proposition 2}
Let $\{ f_{_{\! (\r)}} \}_{_{\! \r}} : \oXi \subset \mbbR^m \to \mbbR^p$ be a family of continuous functions, with $\Xi$ a c.c.o.s., and $\{ h_{_{(\r)}} \}_{_{\! \r}} : \oTh \subset \mbbR^p \to \mbbR$ a family of Lipschitz functions, with $\Th$ a c.c.o.s., such that $f_{_{\! (\r)}} \big( \oXi \, \big) \subset \oTh$ for every $\r \in I_{_{+}} \,$. Let $f_{_{\! (\r)}} \to f$ in the norm $\norm{\phantom{f}}_{C^{0}(\oXi)}$ and $h_{_{(\r)}} \to h$ in the norm $\norm{\phantom{f}}_{Lip(\oTh)}$ when $\r \to 0^+$. Then the family $\{ h_{_{(\r)}} \! \circ f_{_{\! (\r)}} \}_{_{\! \r}}$ converges to $h \circ f$ in the norm $\norm{\phantom{f}}_{C^{0}(\oXi)}$.
\end{proposition}

\begin{proposition}
\label{proposition 3}
Let $\Th \subset \mbbR^m$ be a c.c.o.s., and let $\{ f_{_{\! (\r)}} \}_{_{\! \r}} : \oTh \to \mbbR$ be a family of continuous functions. Suppose that $f_{_{\! (\r)}} \to f$ in the norm $\norm{\phantom{f}}_{C^{0}(\oTh)}$ when $\r \to 0^+ \,$, and that $f$ does not vanish in $\oTh$. Then there exists a $\d > 0$ such that, for $0 < \r < \d$, $f_{_{\! (\r)}}$ does not vanish in $\oTh$ and the family $\{ 1/f_{_{\! (\r)}} \}_{_{\! \r}}$ converges to $1/f$ in the norm $\norm{\phantom{f}}_{C^{0}(\oTh)}$.
\end{proposition}

\subsection{Convergence in a family of sets}

We formalize here the notion of ``convergence in a family of sets''. Indicating 
\beq
Q_{q}^{_{(\r)}} = \bar{U}_{q} \times \big\{ \big( -l, -\r \big) \cup \big( \r, l \big) \big\} \quad ,
\eeq
we say that $\{ f_{_{\! (\r)}} \}_{_{\! \r}}$ converges uniformly in $\{ \oQ_{q}^{_{(\r)}} \}_{_{\! \r}}$ to a function
\beq
f : \overline{\bar{U}}_{q} \times \big\{ \big[ -l, 0 \big) \cup \big( 0 , l \, \big] \big\} \to \mbbR
\eeq
if, for every $\ep > 0$, there exists a $\d > 0$ such that
\beq
\r < \d \quad \Rightarrow \quad \babs{f_{_{\! (\r)}}(\chd, \t) - f(\chd, \t)} < \ep \quad \text{for every} \,\, (\chd, \t) \in \oQ_{q}^{_{(\r)}} \quad.
\eeq

To characterize the uniform convergence in the family of sets $\{ \oQ_{q}^{_{(\r)}} \}_{_{\! \r}} \, $, we use the following propositions

\begin{lemma}
\label{lemma Qr 1}
Let $q \in \mcalC_2$ and let $\{ f_{_{\! (\r)}} \}_{_{\! \r}}$ be a family of continuous functions $\oW_{{\! q}} \to \mbbR \,$. Let $\{ f_{_{\! (\r)}} \}_{_{\! \r}}$ converge uniformly in $\{ \oQ_{q}^{_{(\r)}} \}_{_{\! \r}}$ to a function $f : \oW_{\! {q}} \setminus \mcalC_2 \to \mbbR \,$, and let $\Th$ be a c.c.o.s.~such that $\oTh \, \cap \, \mcalC_2 = \varnothing \,$. Then $\{ f_{_{\! (\r)}} \}_{_{\! \r}}$ converges uniformly to $f$ in $\oTh \,$.
\end{lemma}
\begin{proof}
Let's def\mbox{}ine the normal distance $d_{_{\perp}}(p)$ of a point $p \in \oTh$ from $\mcalC_2$ as the distance measured from $p$ to $\mcalC_2$ along a geodesic normal to $\mcalC_2$ and passing through $p \,$. We have $d_{_{\perp}}(p) = \t (p) \,$, where $\t(p)$ is the normal coordinate of $p$ in the cod-1 Gaussian normal cordinates. It follows that $d_{_{\perp}}$ is continuous on a compact set, and therefore has a minimum in $\oTh \,$. Since $\oTh \cap \mcalC_2 = \varnothing \,$, this minimum is strictly positive. Therefore, there exists a $\d > 0$ such that $\r < \d$ implies $\oTh \subset \oQ_{q}^{_{(\r)}}$. The thesis then follows simply from the def\mbox{}inition of uniform convergence in $\{ \oQ_{q}^{_{(\r)}} \}_{_{\! \r}} \,$.
\end{proof}

\begin{corollary}
\label{corollary Qr}
Let $q \in \mcalC_2$ and let $\{ f_{_{\! (\r)}} \}_{_{\! \r}}$ be a family of continuous functions $\oW_{{\! q}} \to \mbbR \,$, which converges uniformly in $\{ \oQ_{q}^{_{(\r)}} \}_{_{\! \r}}$ to a function $f : \oW_{\! {q}} \setminus \mcalC_2 \to \mbbR \,$. Then $f$ is continuous.
\end{corollary}
\begin{proof}
Let $0 < a < l \,$, and def\mbox{}ine $S_{q}(a) = \bar{U}_{q} \times \big\{ \big( -l, -a \big) \cup \big( a, l \big) \big\}$. Since $\overline{S}_{q}(a)$ is compact, the lemma \ref{lemma Qr 1} implies that $\{ f_{_{\! (\r)}} \}_{_{\! \r}}$ converges uniformly in $\overline{S}_{q}(a)$ to $f$. On the other hand, since $C^{0} \big( \overline{S}_{q}(a) \big)$ with the $\sup$ norm is a Banach space, $f$ is continuous on $\overline{S}_{q}(a) \,$. Since this holds for every $0 < a < l \,$, the thesis follows.
\end{proof}

\begin{proposition}
\label{proposition Qr}
Let $q \in \mcalC_2$ and let $\{ f_{_{\! (\r)}} \}_{_{\! \r}}$ be a family of continuous functions $\oW_{{\! q}} \to \mbbR \,$. Let $\{ f_{_{\! (\r)}} \}_{_{\! \r}}$ converge uniformly in $\{ \oQ_{q}^{_{(\r)}} \}_{_{\! \r}}$ to $f: \oW_{\! {q}} \setminus \mcalC_2 \to \mbbR \,$. Then, for every $\chd \in \overline{\bar{U}}_{q} \,$, $\lim_{\r \to 0^+} f_{_{\! (\r)}}(\chd, \r)$ exists f\mbox{}inite if and only if $\lim_{\t \to 0^{+}} f(\chd, \t)$ exists f\mbox{}inite. Furthermore, in this case we have
\beq
\lim_{\t \to 0^{+}} f(\chd, \t) = \lim_{\r \to 0^+} f_{_{(\r)}}(\chd, \r) \quad .
\eeq
An analogous proposition involving $\lim_{\r \to 0^+} f_{_{\! (\r)}}(-\chd, \r)$ and $\lim_{\t \to 0^{-}} f(\chd, \t)$ holds.
\end{proposition}
\begin{proof}
Let's consider the case correspondent to the $+$ sign, the other being completely analogous. Since the proposition 

Let's f\mbox{}ix $\ep > 0$, and def\mbox{}ine $\vep = \ep/2 \, $. Since $\{ f_{_{\! (\r)}} \}_{_{\! \r}}$ converge uniformly in $\{ \oQ_{q}^{_{(\r)}} \}_{_{\! \r}}$ to $f$, it follows in particular that $\exists \, \d_{_{1}} > 0$ such that $\r < \d_{_{1}}$ implies
\beq
\babs{f_{_{\! (\r)}}(\chd, \r) - f(\chd, \r)} < \vep \quad .
\eeq \\
$(\Rightarrow)\, $ Calling $F(\chd) = \lim_{\r \to 0^+} f_{_{\! (\r)}}(\chd, \r) \,$, there exists a $\d_{_{2}} > 0$ such that $\r < \d_{_{2}}$ implies
\beq
\babs{f_{_{\! (\r)}}(\chd, \r) - F(\chd)} < \vep \quad .
\eeq
Let us call $\d = \mathrm{min} \{ \d_{_{1}}, \d_{_{2}} \} \,$. Since
\beq
\babs{f_{_{\! (\r)}}(\chd, \r) - f(\chd, \r) \pm F(\chd)} \geq \Babs{\babs{f_{_{\! (\r)}}(\chd, \r) - F(\chd)} - \babs{f(\chd, \r) - F(\chd)}} \quad ,
\eeq
we conclude that $\r < \d$ implies
\beq
\babs{f(\chd, \r) - F(\chd)} \leq \babs{f_{_{\! (\r)}}(\chd, \r) - F(\chd)} + \babs{f_{_{\! (\r)}}(\chd, \r) - f(\chd, \r)} < \frac{\vep}{2} + \frac{\vep}{2} = \ep \quad .
\eeq
$(\Leftarrow)\, $ Calling $\bar{F}(\chd) = \lim_{\t \to 0^{+}} f(\chd, \t) \,$, there exists a $\d_{_{2}} > 0$ such that $\r < \d_{_{2}}$ implies
\beq
\babs{f(\chd, \r) - \bar{F}(\chd)} < \vep \quad .
\eeq
Let us call $\d = \mathrm{min} \{ \d_{_{1}}, \d_{_{2}} \} \,$. Since
\beq
\babs{f_{_{\! (\r)}}(\chd, \r) - f(\chd, \r) \pm \bar{F}(\chd)} \geq \Babs{\babs{f_{_{\! (\r)}}(\chd, \r) - \bar{F}(\chd)} - \babs{f(\chd, \r) - \bar{F}(\chd)}} \quad ,
\eeq
we conclude that $\r < \d$ implies
\beq
\babs{f_{_{\! (\r)}}(\chd, \r) - \bar{F}(\chd)} \leq \babs{f(\chd, \r) - \bar{F}(\chd)} + \babs{f_{_{\! (\r)}}(\chd, \r) - f(\chd, \r)} < \frac{\vep}{2} + \frac{\vep}{2} = \ep \quad .
\eeq
\end{proof}

\subsubsection{Convergence of products, compositions and reciprocals}

Let us now omit the dependence on the coordinates $\chd$, and consider a set $K \subset [-d , d \, ]$ (possibly $K = [-d , d \, ]$ itself). Let's def\mbox{}ine the uniform convergence in the family of sets $\{ K \cap [-d , -\r \, ] \cup [ \, \r , d \, ] \}_{_{\! \r}}$ exactly as the uniform convergence in $\{ \oQ_{q}^{_{(\r)}} \}_{_{\! \r}} \,$. That is, we say that a family of functions $\{ f_{_{\! (\r)}} \}_{_{\! \r}}$ converges uniformly in $\{ K \cap [ -d , - \r \, ] \cup [ \, \r , d \, ] \}_{_{\! \r}}$ to a function $f : K \cap [-d , d \,] \setminus \{ 0 \} \to \mbbR$ if, for every $\ep > 0$, there exists a $\d > 0$ such that
\beq
0 < \r < \d \quad \Rightarrow \quad \babs{f_{_{\! (\r)}}(x) - f(x)} < \ep \quad \text{for every} \,\, x \in K \cap [ -d , -\r \, ] \cup [ \, \r , d \, ] \quad .
\eeq
It is trivial to see that the uniform convergence in $\{ \oQ_{q}^{_{(\r)}} \}_{_{\! \r}}$ implies the uniform convergence in $\{ K \cap [-d, -\r \,] \cup [ \, \r , d \, ] \}_{_{\! \r}}$ at every f\mbox{}ixed $\chd \in \overline{\bar{U}}_{q} \,$, for any choice of $K$. Choosing $K = [ -d , 0 \, ]$ and $K = [ \, 0 , d \, ]$ one gets the def\mbox{}inition of uniform convergence in $\{ [-d, -\r \,] \}_{_{\! \r}}$ and in $\{ [ \, \r , d \, ] \}_{_{\! \r}}$ respectively.

Using the corollary \ref{corollary Qr} and the proposition \ref{proposition Qr} it is straightforward to generalize the propositions \ref{proposition 1}, \ref{proposition 2}  and \ref{proposition 3}, that is the following propositions hold

\begin{lemma}
\label{lemma 1 Qr}
Let $\{ f_{_{\! (\r)}} \}_{_{\!\r}} : [ -d , 0 \,] \to \mbbR$ be a family of continuous functions, and let $K \in [ -d , 0 \,] \,$. Let $\{ f_{_{\! (\r)}} \}_{_{\!\r}}$ converge uniformly in $\{ [ -d, - \r \, ] \}_{_{\! \r}}$ to $f : [ -d, 0 ) \to \mbbR$.
Then $\{ f_{_{\! (\r)}} \}_{_{\!\r}}$ converges uniformly in $\{ K \cap [ -d, - \r \, ] \}_{_{\! \r}}$ to $f$.
\end{lemma}

\begin{proposition}
\label{proposition 1 Qr}
Let $\{ f_{_{\! (\r)}} \}_{_{\!\r}}, \{ h_{_{(\r)}} \}_{_{\!\r}} : [ -d , 0 \,] \to \mbbR$ be two families of continuous functions. Let $\{ f_{_{\! (\r)}} \}_{_{\!\r}}$ converge uniformly in $\{ [ -d, - \r \, ] \}_{_{\! \r}}$ to $f : [ -d, 0 ) \to \mbbR$, and let $\{ h_{_{(\r)}} \}_{_{\!\r}}$ converge uniformly in $\{ [ -d, - \r \, ] \}_{_{\! \r}}$ to $h : [ -d , 0 ) \to \mbbR$. Let the limits $\lim_{\r \to 0^+} f_{_{(\r)}}(-\r)$ and $\lim_{\r \to 0^+} h_{_{(\r)}}(-\r)$ exist f\mbox{}inite. Then the family of products $\{ f_{_{\! (\r)}} h_{_{(\r)}} \}_{_{\! \r}}$ converges to $f h$ uniformly in $\{ [ -d, - \r \, ] \}_{_{\! \r}}$.
\end{proposition}

\begin{proposition}
\label{proposition 2 Qr}
Let $\{ f_{_{\! (\r)}} \}_{_{\!\r}} :  [ -d , 0 \,] \to \mbbR^n$ be a family of continuous functions, and $\{ h_{_{(\r)}} \}_{_{\!\r}} : \oTh \subset \mbbR^n \to \mbbR$ a family of Lipschitz functions, with $\Th$ a c.c.o.s., such that $f_{_{\! (\r)}} \big( [ -d , 0 \,] \big) \subset \oTh$ for every $\r \,$. Let $\{ f_{_{\! (\r)}} \}_{_{\!\r}}$ converge uniformly in $\{ [ -d , -\r \, ] \}_{_{\! \r}}$ to $f : [ -d, 0 ) \to \mbbR \,$, and let $\{ h_{_{(\r)}} \}_{_{\!\r}}$ converge to $h : \oTh \to \mbbR$ in the norm $\norm{\phantom{f}}_{Lip(\oTh)} \,$. Then the family $\{ h_{_{(\r)}} \! \circ f_{_{\! (\r)}} \}_{_{\!\r}}$ converges uniformly in $\{ [ -d , -\r \, ] \}_{_{\! \r}}$ to $h \circ f$.
\end{proposition}

\begin{proposition}
\label{proposition 3 Qr}
Let $K \subset [-d , 0 \,]$, and let $\{ f_{_{\! (\r)}} \}_{_{\!\r}} : [ -d, 0 \,] \to \mbbR$ be a family of continuous functions. Let $\{ f_{_{\! (\r)}} \}_{_{\!\r}}$ converge uniformly in $\{ [ -d , -\r \, ] \}_{_{\! \r}}$ to $f : [ -d , 0 ) \to \mbbR$, such that $\inf_{K \cap [ -d , 0 )} \,\, \abs{f} > 0 \,$. Then there exists a $\d > 0$ such that $1/f_{_{\! (\r)}}$ is well-def\mbox{}ined on $K \cap [-d , 0 \,]$ for $0 < \r < \d$, and the family $\{ 1/f_{_{\! (\r)}} \}_{_{\! \r}}$ converges uniformly in $\{ K \cap [ -d , -\r \, ] \}_{_{\! \r}}$ to $1/f \,$.
\end{proposition}

\noi Analogous versions of the lemma \ref{lemma 1 Qr} and of the propositions \ref{proposition 1 Qr}, \ref{proposition 2 Qr} and \ref{proposition 3 Qr} hold for the uniform convergence respectively in $\{ [ \, \r , d \, ] \}_{_{\! \r}}$ and in $\{ K \cap [ \, \r , d \, ] \}_{_{\! \r}} \,$.

\subsection{Convergence of the slope function}

Let us remind the well-known

\begin{proposition}
\label{proposition well-known}
Let $K \subset \mbbR$ be a closed interval, and let $\{ f_{_{\! (\r)}} \}_{_{\! \r}}$ be a family of $C^1$ functions $K \to \mbbR \, $. Let the family of derivative functions $\{ f^{\p}_{^{\! (\r)}} \}_{_{\! \r}}$ converge uniformly in $K$ to a function $g$, and let $\{ f_{_{\! (\r)}} \}_{_{\! \r}}$ converge at least in a point of of $K$. Then the family $\{ f_{_{\! (\r)}} \}_{_{\! \r}}$ converges uniformly in $K$ to a derivable function $f$, and $f^{\p} = g$.
\end{proposition}

The proposition above can be generalized to the uniform convergence in $\{ [ -d , -\r \, ] \}_{_{\! \r}}$ and in $\{ [ \, \r , d \, ] \}_{_{\! \r}}\,$. In particular we have
\begin{proposition}
\label{proposition slope minus}
Let $\{ f_{_{\! (\r)}} \}_{_{\! \r}}$ be a family of $C^1$ functions $[-d , 0 \, ] \to \mbbR \, $, and let the family of derivative functions $\{ f^{\p}_{^{\! (\r)}} \}_{_{\! \r}}$ converge uniformly in $\{ [-d , -\r \, ] \}_{_{\! \r}}$ to a function $g : [-d , 0 \, ) \to \mbbR$. Let $\{ f_{_{\! (\r)}}(-d) \}_{_{\! \r}}$ converge \'a la Cauchy. Then the family $\{ f_{_{\! (\r)}} \}_{_{\! \r}}$ converges uniformly in $\{ [-d , -\r \, ] \}_{_{\! \r}}$ to a derivable function $f$, and $f^{\p} = g \,$.
\end{proposition}
\begin{proof}
Let us take a number $a \in (0 , d)\,$. Since the family $\{ f^{\p}_{^{\! (\r)}} \}_{_{\! \r}}$ converges uniformly in $\{ [-d , -\r \, ] \}_{_{\! \r}}$ to $g$, it converges uniformly to $g$ also in $[-d , -a \, ] \,$. Then by the proposition \ref{proposition well-known} the family $\{ f_{_{\! (\r)}} \}_{_{\! \r}}$ converges uniformly in $[-d , -a \, ]$ to a derivable function $f$, and $f^{\p} = g \,$. Since this holds for every $a \in (0 , d) \,$, this procedure def\mbox{}ines a derivable function $f : [-d , 0 \, ) \to \mbbR \,$. Let us see that $\{ f_{_{\! (\r)}} \}_{_{\! \r}}$ converges uniformly in $\{ [-d , -\r \, ] \}_{_{\! \r}}$ to $f$. Let us f\mbox{}ix $\ep > 0$, and def\mbox{}ine $\vep_{_{1}} = \ep/2$, $\vep_{_{2}} = \ep/2d \,$. Be $x \in [-d , 0 \, )$. Applying the Lagrange theorem to $f_{_{\! (\r)}} - f$ on $[-d , x \,]$, we get
\beq
\big( f_{_{\! (\r)}} - f \big)(x) = \big( f_{_{\! (\r)}} - f \big)(-d) + (x + d) \, \big( f^{\p}_{^{\! (\r)}} - f^{\p} \big)(\xi)
\eeq
and therefore
\beq
\babs{f_{_{\! (\r)}}(x) - f (x)} \leq \babs{f_{_{\! (\r)}}(-d) - f (-d)} + d \, \babs{f^{\p}_{^{\! (\r)}}(\xi) - f^{\p}(\xi)}
\eeq
From the hypothesis we know that $\exists \, \d_{_{1}} > 0$ such that
\beq
0 < \r < \d_{_{1}} \quad \Rightarrow \quad \babs{f_{_{\! (\r)}}(-d) - f (-d)} < \vep_{_{1}} \quad ,
\eeq
and that $\exists \, \d_{_{2}} > 0$ such that
\beq
0 < \r < \d_{_{2}} \quad \Rightarrow \quad \babs{f^{\p}_{^{\! (\r)}}(\xi) - f^{\p}(\xi)} < \vep_{_{2}} \quad \forall \, \xi \in [-d, -\r \, ] \quad .
\eeq
Calling $\d = \min \{ \d_{_{1}}, \d_{_{2}} \}$, I get
\beq
0 < \r < \d \quad \Rightarrow \quad \babs{f_{_{\! (\r)}}(x) - f (x)} < \vep_{_{1}} + d \, \vep_{_{2}} = \ep
\eeq
for every $x \in [-d, -\r \, ] \,$, and the thesis follows.
\end{proof}

Moreover we have
\begin{proposition}
\label{proposition slope plus}
Let $\{ f_{_{\! (\r)}} \}_{_{\! \r}}$ be a family of $C^1$ functions $[0 , d \, ] \to \mbbR \, $, and let the family of derivative functions $\{ f^{\p}_{^{\! (\r)}} \}_{_{\! \r}}$ converge uniformly in $\{ [ \, \r, d \, ] \}_{_{\! \r}}$ to a function $g : ( 0, d \, ] \to \mbbR$ such that $\lim_{x \to 0^+} g(x) = G$ exists f\mbox{}inite. Let $\lim_{\r \to  0^+} f_{_{\! (\r)}}(\r) = F$ exist f\mbox{}inite. Then the family $\{ f_{_{\! (\r)}} \}_{_{\! \r}}$ converges uniformly in $\{ [ \, \r, d \, ] \}_{_{\! \r}}$ to a derivable function $f$, and $f^{\p} = g \,$.
\end{proposition}
\begin{proof}
Let us def\mbox{}ine the function $\bar{g} : [ \, 0 , d \,] \to \mbbR$
\begin{equation}
\bar{g}(x) =
\begin{cases}
\phantom{a} g(x) & \text{if} \,\,\, x \in (0, d \,] \quad ,\\
\phantom{a} G & \text{if} \,\,\, x = 0 \quad .
\end{cases}
\end{equation}
By hypothesis $\bar{g}$ is continuous and therefore primitivable. Let's def\mbox{}ine then the function $f : [ \, 0 , d \,] \to \mbbR$ as the primitive of $\bar{g}$ which obeys
\beq
f(x) = F + \int_{0}^{x} \bar{g}(y) \, dy \quad .
\eeq
The function $f$ is derivable (since it is a primitive) and $f^{\p} = g$ over $( 0 , d \,] \,$. Let's see that the family $\{ f_{_{\! (\r)}} \}_{_{\! \r}}$ converges uniformly in $\{ [ \, \r, d \, ] \}_{_{\! \r}}$ to $f$. Since $\abs{\bar{g}}$ is continuous on the compact set $[ \, 0 , d \,]$, by the Weierstrass theorem it assumes there a maximum value $M \, $. Let us f\mbox{}ix $\ep > 0$, and def\mbox{}ine $\vep_{_{1}} = \ep/3$ and $\vep_{_{2}} = \ep/3d \,$. Taken $0 < \r < d \,$ and applying the Lagrange theorem on $[ \, \r , x \,] \,$, we have that
\beq
\big( f_{_{\! (\r)}} - f \big)(x) = \big( f_{_{\! (\r)}} - f \big)(\r) + (x - \r) \, \big( f^{\p}_{^{\! (\r)}} - f^{\p} \big)(\xi)
\eeq
for every $x \in (\r, d \,]$, where $\xi \in ( \r , x ) \,$. It follows that
\beq
\babs{\big( f_{_{\! (\r)}} - f \big)(x)} \leq \babs{f_{_{\! (\r)}}(\r) - F \,} + \babs{F - f(\r)} + d \, \babs{\big( f^{\p}_{^{\! (\r)}} - g \big)(\xi)}
\eeq
and using
\beq
\babs{F - f(\r)} = \bbabs{\int_{0}^{\r} \bar{g}(y) \, dy \,} \leq \int_{0}^{\r} \babs{\bar{g}(y)} \, dy \leq M \, \r
\eeq
we arrive at
\beq
\babs{f_{_{\! (\r)}}(x) - f(x)} \leq \babs{f_{_{\! (\r)}}(\r) - F \,} + M \, \r + d \, \babs{f^{\p}_{^{\! (\r)}}(\xi) - g(\xi)} \quad .
\eeq
Since $\lim_{\r \to  0^+} f_{_{\! (\r)}}(\r) = F$, there exists $\d_{_{1}} > 0$ such that $0 < \r < \d_{_{1}} \Rightarrow \babs{f_{_{\! (\r)}}(\r) - F \,} < \vep_{_{1}}$. Moreover, calling $\d_{_{2}} = \ep/3 M$, we have that $0 < \r < \d_{_{2}} \Rightarrow M \r < \vep_{_{1}}$. Finally, since $\{ f^{\p}_{^{\! (\r)}} \}_{_{\! \r}}$ converge uniformly in $\{ [ \, \r, d \, ] \}_{_{\! \r}}$ to $g \,$, there exists $\d_{_{3}} > 0$ such that
\beq
0 < \r < \d_{_{3}} \quad \Rightarrow \quad \babs{f^{\p}_{^{\! (\r)}}(\xi) - g(\xi)} < \vep_{_{2}} \quad , \quad \forall \, \xi \in [\, \r, d \,] \quad .
\eeq
Calling $\d = \min \{ \d_{_{1}} , \d_{_{2}} , \d_{_{3}} \}$ we conclude that
\beq
0 < \r < \d \quad \Rightarrow \babs{f_{_{\! (\r)}}(x) - f(x)} < 2 \, \vep_{_{1}} + d \, \vep_{_{2}} = \ep \quad , \quad \forall \, x \in [\, \r, d \,] \quad .
\eeq
\end{proof}

\section{Induced metric and curvature tensors}
\label{appendix ind metr and ext curv}

We discuss here the implications of our ansatz on the behaviour of the cod-1 induced metric and the curvature tensors which characterize the geometry of the cod-1 brane.

\subsection{The induced metric}
\label{appendix The induced metric}

By def\mbox{}inition, in cod-1 GNC the $\t\t$ and $\t\m$ components of cod-1 induced metric satisfy \cite{CarrollBook}
\begin{align}
\label{c1 induced metric GNC trivial k app}
\hg_{\t\t}^{_{(\r)}} &= 1 & \hg_{\t \m}^{_{(\r)}} &= 0 \\[2mm]
\label{c1 induced metric GNC trivial app}
\hg_{\t\t} &= 1 & \hg_{\t \m} &= 0 \quad ,
\end{align}
so the only non-trivial components are the $\m\n$ ones
\begin{align}
\label{c1 induced metric GNC app}
\hg_{\m\n}^{_{(\r)}} &= \dem \hvf_{^{(\r)}}^{_{A}} \,\, \den \hvf_{^{(\r)}}^{_{B}} \,\, g^{_{(\r)}}_{_{AB}} \big\rvert_{\hvfd_{^{(\r)}}} &
\hg_{\m\n} &= \dem \hvf^{_{A}} \,\, \den \hvf^{_{B}} \,\, g_{_{AB}} \big\rvert_{\hvfd} \quad ,
\end{align}
where we indicated $\dem \equiv \de_{\ch^{\m}}$ and $\den \equiv \de_{\ch^{\n}}$.

Taking into account our ansatz and using the propositions \ref{proposition 1} and \ref{proposition 2} of appendix \ref{appendix convergence}, it is straightforward to see that $\hg_{\m\n}^{_{(\r)}}$ converges uniformly to $\hg_{\m\n}$ in $\oW_{\!\!q}$, which implies that $\hg_{\m\n}$ is continuous in $\oW_{\!\!q}$. The same result holds for the higher derivatives of $\hg_{\m\n}^{_{(\r)}}$ in the parallel directions, as can be easily checked by expanding $\de_{\chd}^{\a} \, \hg_{\m\n}^{_{(\r)}}$ as in (\ref{detau hgmn}). Therefore, any derivative $\de_{\chd}^{\a} \, \hg_{\m\n}^{_{(\r)}}$ with $\abs{\a}\leq n-1$ converges uniformly to $\de_{\chd}^{\a} \, \hg_{\m\n}$ in $\oW_{\!\!q}$. In particular, at $\t$ f\mbox{}ixed $\hg_{\m\n}$ is a function of class $C^{n-1}(\oW_{\!\!q})$ of the variables $\chd$, and any derivative $\de_{\chd}^{\a} \, \hg_{\m\n}$ with $\abs{\a}\leq n-1$ is a continuous function of $\t$ and $\chd$.

Regarding the partial derivatives of f\mbox{}irst order with respect to $\t$, expanding $\de_\t \, \hg_{\m\n}^{_{(\r)}}$ as
\begin{multline}
\label{detau hgmn}
\detau \, \hg_{\m\n}^{_{(\r)}} = \detau \dem \hvf_{^{(\r)}}^{_{A}} \,\, \den \hvf_{^{(\r)}}^{_{B}} \,\, g^{_{(\r)}}_{_{AB}} \big\rvert_{\hvfd_{^{(\r)}}} + \dem \hvf_{^{(\r)}}^{_{A}} \,\, \detau \den \hvf_{^{(\r)}}^{_{B}} \,\, g^{_{(\r)}}_{_{AB}} \big\rvert_{\hvfd_{^{(\r)}}} + \\
+ \dem \hvf_{^{(\r)}}^{_{A}} \,\, \den \hvf_{^{(\r)}}^{_{B}} \,\, \detau \hvf_{^{(\r)}}^{_{L}} \, \de_{_{L}} g^{_{(\r)}}_{_{AB}} \big\rvert_{\hvfd_{^{(\r)}}}
\end{multline}
it is easy to see that it remains bounded in $\oW_{\!\!q}$, since it is a linear combination of products of bounded functions. The same is true of any derivative $\detau \de_{\chd}^{\a} \, \hg_{\m\n}^{_{(\r)}}$ with $\abs{\a}\leq n-1$. Moreover, expanding as in (\ref{detau hgmn}) the $\t$ derivative of the limit conf\mbox{}iguration $\detau \hg_{\m\n}$, it is immediate to see that it is in general discontinuous in $\t = 0$, since it contains $\de_{\t} \hvf^{_{A}}$ and $\de_{\t} \dem \hvf^{_{A}}$ which by our ansatz display that behaviour. The same is true of any derivative $\detau \de_{\chd}^{\a} \, \hg_{\m\n}^{_{(\r)}}$ with $\abs{\a}\leq n-1$. It is not dif\mbox{}f\mbox{}icult to see that the conditions (\ref{Alinete}) and (\ref{Bruna}) imply
\begin{align}
\label{Patty app}
\lim_{\r \to 0^+} \detau \, \hg_{\m\n}^{_{(\r)}} \big\vert_{\t = \r} &= \detau \, \hg_{\m\n} \big\vert_{\t = 0^+} \quad & \quad \lim_{\r \to 0^+} \detau \, \hg_{\m\n}^{_{(\r)}} \big\vert_{\t = - \r} &= \detau \, \hg_{\m\n} \big\vert_{\t = 0^-} \quad .
\end{align}

\subsection{The Einstein tensor and the extrinsic curvature}
\label{appendix Einstein tensor extrinsic curvature}

Let us now consider the Einstein tensor built from the metric induced on the cod-1 brane. As we saw above, the only partial derivative of $\hg_{\m\n}^{_{(\r)}}$ of order $\leq 2$ which does not remain bounded on $\oW_{\!\!q}$ is $\de_{\t}^{2} \, \hg_{\m\n}^{_{(\r)}}$. Keeping ($\approx$) only the terms containing $\de_{\t}^{2}$, the Ricci tensor reads
\begin{align}
\hR^{_{(\r)}}_{\t\t} &\approx - \half \,\, \hg^{\la\s}_{_{(\r)}} \, \de_{\t}^{2} \, \hg_{\la\s}^{_{(\r)}} \\[2mm]
\hR^{_{(\r)}}_{\t\m} &\approx 0 \\[2mm]
\hR^{_{(\r)}}_{\m\n} &\approx - \half \,\, \de_{\t}^{2} \, \hg_{\m\n}^{_{(\r)}} \quad ,
\end{align}
from which it is easy to see that $\hG^{_{(\r)}}_{\t\t}$ and $\hG^{_{(\r)}}_{\t\m}$ remain bounded on $\oW_{\!\!q}$. Therefore, the only components of $\hG^{_{(\r)}}_{ab}$ which can diverge are the $\m\n$ ones.

Regarding the extrinsic curvature of the cod-1 brane, the expression (\ref{c-1 extrinsic curvature ok}) for our family of conf\mbox{}igurations reads, in cod-1 GNC, as follows
\begin{align}
\hK^{_{(\r)}}_{\t\t} &= n^{_{\! (\r)}}_{_{L}} \Big( \de_{\t}^{2} \hvf^{_{L}}_{^{\!(\r)}} + \G^{_{(\r) L}}_{_{AB}}\Big\rvert_{\hvf_{^{\!(\r)}}^{\cdot}} \de_{\t} \hvf^{_{A}}_{^{\!(\r)}} \, \de_{\t} \hvf^{_{B}}_{^{\!(\r)}} \Big) \label{c-1 extrinsic curvature GNC tt app} \\[2mm]
\hK^{_{(\r)}}_{\t\m} &= n^{_{\! (\r)}}_{_{L}} \Big( \de_{\t} \de_{\m} \hvf^{_{L}}_{^{\!(\r)}} + \G^{_{(\r) L}}_{_{AB}}\Big\rvert_{\hvf_{^{\!(\r)}}^{\cdot}} \, \de_{\t} \hvf^{_{A}}_{^{\!(\r)}} \, \de_{\m} \hvf^{_{B}}_{^{\!(\r)}} \Big) \label{c-1 extrinsic curvature GNC tm app} \\[2mm]
\hK^{_{(\r)}}_{\m\n} &= n^{_{\! (\r)}}_{_{L}} \Big( \demden \, \hvf^{_{L}}_{^{\!(\r)}} + \G^{_{(\r) L}}_{_{AB}}\Big\rvert_{\hvf_{^{\!(\r)}}^{\cdot}} \dem \hvf^{_{A}}_{^{\!(\r)}} \, \den \hvf^{_{B}}_{^{\!(\r)}} \Big) \quad . \label{c-1 extrinsic curvature GNC mn app}
\end{align}
Let's f\mbox{}irst of all note that the evaluation in $\hvf_{^{\!(\r)}}^{\cdot}$ of the bulk connection coef\mbox{}f\mbox{}icients $\G^{_{(\r) L}}_{_{AB}}$ remains bounded on $\oW_{\!\!q}$ when $\r \to 0^+$, as a consequence of our ansatz and of propositions \ref{proposition 1} and \ref{proposition 2} of appendix \ref{appendix convergence}. Moreover, by our ansatz the terms $\dem \hvf^{_{A}}_{^{\!(\r)}}$, $\de_{\t} \hvf^{_{A}}_{^{\!(\r)}}$, $\demden \, \hvf^{_{A}}_{^{\!(\r)}}$ and $\de_{\t} \de_{\m} \hvf^{_{A}}_{^{\!(\r)}}$ remain bounded on $\oW_{\!\!q}$. Therefore, among the terms in round parenthesis of (\ref{c-1 extrinsic curvature GNC tt app})--(\ref{c-1 extrinsic curvature GNC mn app}), the only term which can diverge is $\de_{\t}^{2} \hvf^{_{L}}_{^{\!(\r)}}$. On the other hand, the normal 1-form $n^{_{\! (\r)}}_{_{L}}$ is def\mbox{}ined (apart from an overall sign, which encodes the choice of orientation) by the relations
\begin{align}
\label{Mariana}
n^{_{\! (\r)}}_{_{L}} \, \de_{\t} \hvf^{_{L}}_{^{\!(\r)}} &= 0 & n^{_{\! (\r)}}_{_{L}} \, \de_{\m} \hvf^{_{L}}_{^{\!(\r)}} &= 0 & g_{_{(\r)}}^{_{AB}}\big\rvert_{\hvfd_{^{\!(\r)}}} \, n^{_{\! (\r)}}_{_{A}} \, n^{_{\! (\r)}}_{_{B}} &= 1 \quad ,
\end{align}
where $\de_{\t} \hvf^{_{L}}_{^{\!(\r)}}$, $\de_{\m} \hvf^{_{L}}_{^{\!(\r)}}$ and $g_{_{(\r)}}^{_{AB}}\big\rvert_{\hvfd_{^{\!(\r)}}}$ remain bounded on $\oW_{\!\!q}$. Therefore, the same holds for $n^{_{\! (\r)}}_{_{L}}$ as well. We conclude that $\hK^{_{(\r)}}_{\t\m}$ and $\hK^{_{(\r)}}_{\m\n}$ remain bounded on $\oW_{\!\!q}$ when $\r \to 0^+$, and that the only term in $\hK^{_{(\r)}}_{\t\t}$ which does not remain bounded is $n^{_{\! (\r)}}_{_{L}} \de_{\t}^{2} \hvf^{_{L}}_{^{\!(\r)}}$.

\section{Convergence properties and integration}
\label{appendix integration}

\begin{proposition}
Let $f_{_{\! (\r)}}$ and $h_{_{(\r)}}$ be two families of continuous functions $[-l , l] \to \mbbR$, such that $f_{_{\! (\r)}}$ converges to $f$ in the norm of $C^{0} \big( [-l , l] \big)$ and that the following limits exist
\begin{align}
&i) \, I_{h} = \lim_{\r \to 0^+} \int_{-\r}^{\r} \! d\t \, \, h_{_{(\r)}} \quad , \\[2mm]
&ii) \, I_{\abs{h}} = \lim_{\r \to 0^+} \int_{-\r}^{\r} \! d\t \, \, \babs{h_{_{(\r)}}} \quad . 
\end{align}
Then $\int_{-\r}^{\r} \! d\t \, f_{_{\! (\r)}} \, h_{_{(\r)}}$ has a limit for $\r \to 0^+$, and we have
\beq
\lim_{\r \to 0^+} \int_{-\r}^{\r} \! d\t \, \, f_{_{\! (\r)}} \, h_{_{(\r)}} = f(0) \, I_{h} \quad .
\eeq
\end{proposition}
\begin{proof}
Let's f\mbox{}ix $\ep > 0$, and consider separately the subcases $f(0) = 0$ and $f(0) \neq 0$.

\noi 1) Be $f(0) = 0$. We have 
\begin{align}
\bbabs{\int_{-\r}^{\r} \! d\t \, f_{_{\! (\r)}} \, h_{_{(\r)}}} &\leq \int_{-\r}^{\r} \! d\t \, \, \babs{ f_{_{\! (\r)}}(\t) - f (0)} \, \babs{h_{_{(\r)}}(\t)} \leq \nn \\
&\leq \int_{-\r}^{\r} \! d\t \, \, \Big( \babs{ f_{_{\! (\r)}}(\t) - f (\t)} + \babs{ f(\t) - f (0)} \Big) \, \babs{h_{_{(\r)}}(\t)} \quad .
\end{align}
Since $f_{_{\! (\r)}} \to f$ uniformly, $\exists \, \d_1 > 0$ such that $0 < \r < \d_1 \Rightarrow \babs{ f_{_{\! (\r)}}(\t) - f (\t)} < \ep/2 \big( I_{\abs{h}} + 1 \big)$, $\forall \t \in [-l, l \, ]$. Since $\big( C^{0}([-l , l]), \norm{\phantom{f}}_{C^{0}([-l , l])} \big)$ is a Banach space, $f$ is continuous. Therefore, $\exists \, \d_2 > 0$ such that $\abs{\t} < \d_2 \Rightarrow \babs{ f(\t) - f (0)} < \ep/2 \big( I_{\abs{h}} + 1 \big)$. On the other hand, the hypothesis $ii)$ implies that $\exists \, \d_3 > 0$ such that $0 < \r < \d_3 \Rightarrow \int_{- \r}^{\r} d\t \, \abs{h_{_{(\r)}}} < (I_{\abs{h}} + 1)$. Calling $\d = \min \{ \d_1 , \d_2 , \d_3 \}$ we have that $0 < \r < \d$ implies
\beq
\bbabs{\int_{-\r}^{\r} \! d\t \, f_{_{\! (\r)}} \, h_{_{(\r)}}} \leq \int_{-\r}^{\r} \! d\t \,\, \frac{\ep}{I_{\abs{h}} + 1} \,\, \babs{h_{_{(\r)}}(\t)} < \ep \quad .
\eeq

\noi 2) Be $f(0) \neq 0$. We have
\begin{multline}
\bbabs{\int_{-\r}^{\r} \! d\t \, \, f_{_{\! (\r)}} \, h_{_{(\r)}} - f(0) \, I_{h}} \leq \\
\leq \bbabs{\int_{-\r}^{\r} \! d\t \, \, \Big( f_{_{\! (\r)}}(\t) - f (0)\Big) \, h_{_{(\r)}}(\t)} + \bbabs{\int_{-\r}^{\r} \! d\t \, \, f(0) \, h_{_{(\r)}}(\t) - f(0) \, I_{h}} \leq \\
\leq \int_{-\r}^{\r} \! d\t \, \, \babs{ f_{_{\! (\r)}}(\t) - f (0)} \, \babs{h_{_{(\r)}}(\t)} + \bbabs{ \, f(0) \int_{-\r}^{\r} \! d\t \, \Big( h_{_{(\r)}}(\t) - I_{h} \Big) } \leq \\
\leq \int_{-\r}^{\r} \! d\t \, \, \Big( \babs{ f_{_{\! (\r)}}(\t) - f (\t)} + \babs{ f(\t) - f (0)} \Big) \, \babs{h_{_{(\r)}}(\t)} + \babs{f(0)} \bbabs{\int_{-\r}^{\r} \! d\t \, \Big( h_{_{(\r)}}(\t) - I_{h} \Big)} \,\, . \label{important inequality}
\end{multline}
Similarly to the case 1), we can f\mbox{}ind a $\d^{\p} > 0$ such that $0 < \r < \d^{\p}$ implies 
\beq
\int_{-\r}^{\r} \! d\t \, \, \Big( \babs{ f_{_{\! (\r)}}(\t) - f (\t)} + \babs{ f(\t) - f (0)} \Big) \, \babs{h_{_{(\r)}}(\t)} < \frac{2 \, \ep}{3} \quad .
\eeq
On the other hand, the hypothesis $i)$ implies that $\exists \, \d_4 > 0$ such that $0 < \r < \d_4 \Rightarrow$
\beq
\bbabs{\int_{-\r}^{\r} \! d\t \,\, h_{_{(\r)}}(\t) - I_{h}} < \frac{\ep}{3 \, \abs{f(0)}} \quad . 
\eeq
Calling $\d = \min \{ \d^{\p}, \d_4 \}$, by the inequality (\ref{important inequality}) we have that $0 < \r < \d$ implies
\beq
\bbabs{\int_{-\r}^{\r} \! d\t \, \, f_{_{\! (\r)}} \, h_{_{(\r)}} - f(0) \, I_{h}} < \frac{2 \, \ep}{3} + \frac{\ep}{3} = \ep \quad .
\eeq
\end{proof}

\section{The auxiliary vector f\mbox{}ields}
\label{appendix auxiliary fields}

Let's consider the vector f\mbox{}ields
\begin{align}
u_{^{[\m]}}^{_{\! (\r) A}}(\t) &= \de_{\m} \,\hvf^{_{A}}_{^{(\r)}}\big\rvert_{\chd} & u_{^{[\m]}}^{_{A}}(\t) &= \de_{\m} \,\hvf^{_{A}}\big\rvert_{\chd} \quad ,
\end{align}
where it is intended that we pick a $\chd \in \mbbR^4$ and we work at $\chd$ f\mbox{}ixed. By our ansatz, $\mbf{u}_{^{[\m]}}^{_{\! (\r)}}(\t)$ is smooth and converges uniformly to $\mbf{u}_{_{[\m]}}(\t)$ which is continuous on $[-l,l]$ and smooth on $[-l,l] \setminus \{ 0 \}$. The convergence properties we impose in section \ref{Convergence properties} imply that the families $\detau u_{^{[\m]}}^{_{\! (\r) A}}(\t)$ remain bounded.

\subsection{Pseudo-orthogonalizing the vector f\mbox{}ields}

As a f\mbox{}irst step towards def\mbox{}ining the auxiliary vector f\mbox{}ields, we want to orthogonalize the vector f\mbox{}ields $\mbf{u}_{^{[\m]}}^{_{\! (\r)}}$ and $\mbf{u}_{_{[\m]}}$. This procedure, which would be completely straightforward in a Riemannian space, presents some subtleties in a pseudo-Riemannian space. In fact, the Gram-Schmidt orthogonalization procedure does not necessarily work in the latter case, since for some value of $\m$ it may generate a null vector which does not allow to carry on further the procedure. However, as we show below, it is always possible to perform a linear transformation on $\mbf{u}_{^{[\m]}}^{_{\! (\r)}}$ and $\mbf{u}_{_{[\m]}}$, such that the Gram-Schmidt procedure on the transformed vector f\mbox{}ields is well-def\mbox{}ined in a neighbourhood of $\t = 0$ for $\r$ small enough.

The 4D vector subspace generated by the vectors $\mbf{u}_{_{[\m]}}(0)$ is endowed with a non-degenerate, symmetric bilinear form of signature $(3,1)$ induced by the metric $\mbf{g}\big\rvert_{\vfd (0)}$, which can be diagonalized by a general theorem. This implies that we can f\mbox{}ind a $4 \times 4$ matrix $\mcal{D}$ such that the vectors
\beq
\mbf{v}_{_{\! [\m]}} = \mcal{D}_{\m}^{\,\,\, \n} \,\, \mbf{u}_{_{[\n]}}(0)
\eeq
are pseudo-orthonormal with respect to the metric $\mbf{g}\big\rvert_{\vfd (0)}$.\footnote{Where ``pseudo-'' refers to the fact that one of the vectors $\mbf{v}$ has negative norm.} We now def\mbox{}ine the vector f\mbox{}ields
\beq
\mbf{v}_{_{\! [\m]}}(\t) = \mcal{D}_{\m}^{\,\,\, \n} \,\, \mbf{u}_{_{[\n]}}(\t) \quad ,
\eeq
and from them we def\mbox{}ine the vector f\mbox{}ields $\mbf{w}_{_{\! [\m]}}(\t)$ by using the Gram-Schmidt orthogonalization procedure
\begin{align}
\mbf{w}_{_{\! [0]}}(\t) &= \mbf{v}_{_{\! [0]}}(\t) \label{GramSchmidt 0} \\[2mm]
\mbf{w}_{_{\! [\m]}}(\t) &= \mbf{v}_{_{\! [\m]}}(\t) - \sum_{i = 0}^{\m - 1} \, \frac{\big\langle \mbf{w}_{_{\! [i]}} , \mbf{v}_{_{\! [\m]}} \big\rangle (\t)}{\big\langle \mbf{w}_{_{\! [i]}} , \mbf{w}_{_{\! [i]}} \big\rangle (\t)} \,\, \, \mbf{w}_{_{\! [i]}}(\t) \label{GramSchmidt mu}
\end{align}
where we introduced the notation
\beq
\big\langle \mbf{a} , \mbf{b} \big\rangle (\t) = a^{_{M}}(\t) \,\, b^{_{N}}(\t) \,\, g_{_{MN}}\Big\rvert_{\vfd (\t)} \quad .
\eeq
It shoud be clear that $\mbf{w}_{_{\! [\m]}}(0) = \mbf{v}_{_{\! [\m]}}(0)$, since the latter vectors are already pseudo-orthonormal, however in general $\mbf{w}_{_{\! [\m]}}(\t) \neq \mbf{v}_{_{\! [\m]}}(\t)$ for $\t \neq 0$ (the matrix $\mcal{D}$ does not depend on $\t$).

To show that $\mbf{w}_{_{\! [\m]}}(\t)$ are well-def\mbox{}ined in a neighbourhood of $\t = 0$, let us f\mbox{}ix $\ep$ with $0< \ep < 1$ and note that $\mbf{u}_{_{[\m]}}(\t)$ are continuous by our ansatz, therefore so are $\mbf{v}_{_{\! [\m]}}(\t)$. Since $g_{_{MN}}\big\rvert_{\vfd(\t)}$ is continuous by our ansatz, $\babs{\big\langle \mbf{w}_{_{\! [0]}}, \mbf{w}_{_{\! [0]}} \big\rangle (\t)}$ is as well. Now, by construction we have that $\babs{\big\langle \mbf{w}_{_{\! [0]}}, \mbf{w}_{_{\! [0]}} \big\rangle (0)} = 1$, and so there exist a $d_{_{0}} > 0$ such that
\beq
\babs{\big\langle \mbf{w}_{_{\! [0]}}, \mbf{w}_{_{\! [0]}} \big\rangle (\t)} \geq \ep \quad \text{for} \quad \t \in [-d_{_{0}}, d_{_{0}}] \quad .
\eeq
This implies that $\mbf{w}_{_{\! [1]}}(\t)$ is well-def\mbox{}ined and continuous on $[-d_{_{0}}, d_{_{0}}]$. Since $\babs{\big\langle \mbf{w}_{_{\! [1]}}, \mbf{w}_{_{\! [1]}} \big\rangle (0)} = 1$ by construction, there exist a $d_{_{1}}$, with $0 < d_{_{1}} \leq d_{_{0}}$, such that $\babs{\big\langle \mbf{w}_{_{\! [1]}}, \mbf{w}_{_{\! [1]}} \big\rangle (\t)} \geq \ep$ for $\t \in [-d_{_{1}}, d_{_{1}}]$. %This implies that $\mbf{w}_{_{\! (2)}}(\t)$ is well-def\mbox{}ined and continuous on $[-d_{_{1}}, d_{_{1}}]$.
It is clear that we can iterate this process, f\mbox{}inding $\{ d_{_{0}}, d_{_{1}}, d_{_{2}}, d_{_{3}} \}$ with $0 < d_{_{3}} \leq d_{_{2}} \leq d_{_{1}} \leq d_{_{0}}$ such that the f\mbox{}ields $\{ \mbf{w}_{_{\! [\m]}}(\t) \}_{_{\m}}$ are well-def\mbox{}ined on $[-d_{_{3}}, d_{_{3}}]$ and
\beq
\babs{\big\langle \mbf{w}_{_{\! [\m]}}, \mbf{w}_{_{\! [\m]}} \big\rangle (\t)} \geq \ep \quad \text{for} \quad \t \in [-d_{_{3}}, d_{_{3}}] \quad \text{and} \quad \m = 0, \ldots , 3 \quad .
\eeq
It is then easy to prove that the f\mbox{}ields $\{ \mbf{w}_{_{\! [\m]}}(\t) \}_{_{\m}}$ are continuous on $[-d_{_{3}},d_{_{3}}]$ and smooth on $[-d_{_{3}} , d_{_{3}}] \setminus \{ 0 \}$, exactly as the f\mbox{}ields $\{ \mbf{u}_{_{\! [\m]}}(\t) \}_{_{\m}}$.

Let us now def\mbox{}ine the vector f\mbox{}ields
\beq
\mbf{v}_{_{\! [\m]}}^{_{(\r)}}(\t) = \mcal{D}_{\m}^{\,\,\, \n} \,\, \mbf{u}_{_{[\n]}}^{_{(\r)}}(\t) \quad ,
\eeq
and from them we def\mbox{}ine the vector f\mbox{}ields $\mbf{w}_{_{\! [\m]}}^{_{(\r)}}(\t)$ by using the Gram-Schmidt orthogonalization procedure
\begin{align}
\mbf{w}_{_{\! [0]}}^{_{(\r)}}(\t) &= \mbf{v}_{_{\! [0]}}^{_{(\r)}}(\t) \label{GramSchmidt k 0} \\[2mm]
\mbf{w}_{_{\! [\m]}}^{_{(\r)}}(\t) &= \mbf{v}_{_{\! [\m]}}^{_{(\r)}}(\t) - \sum_{i = 0}^{\m - 1} \, \frac{\big\langle \mbf{w}_{_{\! [i]}}^{_{(\r)}} , \mbf{v}_{_{\! [\m]}}^{_{(\r)}} \big\rangle_{_{\!\! \r}} (\t)}{\big\langle \mbf{w}_{_{\! [i]}}^{_{(\r)}} , \mbf{w}_{_{\! [i]}}^{_{(\r)}} \big\rangle_{_{\!\! \r}} (\t)} \,\, \, \mbf{w}_{_{\! [i]}}^{_{(\r)}}(\t) \label{GramSchmidt k mu} \quad ,
\end{align}
where we introduced the notation
\beq
\big\langle \mbf{a} , \mbf{b} \big\rangle_{_{\!\! \r}} (\t) = a^{_{M}}(\t) \,\, b^{_{N}}(\t) \,\, g_{_{MN}}^{_{(\r)}}\Big\rvert_{\vfd_{(\r)}(\t)} \quad .
\eeq
By our ansatz, the smooth f\mbox{}ields $\mbf{u}_{^{[\m]}}^{_{\! (\r)}}(\t)$ converge uniformly in $[-l, l \, ]$ to the f\mbox{}ields $\mbf{u}_{_{[\m]}}(\t)$, with bounded $\de_{\t}$ derivative. It is immediate to see that the same holds for $\mbf{v}_{_{\! [\m]}}^{_{(\r)}}(\t) \to \mbf{v}_{_{\! [\m]}}(\t)$. The propositions \ref{proposition 1} and \ref{proposition 2} of appendix \ref{appendix convergence} then imply that $\babs{\big\langle \mbf{v}_{_{\! [\m]}}^{_{(\r)}} ,\mbf{v}_{_{\! [\m]}}^{_{(\r)}} \big\rangle_{_{\!\!\r}} (\t)}$ converges uniformly to $\babs{\big\langle \mbf{v}_{_{\! [\m]}} ,\mbf{v}_{_{\! [\m]}} \big\rangle (\t)}$ in $[-l, l \, ]$ for $\m = 0 , \ldots , 3 \,$. It follows that there exist a $\d_{_{0}} > 0$ such that
\beq
0 < \r \leq \d_{_{0}} \,\, \Rightarrow \,\, \Babs{\big\langle \mbf{w}_{_{\! [0]}}^{_{(\r)}} ,\mbf{w}_{_{\! [0]}}^{_{(\r)}} \big\rangle_{_{\!\! \r}} (\t)} \geq \frac{\ep}{2} \qquad \text{for} \quad \t \in [- d_{_{0}}, d_{_{0}}] \quad ,
\eeq
which in turn implies that $\mbf{w}_{_{\! [1]}}^{_{(\r)}}(\t)$ is well-def\mbox{}ined and continuous on $[- d_{_{0}}, d_{_{0}}]$ for $\r \leq \d_{_{0}}$. Moreover, taking into account our ansatz, it can be seen that $\mbf{w}_{_{\! [1]}}^{_{(\r)}}(\t)$ is in fact derivable and has bounded derivative in $[- d_{_{0}}, d_{_{0}}]$ for $\r \leq \d_{_{0}}$. Furthermore, by propositions \ref{proposition 1}, \ref{proposition 2} and \ref{proposition 3} of appendix \ref{appendix convergence} it follows that $\mbf{w}_{_{\! [1]}}^{_{(\r)}}(\t)$ converges uniformly in $[- d_{_{0}}, d_{_{0}}]$ to $\mbf{w}_{_{\! [1]}}(\t) \,$. Also in this case we can iterate the procedure, f\mbox{}inding positive numbers $\d_{_{3}} \leq \d_{_{2}} \leq \d_{_{1}} \leq \d_{_{0}}$ such that
\beq
0 < \r \leq \d_{_{3}} \,\, \Rightarrow \,\, \Babs{\big\langle \mbf{w}_{_{\! [\m]}}^{_{(\r)}} ,\mbf{w}_{_{\! [\m]}}^{_{(\r)}} \big\rangle_{_{\!\! \r}} (\t)} \geq \frac{\ep}{2} \qquad \text{for} \quad \t \in [- d_{_{3}}, d_{_{3}}] \quad \text{and} \quad \m = 0, \ldots , 3 \quad ,
\eeq
which implies that the f\mbox{}ields $\mbf{w}_{_{\! [\m]}}^{_{(\r)}}(\t)$ are well-def\mbox{}ined in $[- d_{_{3}}, d_{_{3}}]$ for $\r \leq \d_{_{3}}$. It is again easy to prove that the f\mbox{}ields $\mbf{w}_{_{\! [\m]}}^{_{(\r)}}(\t)$, which as we mentioned above converge uniformly to the f\mbox{}ields $\mbf{w}_{_{\! [\m]}}(\t)$, are smooth in $[- d_{_{3}}, d_{_{3}}]$ for $\r \leq \d_{_{3}}$ and their f\mbox{}irst derivative remains bounded.

\subsection{Constructing the auxiliary f\mbox{}ields}
\label{Constructing the auxiliary fields}

We can now use the f\mbox{}ields $\mbf{w}_{_{\! [\m]}}(\t)$ and the families $\mbf{w}_{_{\! [\m]}}^{_{(\r)}}(\t)$ to construct the auxiliary vector f\mbox{}ields. Let's complete the set $\{ \mbf{w}_{_{\! [\m]}}(0) \}_{_{\m}}$ to a pseudo-orthonormal basis by means of two vectors $\mbf{r}$ and $\mbf{s}$ (which we remark are independent of $\t$), and def\mbox{}ine
\begin{align}
\mbf{Y}(\t) &= \mbf{r} - \sum_{i = 0}^{3} \, \frac{\big\langle \mbf{w}_{_{\! [i]}} , \mbf{r} \big\rangle (\t)}{\big\langle \mbf{w}_{_{\! [i]}} , \mbf{w}_{_{\! [i]}} \big\rangle (\t)} \,\, \, \mbf{w}_{_{\! [i]}}(\t) \label{def Y} \\[2mm]
\mbf{Y}_{_{\!\! (\r)}}(\t) &= \mbf{r} - \sum_{i = 0}^{3} \, \frac{\big\langle \mbf{w}_{_{\! [i]}}^{_{(\r)}} , \mbf{r} \big\rangle_{_{\!\! \r}} (\t)}{\big\langle \mbf{w}_{_{\! [i]}}^{_{(\r)}} , \mbf{w}_{_{\! [i]}}^{_{(\r)}} \big\rangle_{_{\!\! \r}} (\t)} \,\, \, \mbf{w}_{_{\! [i]}}^{_{(\r)}}(\t) \label{def Yk} \quad .
\end{align}
It is obvious from the construction above that the vector f\mbox{}ields $\mbf{Y}(\t)$ and $\mbf{Y}_{_{\!\! (\r)}}(\t)$ are well def\mbox{}ined on $[- d_{_{3}}, d_{_{3}} ]$ for $\r \leq \d_{_{3}}$. Performing an analysis analogous to the one performed above, we conclude that there exists a $d_{_{4}} > 0$ and a $\d_{_{4}} > 0$, with $d_{_{4}} \leq d_{_{3}}$ and $\d_{_{4}} \leq \d_{_{3}}$, such that
\beq
\big\langle \mbf{Y} , \mbf{Y} \big\rangle (\t) \geq \ep \quad \text{and} \quad \big\langle \mbf{Y}_{_{\!\! (\r)}} , \mbf{Y}_{_{\!\! (\r)}} \big\rangle_{_{\!\! \r}} (\t) \geq \ep/2 \quad \text{for} \quad \t \in [-d_{_{4}}, d_{_{4}}] \quad \text{and} \quad \r \leq \d_{_{4}} \quad .
\eeq
This in turn allows to construct 
\begin{align}
\mbf{Z}(\t) &= \mbf{s} - \sum_{i = 0}^{3} \, \frac{\big\langle \mbf{w}_{_{\! [i]}} , \mbf{s} \big\rangle (\t)}{\big\langle \mbf{w}_{_{\! [i]}} , \mbf{w}_{_{\! [i]}} \big\rangle (\t)} \,\, \, \mbf{w}_{_{\! [i]}}(\t) - \frac{\big\langle \mbf{Y} , \mbf{s} \big\rangle (\t)}{\big\langle \mbf{Y} , \mbf{Y} \big\rangle (\t)} \,\, \, \mbf{Y}(\t) \label{def Z} \\[2mm]
\mbf{Z}_{_{(\r)}}(\t) &= \mbf{s} - \sum_{i = 0}^{3} \, \frac{\big\langle \mbf{w}_{_{\! [i]}}^{_{(\r)}} , \mbf{s} \big\rangle_{_{\!\! \r}} (\t)}{\big\langle \mbf{w}_{_{\! [i]}}^{_{(\r)}} , \mbf{w}_{_{\! [i]}}^{_{(\r)}} \big\rangle_{_{\!\! \r}} (\t)} \,\, \, \mbf{w}_{_{\! [i]}}^{_{(\r)}}(\t) - \frac{\big\langle \mbf{Y}_{_{\!\! (\r)}} , \mbf{s} \big\rangle_{_{\!\! \r}} (\t)}{\big\langle \mbf{Y}_{_{\!\! (\r)}} , \mbf{Y}_{_{\!\! (\r)}} \big\rangle_{_{\!\! \r}} (\t)} \,\, \, \mbf{Y}_{_{\!\! (\r)}}(\t) \label{def Zk} \quad ,
\end{align}
and again we can f\mbox{}ind a $d > 0$ and a $\bd > 0$, with $d \leq d_{_{4}}$ and $\bd \leq \d_{_{4}}$, such that 
\beq
\big\langle \mbf{Z} , \mbf{Z} \big\rangle (\t) \geq \ep \quad \text{and} \quad \big\langle \mbf{Z}_{_{(\r)}} , \mbf{Z}_{_{(\r)}} \big\rangle_{_{\!\! \r}} (\t) \geq \ep/2 \quad \text{for} \quad \t \in [-d, d \, ] \quad \text{and} \quad \r \leq \bd \quad .
\eeq
It is easy to see that the vector f\mbox{}ields $\mbf{Y}_{_{\!\! (\r)}}$ and $\mbf{Z}_{_{\! (\r)}}$ are smooth on $[-d, d \, ]$, that the vector f\mbox{}ields $\mbf{Y}$ and $\mbf{Z}$ are continous on $[-d, d \, ]$ and smooth on $[-d, d \, ] \setminus \{ 0 \}$, and that the families $\{ \mbf{Y}_{_{\!\! (\r)}} \}_{_{\! \r}}$ and $\{ \mbf{Z}_{_{\! (\r)}} \}_{_{\! \r}}$, by the propositions \ref{proposition 1}, \ref{proposition 2} and \ref{proposition 3} of appendix \ref{appendix convergence}, converge uniformly and with bounded derivative respectively to $\mbf{Y}$ and $\mbf{Z} \,$.

Finally, we introduce
\begin{align}
\mbf{y}(\t) &= \frac{\mbf{Y}(\t)}{\sqrt{\babs{\big\langle \mbf{Y} , \mbf{Y} \big\rangle (\t)}}} & \mbf{z}(\t) &= \frac{\mbf{Z}(\t)}{\sqrt{\babs{\big\langle \mbf{Z} , \mbf{Z} \big\rangle (\t)}}} \label{def y and z} \\[2mm]
\mbf{y}_{_{\!\! (\r)}}(\t) &= \frac{\mbf{Y}_{_{\!\! (\r)}}(\t)}{\sqrt{\babs{\big\langle \mbf{Y}_{_{\!\! (\r)}} , \mbf{Y}_{_{\!\! (\r)}} \big\rangle_{_{\!\!\r}} (\t)}}} & \mbf{z}_{_{(\r)}}(\t) &= \frac{\mbf{Z}_{_{(\r)}}(\t)}{\sqrt{\babs{\big\langle \mbf{Z}_{_{(\r)}} , \mbf{Z}_{_{(\r)}} \big\rangle_{_{\!\!\r}} (\t)}}} \label{def yk and zk} 
\end{align}
which are clearly well-def\mbox{}ined for $\t \in [-d, d \, ]$ and $\r \leq \bd \,$. By construction, $\mbf{y}(\t)$ and $\mbf{z}(\t)$ are orthonormal, and they are both orthogonal to the vector f\mbox{}ields $\{ \mbf{u}_{_{[\m]}}(\t) \}_{\m}$. Analogously, for every $\r \leq \bd$ the vector f\mbox{}ields $\mbf{y}_{_{\!\! (\r)}}(\t)$ and $\mbf{z}_{_{(\r)}}(\t)$ are orthonormal, and they are both orthogonal to the vector f\mbox{}ields $\{ \mbf{u}_{_{[\m]}}^{_{\! (\r)}}(\t) \}_{\m}$. Again, it can be seen that $\mbf{y}(\t)$ and $\mbf{z}(\t)$ are continuous on $[-d, d \, ]$ and smooth on $[-d , d \, ] \setminus \{ 0 \}$, and that $\mbf{y}_{_{\!\! (\r)}}(\t)$ and $\mbf{z}_{_{(\r)}}(\t)$ are smooth on $[-d, d \, ]$ for $\r \leq \bd$ and converge uniformly with bounded derivative respectively to $\mbf{y}(\t)$ and $\mbf{z}(\t)$ on $[-d, d \, ] \,$. Therefore, the families $\{ \mbf{y}_{_{\!\! (\r)}} \}_{_{\! \r}}$, $\{ \mbf{z}_{_{(\r)}} \}_{_{\! \r}}$ and the limit conf\mbox{}igurations $\mbf{y}$, $\mbf{z}$ satisfy the required properties of the auxiliary vector f\mbox{}ields.

Note that the choice of auxiliary vector f\mbox{}ields is not unique. For example, performing a rigid 2D rotation (that is, independent of $\t$ and $\r$) on the frames $\big\{ \big(\mbf{y}_{_{\!\! (\r)}}(\t), \mbf{z}_{_{(\r)}}(\t) \big) \big\}_{_{\!\r}}$ and $\big( \mbf{y}(\t) , \mbf{z}(\t) \big)$ we obtain another couple of families and limit conf\mbox{}igurations which have exactly the same properties of the old ones, and therefore as equally good as auxiliary vector f\mbox{}ields.

\section{The slope function}
\label{appendix slope function}

We discuss here the properties of the functions $S_{_{\! (\r)}}$ and $S$, introduced in section \ref{The slope function}. As we do in the main text, we drop the dependence on $\chd$. By def\mbox{}inition $S_{_{\! (\r)}}$ and $S$ satisfy the relations (\ref{auxiliary 1})--(\ref{auxiliary 2 lim}), from which, using the orthonormality of the auxiliary f\mbox{}ields, it is easy to get
\begin{align}
\cos \big(S (\t)\big) &= \big\langle \mbf{y} , \mbfscn \big\rangle (\t) & \sin \big(S (\t)\big) &= \big\langle \mbf{z} , \mbfscn \big\rangle (\t) \label{cos S sin S} \\[3mm]
\cos \big(S_{_{\! (\r)}} (\t)\big) &= \big\langle \mbf{y}_{_{\!\! (\r)}} , \mbfscn_{_{\! (\r)}} \big\rangle_{_{\!\! \r}} (\t) & \sin \big(S_{_{\! (\r)}} (\t)\big) &= \big\langle \mbf{z}_{_{(\r)}} , \mbfscn_{_{\! (\r)}} \big\rangle_{_{\!\! \r}} (\t) \label{cos Sk sin Sk} \quad ,
\end{align}
where the relations (\ref{cos S sin S}) hold for $\t \in [-d , d \, ] \setminus \{0\}$ while the relations (\ref{cos Sk sin Sk}) hold for $\t \in [-d , d \, ]$.

\subsection{Existence of smooth solutions}
\label{Existence of smooth solutions}

It is clear that the equations (\ref{cos S sin S}) and (\ref{cos Sk sin Sk}) always admit solutions, since we have 
\beq
\label{Lucia 2}
\Big( \big\langle \mbf{y} , \mbfscn \big\rangle (\t) \Big)^2 + \Big( \big\langle \mbf{z} , \mbfscn \big\rangle (\t) \Big)^2 = 1
\eeq
as a consequence of $\mbfscn$ being normalized and $\big( \mbf{y} , \mbf{z} \big)$ being orthonormal, and analogously we have
\beq
\label{Lucia}
\Big( \big\langle \mbf{y}_{_{\!\! (\r)}} , \mbfscn_{_{\! (\r)}} \big\rangle_{_{\!\! \r}} (\t) \Big)^2 + \Big( \big\langle \mbf{z}_{_{(\r)}} , \mbfscn_{_{\! (\r)}} \big\rangle_{_{\!\! \r}} (\t) \Big)^2 = 1 \quad .
\eeq
For example, f\mbox{}ixing $\t \in [-d , d \, ] \setminus \{0\} \,$, the solutions of (\ref{cos S sin S}) in $\t$ are of the form
\beq
\z\vert_{\t} = \bz\vert_{\t} + 2 \pi \, \bk \quad ,
\eeq
with $\bk \in \mbbZ \,$, where $\bz\vert_{\t}$ is a particular solution of (\ref{cos S sin S}) in $\t \,$. An analogous statement holds for the solutions of (\ref{cos Sk sin Sk}) at a f\mbox{}ixed $\t \in [-d , d \, ] \,$. The point is that the inf\mbox{}initely many functions $S(\t)$ and $S_{_{\! (\r)}}(\t)$ which solve (\ref{cos S sin S})--(\ref{cos Sk sin Sk}) can be very irregular (they even can be nowhere continuous). On the other hand, for our purpouses it is crucial the functions $S_{_{\! (\r)}}(\t)$ to be derivable, otherwise we could not use the relation (\ref{Heydouga}) to obtain the f\mbox{}irst equality in (\ref{integral S}), %express $\mscrI$ as in (\ref{integral Sk}),
and the family $\big\{ S_{_{\! (\r)}} \big\}_{_{\!\! \r}}$ to converge uniformly to $S$ in $\{ [-d , -\r \,] \cup [\, \r , d \,] \}_{_{\! \r}}\,$, in order to obtain the second equality in (\ref{integral S}).

Let us see f\mbox{}irst that it is indeed possible to f\mbox{}ind smooth solutions of (\ref{cos S sin S}) in $[-d , d \, ] \setminus \{0\}$ and of (\ref{cos Sk sin Sk}) in $[-d , d \, ]$. Def\mbox{}ining the functions
\begin{align}
f^{(c)}(\z, \t) &= \cos \z - \big\langle \mbf{y} , \mbfscn \big\rangle (\t) & f^{(s)}(\z, \t) &= \sin \z - \big\langle \mbf{z} , \mbfscn \big\rangle (\t) \label{fc and fs}
\end{align}
and letting $\z_{_{-}}$ be a solution of (\ref{cos S sin S}) in $\t = -d$, the implicit function theorem (see, for example, \cite{Rudin:MathAnalysis}) assures that there exist a unique smooth solution of (\ref{cos S sin S}) in a neighbourhood of $\t = -d$ which assumes the value $\z_{_{-}}$ in $\t = -d$. Since $\de_{\z} \, f^{(c)} = - \sin \z$ and $\de_{\z} \, f^{(s)} = \cos \z$ are never both vanishing, we can repeatedly use the implicit function theorem (if necessary) to extend the solution until $\t = 0$ (excluded), where $\mbfscn$ is not def\mbox{}ined. We then obtain a unique smooth solution $\s_{_{\! -}}: [-d, 0) \to \mbbR$ of (\ref{cos S sin S}) which satisf\mbox{}ies $\s_{_{\! -}}(-d) = \z_{_{-}}$. Similarly, letting $\z_{_{+}}$ be a solution of (\ref{cos S sin S}) in $\t = d$, we obtain a unique smooth solution $\s_{_{\! +}}: (0, d \, ] \to \mbbR$ of (\ref{cos S sin S}) which satisf\mbox{}ies $\s_{_{\! +}}(d) = \z_{_{+}}$.

Note that, by our ansatz, the right hand sides of (\ref{cos S sin S}) have f\mbox{}inite limits for $\t \to 0^-$ and for $\t \to 0^+$. Since $\cos$ and $\sin$ are continuous functions, it follows that $\s_{_{\! -}}$ has a f\mbox{}inite limit for $\t \to 0^-$, and that $\s_{_{\! +}}$ has a f\mbox{}inite limit for $\t \to 0^+$, and indeed
\begin{align}
\label{Layla Rose}
\cos \Big(\lim_{\t \to 0^{\pm}} \s_{_{\! \pm}} \Big) &= \big\langle \mbf{y}(0) , \mbfscn_{_{\pm}} \big\rangle & \sin \Big( \lim_{\t \to 0^{\pm}} \s_{_{\! \pm}} \Big) &= \big\langle \mbf{z}(0) , \mbfscn_{_{\pm}} \big\rangle \quad .
\end{align}
Due to the discontinuity of $\mbfscn$ in $\t = 0$, in general we have $\lim_{\t \to 0^{-}} \s_{_{\! -}} \neq \lim_{\t \to 0^{+}} \s_{_{\! +}}$ for any choice of $\z_{_{-}}$ and $\z_{_{+}}$. Moreover, taking the derivative of the equations (\ref{cos S sin S}) we get for $\t \neq 0$
\begin{equation}
\label{Becca Bali}
\detau \, \s_{_{\! \pm}} (\t) =
\begin{cases}
- \dfrac{\detau \, \big\langle \mbf{y} , \mbfscn \big\rangle (\t)}{\big\langle \mbf{z} , \mbfscn \big\rangle (\t)} \quad , \quad \text{when} \,\, \big\langle \mbf{z} , \mbfscn \big\rangle (\t) \neq 0 \\[6mm]
\phantom{-}\dfrac{\detau \, \big\langle \mbf{z} , \mbfscn \big\rangle (\t)}{\big\langle \mbf{y} , \mbfscn \big\rangle (\t)} \quad , \quad \text{when} \,\, \big\langle \mbf{y} , \mbfscn \big\rangle (\t) \neq 0 \quad .
\end{cases}
\end{equation}
where the two expressions on the right hand side are compatible as a consequence of (\ref{Lucia 2}). These expressions are consistent since the relation (\ref{Lucia 2}) implies that the denominators of the fractions in (\ref{Becca Bali}) cannot vanish at the same time, and they vanish if and only if the correspondent numerator vanishes. Indeed $\detau \, \s_{_{\! -}}$ is well-def\mbox{}ined on $[-d , 0 )$ and $\detau \, \s_{_{\! +}}$ is well-def\mbox{}ined on $( 0 , d \,] \,$. Since $\mbfscn$, $\mbf{y}$ and $\mbf{z}$ have f\mbox{}inite limits for $\t \to 0^{\pm}$, it follows that $\lim_{\t \to 0^{-}} \detau \, \s_{_{\! -}}$ and $\lim_{\t \to 0^{+}} \detau \, \s_{_{\! +}}$ exist f\mbox{}inite.

Regarding the equations (\ref{cos Sk sin Sk}), from any solution $\z_{_{-}}^{_{(\r)}}$ of (\ref{cos Sk sin Sk}) in $\t = -d$ we can construct a unique smooth solution $\s_{_{\!\! (\r)}} : [-d , d \, ] \to \mbbR$ of (\ref{cos Sk sin Sk}) which satisf\mbox{}ies $\s_{_{\!\! (\r)}}(-d) = \z_{_{-}}^{_{(\r)}}$, following a procedure analogous to the one used above. Taking the derivative of the equations (\ref{cos Sk sin Sk}), we obtain
\begin{equation}
\label{Shion}
\detau \s_{_{\!\! (\r)}} (\t) =
\begin{cases}
- \dfrac{\detau \big\langle \mbf{y}_{_{\!\! (\r)}} , \mbfscn_{_{\! (\r)}} \big\rangle_{_{\!\! \r}} (\t)}{\big\langle \mbf{z}_{_{(\r)}} , \mbfscn_{_{\! (\r)}} \big\rangle_{_{\!\! \r}} (\t)} \quad , \quad \text{when} \,\, \big\langle \mbf{z}_{_{(\r)}} , \mbfscn_{_{\! (\r)}} \big\rangle_{_{\!\! \r}} (\t) \neq 0 \\[6mm]
\phantom{-}\dfrac{\detau \big\langle \mbf{z}_{_{(\r)}} , \mbfscn_{_{\! (\r)}} \big\rangle_{_{\!\! \r}} (\t)}{\big\langle \mbf{y}_{_{\!\! (\r)}} , \mbfscn_{_{\! (\r)}} \big\rangle_{_{\!\! \r}} (\t)} \quad , \quad \text{when} \,\, \big\langle \mbf{y}_{_{\!\! (\r)}} , \mbfscn_{_{\! (\r)}} \big\rangle_{_{\!\! \r}} (\t) \neq 0 \quad .
\end{cases}
\end{equation}
where again the two expressions for $\detau \s_{_{\!\! (\r)}}$ are compatible as a consequence of (\ref{Lucia}). Also in this case, the denominators of the fractions in (\ref{Shion}) cannot vanish at the same time and they vanish if and only if the correspondent numerator vanishes. It follows that $\detau \s_{_{\!\! (\r)}}$ is well-def\mbox{}ined on $[-d , d \,]$ for every value of $\r\,$.

Note that from (\ref{Becca Bali}) and (\ref{Shion}) it follows that the choice of $\z_{_{-}}$, $\z_{_{+}}$ and $\z_{_{-}}^{_{(\r)}}$ has no inf\mbox{}luence at all on the derivatives $\detau \, \s_{_{\! \pm}}$ and $\detau \, \s_{_{\!\! (\r)}} \,$, since the latters are determined uniquely by the auxiliary vector f\mbox{}ields, by the tangent vectors $\mbfscn$, $\mbfscn_{_{\! (\r)}}$ and by the metric $\mbf{g}^{_{(\r)}}\vert_{\hvfd} \,$.

\subsection{Selecting a family of smooth solutions}

As we saw above, the smooth solutions of (\ref{cos S sin S}) and (\ref{cos Sk sin Sk}) are in a one-to-one correspondence with the solutions $\z_{_{-}}$ and $\z_{_{+}}$ of (\ref{cos S sin S}) respectively in $\t = -d$ and in $\t = d \,$, and with the family $\{ \z_{_{-}}^{_{(\r)}} \}_{_{\! \r}}$ of solutions of (\ref{cos Sk sin Sk}) in $\t = -d$. We now show that it is possible to choose %the family $\{ \s_{_{\! (\r)}}(\t) \}_{_{\! \r}}$ and the functions $\s_{_{\! -}}(\t)$ and $\s_{_{\! +}}(\t)$ (or equivalently
$\z_{_{-}}$, $\z_{_{+}}$ and $\{ \z_{_{-}}^{_{(\r)}} \}_{_{\! \r}}$ such that $\big\{ \s_{_{\!\! (\r)}} \big\}_{_{\!\! \r}}$ converges to $\s_{_{\!-}}$ uniformly in $\{ [-d , -\r \,] \}_{_{\! \r}} \,$, and converges to $\s_{_{\!+}}$ uniformly in $\{ [\, \r , d \,] \}_{_{\! \r}}\,$.

Let us then f\mbox{}ix $\bz_{_{-}}$, to which is associated the unique smooth solution $\bars_{_{\! -}}: [-d, 0) \to \mbbR$ of (\ref{cos S sin S}) which satisf\mbox{}ies $\bars_{_{\! -}}(-d) = \bz_{_{-}} \,$, and consider the interval $J_{_{-}} = \big[ \bz_{_{-}} - \pi, \bz_{_{-}} + \pi \big)$. For every value of $\r$ we indicate with $\bz_{_{-}}^{_{(\r)}}$ the solution of (\ref{cos Sk sin Sk}) in $\t = -d$ such that $\bz_{_{-}}^{_{(\r)}} \in J_{_{-}} \,$. Since
\begin{align}
\lim_{\r \to 0^+} \big\langle \mbf{y}_{_{\!\! (\r)}} , \mbfscn_{_{\! (\r)}} \big\rangle_{_{\!\! \r}} (-d) &= \big\langle \mbf{y} , \mbfscn \big\rangle (-d) \phantom{aa} & \phantom{aa} \lim_{\r \to 0^+} \big\langle \mbf{z}_{_{(\r)}} , \mbfscn_{_{\! (\r)}} \big\rangle_{_{\!\! \r}} (-d) &= \big\langle \mbf{z} , \mbfscn \big\rangle (-d) \quad ,
\end{align}
we get
\begin{align}
\lim_{\r \to 0^+} \cos \bz_{_{-}}^{_{(\r)}} &= \cos \bz_{_{-}} &  \lim_{\r \to 0^+} \sin \bz_{_{-}}^{_{(\r)}} &= \sin \bz_{_{-}}
\end{align}
and it follows that
\begin{proposition}
\label{proposition 5}
The family $\{ \bz_{_{-}}^{_{(\r)}} \}_{_{\! \r}}$ converges to $\bz_{_{-}}$.
\end{proposition}
\begin{proof}
We have
\begin{align}
\lim_{\r \to 0^+} \cos \big(\bz_{_{-}}^{_{(\r)}} - \bz_{_{-}}\big) &= \lim_{\r \to 0^+} \Big[ \cos \big(\bz_{_{-}}^{_{(\r)}}\big) \cos \big(\bz_{_{-}}\big) + \sin \big(\bz_{_{-}}^{_{(\r)}}\big) \sin \big(\bz_{_{-}}\big) \Big] = 1 \\[2mm]
\lim_{\r \to 0^+} \sin \big(\bz_{_{-}}^{_{(\r)}} - \bz_{_{-}}\big) &= \lim_{\r \to 0^+} \Big[ \sin \big(\bz_{_{-}}^{_{(\r)}}\big) \cos \big(\bz_{_{-}}\big) - \cos \big(\bz_{_{-}}^{_{(\r)}}\big) \sin \big(\bz_{_{-}}\big) \Big] = 0 \quad .
\end{align}
Calling $y^{_{(\r)}}_{_{-}} = \bz_{_{-}}^{_{(\r)}} - \bz_{_{-}}$, we get
\begin{align}
-\pi &\leq y^{_{(\r)}}_{_{-}} \leq \pi & \lim_{\r \to 0^+} \cos \big(y^{_{(\r)}}_{_{-}} \big) &= 1 & \lim_{\r \to 0^+} \sin \big(y^{_{(\r)}}_{_{-}} \big) &= 0
\end{align}
which implies $\lim_{\r \to 0^+} y^{_{(\r)}}_{_{-}} = 0$ and therefore $\lim_{\r \to 0^+} \bz_{_{-}}^{_{(\r)}} = \bz_{_{-}}\, $.
\end{proof}

Let us now call $\{ \bars_{_{\!\! (\r)}} \}_{_{\r}}$ the unique family of smooth solutions $[-d , d \, ] \to \mbbR$ of (\ref{cos Sk sin Sk}) which satisfy $\bars_{_{\!\! (\r)}}(-d) = \bz_{_{-}}^{_{(\r)}}$. We have
\begin{proposition}
\label{proposition 6}
The family of functions $\{ \detau \, \bars_{_{\!\! (\r)}} \}_{_{\! \r}}$ converges uniformly in $\{ [-d , -\r] \}_{_{\! \r}}$ to $\detau \, \bars_{_{\! -}}$.
\end{proposition}
\begin{proof}
Let us introduce the sets
\beqnn
L_{y} = \{ \t \in [-d , 0 \,] : \babs{\big\langle \mbf{y} , \mbfscn \big\rangle (\t)} > 1/2 \} \quad , \quad L_{z} = \{ \t \in [-d , 0 \,] : \babs{\big\langle \mbf{z} , \mbfscn \big\rangle (\t)} > 1/2 \} \quad .
\eeqnn
By the relation (\ref{Lucia 2}), it follows that $L_{y} \, \cup \, L_{z} = [-d , 0 \,]\,$. Note that the families $\big\langle \mbf{y}_{_{\!\! (\r)}} , \mbfscn_{_{\! (\r)}} \big\rangle_{_{\!\! \r}}$ and $\detau \big\langle \mbf{y}_{_{\!\! (\r)}} , \mbfscn_{_{\! (\r)}} \big\rangle_{_{\!\! \r}}$ converge uniformly in $\{ [-d , -\r \,] \}_{_{\! \r}}$ respectively to $\big\langle \mbf{y} , \mbfscn \big\rangle$ and $\detau \big\langle \mbf{y} , \mbfscn \big\rangle$, and $\inf_{L_{y}} \babs{\big\langle \mbf{y} , \mbfscn \big\rangle} > 0 \,$. Analogously, the families $\big\langle \mbf{z}_{_{(\r)}} , \mbfscn_{_{\! (\r)}} \big\rangle_{_{\!\! \r}}$ and $\detau \big\langle \mbf{z}_{_{(\r)}} , \mbfscn_{_{\! (\r)}} \big\rangle_{_{\!\! \r}}$ converge uniformly in $\{ [-d , -\r \,] \}_{_{\! \r}}$ respectively to $\big\langle \mbf{z} , \mbfscn \big\rangle$ and $\detau \big\langle \mbf{z} , \mbfscn \big\rangle$, and $\inf_{L_{z}} \babs{\big\langle \mbf{z} , \mbfscn \big\rangle} > 0 \,$. By lemma \ref{lemma 1 Qr} and propositions \ref{proposition 1 Qr}--\ref{proposition 3 Qr} it follows that
\beq
\dfrac{\detau \big\langle \mbf{y}_{_{\!\! (\r)}} , \mbfscn_{_{\! (\r)}} \big\rangle_{_{\!\! \r}} }{\big\langle \mbf{z}_{_{(\r)}} , \mbfscn_{_{\! (\r)}} \big\rangle_{_{\!\! \r}}} \quad \text{converges to} \quad \dfrac{\detau \big\langle \mbf{y} , \mbfscn \big\rangle}{\big\langle \mbf{z} , \mbfscn \big\rangle} \quad \text{uniformly in} \quad \{ L_{y} \cap [-d , -\r \,] \}_{_{\! \r}} \quad ,
\eeq
and that
\beq
\dfrac{\detau \big\langle \mbf{z}_{_{(\r)}} , \mbfscn_{_{\! (\r)}} \big\rangle_{_{\!\! \r}} }{\big\langle \mbf{y}_{_{\!\! (\r)}} , \mbfscn_{_{\! (\r)}} \big\rangle_{_{\!\! \r}}} \quad \text{converges to} \quad \dfrac{\detau \big\langle \mbf{z} , \mbfscn \big\rangle}{\big\langle \mbf{y} , \mbfscn \big\rangle} \quad \text{uniformly in} \quad \{ L_{z} \cap [-d , -\r \,] \}_{_{\! \r}} \quad .
\eeq
Since $\{ L_{y} \cap [-d , -\r \,] \} \cup \{ L_{z} \cap [-d , -\r \,] \} = [-d , -\r \,]\,$, the thesis follows from the relations (\ref{Becca Bali}) and (\ref{Shion}).
\end{proof}

We have therefore the following situation. The family of derivatives $\{ \detau \, \bars_{_{\!\! (\r)}} \}_{_{\! \r}}$ converges uniformly in $\{ [-d , -\r] \}_{_{\! \r}}$ to $\detau \, \bars_{_{\! -}}$, and $\{ \bars_{_{\!\! (\r)}}(-d) \}_{_{\! \r}} = \bz_{_{-}}^{_{(\r)}}$ converges \'a la Cauchy. Then by the proposition \ref{proposition slope minus} the family $\{ \bars_{_{\!\! (\r)}} \}_{_{\! \r}}$ converges uniformly in $\{ [-d , -\r] \}_{_{\! \r}}$ to a function $\bs_{_{\!-}} : [-d , 0 \,) \to \mbbR \,$, and $\detau \, \bs_{_{\!-}} = \detau \, \bars_{_{\! -}}$. Since $\bs_{_{\!-}}$ and $\bars_{_{\! -}}$ coincide in $\t = -d \,$, they coincide over $[-d , 0 \,) \,$. %$\{ \bars_{_{\!\! (\r)}}(-d) \}_{_{\! \r}}$ converges to $\s_{_{-}}(-d)$, it follows that $s_{_{-}}(-d) = \s_{_{-}}(-d)$, which in turn implies that $s_{_{-}} = \s_{_{-}}$ on $[-d , 0 \,) \,$.
Therefore, the family $\{ \bars_{_{\!\! (\r)}} \}_{_{\! \r}}$ converges uniformly in $\{ [-d , -\r] \}_{_{\! \r}}$ to $\bars_{_{\! -}} \,$. Moreover, since $\lim_{\t \to 0^-} \bars_{_{\! -}}$ exists f\mbox{}inite, by proposition \ref{proposition Qr} we get 
\beq
\label{Noel Monique}
\lim_{\r \to 0^+} \bars_{_{\!\! (\r)}}(-\r) = \lim_{\t \to 0^{-}} \bars_{_{\! -}}(\t) \quad .
\eeq

\subsection{Absence of self-intersections}
\label{section Absence of self-intersections}

We now turn to the choice of $\z_{_{+}}$, or equivalently of $\s_{_{\! +}}$. For future reference note that, regardless to the value of $\z_{_{+}}$, the following proposition holds
\begin{proposition}
\label{proposition 7}
The family of functions $\{ \detau \, \bars_{_{\!\! (\r)}} \}_{_{\! \r}}$ converges uniformly in $\{ [\, \r , d \,] \}_{_{\! \r}}$ to $\detau \, \s_{_{\! +}}$ .
\end{proposition}
\begin{proof}
The proof is completely analogous to that of proposition \ref{proposition 6}.
\end{proof}

First of all, let's see that $\lim_{\r \to 0^+} \bars_{_{\!\! (\r)}} (\r)$ exists f\mbox{}inite. To do this we consider the pillbox integration of $\detau \bars_{_{\!\! (\r)}}$, namely (using (\ref{Shion})) we consider the integral over $[- \r, \r \,]$ of
\begin{equation}
\label{Shion 2}
\detau \bars_{_{\!\! (\r)}} (\t) = \dfrac{\detau \big\langle \mbf{z}_{_{(\r)}} , \mbfscn_{_{\! (\r)}} \big\rangle_{_{\!\! \r}} (\t)}{\big\langle \mbf{y}_{_{\!\! (\r)}} , \mbfscn_{_{\! (\r)}} \big\rangle_{_{\!\! \r}} (\t)}\quad ,
\end{equation}
with the aim of taking the limit $\r \to 0^+$. Note that the relation (\ref{Shion 2}) holds only for $\big\langle \mbf{y}_{_{\!\! (\r)}} , \mbfscn_{_{\! (\r)}} \big\rangle_{_{\!\! \r}} (\t) \neq 0 \,$, and this may cast doubts about the integral of the right hand side being well-def\mbox{}ined. But, since we proved that the left hand side is continuous on $[-d, d \,] \,$, by consistency the right hand side has to have removable discontinuities when $\big\langle \mbf{y}_{_{\!\! (\r)}} , \mbfscn_{_{\! (\r)}} \big\rangle_{_{\!\! \r}} (\t) = 0$ and be continuous elsewhere. This can be indeed verif\mbox{}ied directly using (\ref{Lucia}), therefore the integral over $[- \r, \r \,]$ of (\ref{Shion 2}) is well-def\mbox{}ined. Now, since we showed above that $\lim_{\r \to 0^+} \bars_{_{\!\! (\r)}} (-\r)$ exists f\mbox{}inite, it follows that $\lim_{\r \to 0^+} \bars_{_{\!\! (\r)}} (\r)$ exists f\mbox{}inite if and only if the limit $\r \to 0^+$ of the integral
\begin{equation}
\label{Ozawa}
\int_{-\r}^{\r} \dfrac{\detau \big\langle \mbf{z}_{_{(\r)}} , \mbfscn_{_{\! (\r)}} \big\rangle_{_{\!\! \r}}}{\big\langle \mbf{y}_{_{\!\! (\r)}} , \mbfscn_{_{\! (\r)}} \big\rangle_{_{\!\! \r}}} \,\, d\t
\end{equation}
exists f\mbox{}inite, in which case we have
\begin{equation}
\label{Sora Aoi}
\lim_{\r \to 0^+} \bars_{_{\!\! (\r)}} (\r) - \lim_{\r \to 0^+} \bars_{_{\!\! (\r)}} (-\r) = \lim_{\r \to 0^+} \int_{-\r}^{\r} \dfrac{\detau \big\langle \mbf{z}_{_{(\r)}} , \mbfscn_{_{\! (\r)}} \big\rangle_{_{\!\! \r}}}{\big\langle \mbf{y}_{_{\!\! (\r)}} , \mbfscn_{_{\! (\r)}} \big\rangle_{_{\!\! \r}}} \,\, d\t \quad .
\end{equation}

To establish the existence of the $\r \to 0^+$ limit of (\ref{Ozawa}), let's comment on the range of values that $\bars_{_{\!\! (\r)}} (\r) - \bars_{_{\!\! (\r)}} (-\r)$ can take when $\r$ is arbitrarily small. Let's indicate $I_{_{\! (\r)}} =  \big[ \! - \! \r, \r \, \big]$. Since the families of auxiliary f\mbox{}ields $\{ \mbf{y}_{_{\!\! (\r)}} \}_{_{\! \r}}$ and $\{ \mbf{z}_{_{(\r)}} \}_{_{\! \r}}$ converge uniformly on $[-d, d \,]$ to continuous conf\mbox{}igurations, we have that
\begin{align}
\label{petrushka app}
\lim_{\r \to 0^+} \mbf{y}_{_{\!\! (\r)}} \big( I_{_{\! (\r)}} \big) &= \mbf{y} (0) & \lim_{\r \to 0^+} \mbf{z}_{_{(\r)}} \big( I_{_{\! (\r)}} \big) &= \mbf{z} (0) \quad .
\end{align}
Let's suppose that, for every $\bar{\r} > 0$, there exists a positive $\r \leq \bar{\r}$ such that $\babs{\bars_{_{\!\! (\r)}} (\r) - \bars_{_{\!\! (\r)}} (-\r)} > \pi \,$. The relation (\ref{auxiliary 1}) implies that $\bars_{_{\!\! (\r)}}(\t)$ is the angular (polar) coordinate of the vector $\mbfscn_{_{\! (\r)}}(\t)$ in the $\big( \mbf{y}_{_{\!\! (\r)}}(\t),\mbf{z}_{_{(\r)}}(\t) \big)$ reference frame. Since $\mbfscn_{_{\! (\r)}}(\t)$ is the tangent vector to $\mcalC_1$ in the direction orthogonal to $\mcalC_2$, it follows that the tangent vector rotates of an angle $> \pi$ going from $-\r$ to $\r$ with respect to the above mentioned reference frame. In principle, this relative rotation may be due to the frame itself undergoing a rotation which does not reduce to the identity when $\r \to 0^+$. This is however excluded by (\ref{petrushka}) which states that, for $\r$ small enough, the frames stays nearly constant between $-\r$ to $\r \,$. Therefore, the tangent vector $\mbfscn_{_{\! (\r)}}$ indeed rotates of an angle $> \pi$ between $-\r$ and $\r \,$. Since the points $\t = -\r$ and $\t = \r$ coincide in the limit $\r \to 0^+$, this situation corresponds to a conf\mbox{}iguration where the cod-1 brane is self-intersecting (see f\mbox{}igure \ref{fig self intersection}). However, we excluded from the onset as non-physical the presence of brane self-intersections, by def\mbox{}ining the cod-1 brane to be an embedded submanifold. In addition, the case $\babs{\bars_{_{\!\! (\r)}} (\r) - \bars_{_{\!\! (\r)}} (-\r)} = \pi$ corresponds to a situation where the geometry is pathological, since the bulk shrinks to a line. Therefore only the cases such that $\babs{\bars_{_{\!\! (\r)}} (\r) - \bars_{_{\!\! (\r)}} (-\r)} < \pi$ when $\r \to 0^+$ are allowed in our analysis.
\begin{figure}[t!]
\centering
\begin{tikzpicture}[scale=0.8]
\draw[->,thin,color=black!50] (-5,0) -- (5,0);
\draw[->,thin,color=black!50] (0,-2.5) -- (0,1.5);
\draw[thick,color=blue] (-4,0) -- (2,0);
\draw[thick,color=blue] (2,0) .. controls (4,0) and (4,0.8) .. (2,0);
\draw[thick,color=blue] (2,0) -- +(200:6);
\draw[->,thick] (-2.5, 0) -- +(1,0) node[above,color=black] {$\mathbf{n}_{_{-}}$};
\draw[->,thick] (2, 0)+(200:3.4) -- +(200:4.4) node[above,color=black] {$\mathbf{n}_{_{+}}$};
\end{tikzpicture}
\caption{Self-intersections and $\babs{\bars_{_{\! (\r)}} (\r) - \bars_{_{\! (\r)}} (-\r)} = 10 \pi/9 \,$.}
\label{fig self intersection}
\end{figure}

This condition implies that, by exploiting the freedom in the choice of the auxiliary f\mbox{}ields mentioned at the end of section \ref{Constructing the auxiliary fields}, we can choose the families $\{ \mbf{y}_{_{\!\! (\r)}} \}_{_{\! \r}}$ and $\{ \mbf{z}_{_{(\r)}} \}_{_{\! \r}}$ such that $\big\langle \mbf{y}_{_{\!\! (\r)}} , \mbfscn_{_{\! (\r)}} \big\rangle_{_{\!\! \r}} (\t) > 0$ for $\t \in [ -\r , \r \,]$ when $\r$ is suf\mbox{}f\mbox{}iciently small. With this choice (\ref{Lucia}) implies that
\begin{equation}
\big\langle \mbf{y}_{_{\!\! (\r)}} , \mbfscn_{_{\! (\r)}} \big\rangle_{_{\!\! \r}} (\t) = \sqrt{1 - \big\langle \mbf{z}_{_{(\r)}} , \mbfscn_{_{\! (\r)}} \big\rangle^{2}_{_{\!\! \r}} (\t)} \quad ,
\end{equation}
and the integral (\ref{Ozawa}) can be evaluated exactly to give
\begin{equation}
\label{Maria Ozawa}
\int_{-\r}^{\r} \dfrac{\detau \big\langle \mbf{z}_{_{(\r)}} , \mbfscn_{_{\! (\r)}} \big\rangle_{_{\!\! \r}}}{\sqrt{1 - \big\langle \mbf{z}_{_{(\r)}} , \mbfscn_{_{\! (\r)}} \big\rangle^{2}_{_{\!\! \r}}}} \,\, d\t = \bigg[ \arcsin \Big( \big\langle \mbf{z}_{_{(\r)}} , \mbfscn_{_{\! (\r)}} \big\rangle_{_{\!\! \r}} (\r) \Big) - \arcsin \Big( \big\langle \mbf{z}_{_{(\r)}} , \mbfscn_{_{\! (\r)}} \big\rangle_{_{\!\! \r}} (-\r) \Big) \bigg] \, .
\end{equation}
Taking the limit $\r \to 0^+$ of the previous equation and using (\ref{Sora Aoi}) we obtain
\begin{equation}
\label{Julie Kay}
\lim_{\r \to 0^+} \bars_{_{\!\! (\r)}} (\r) - \lim_{\r \to 0^+} \bars_{_{\!\! (\r)}} (-\r) = \arcsin \big\langle \mbf{z}(0) , \mbfscn_{_{+}} \big\rangle - \arcsin \big\langle \mbf{z}(0) , \mbfscn_{_{-}} \big\rangle \quad ,
\end{equation}
which implies that $\lim_{\r \to 0^+} \bars_{_{\!\! (\r)}} (\r)$ exists f\mbox{}inite.

The state of af\mbox{}fairs is then the following. By proposition \ref{proposition 7}, the family of derivatives $\{ \detau \, \bars_{_{\!\! (\r)}} \}_{_{\! \r}}$ converges uniformly in $\{ [\, \r , d \,] \}_{_{\! \r}}$ to $\detau \, \s_{_{\! +}}$, regardless of the choice for $\z_{_{+}}$. Furthermore, as we mentioned above, $\detau \, \s_{_{\! +}}$ has a f\mbox{}inite limit for $\t \to 0^+$, again regardless of the choice for $\z_{_{+}}$. Since $\lim_{\r \to 0^+} \bars_{_{\!\! (\r)}} (\r)$ exists f\mbox{}inite, the proposition \ref{proposition slope plus} implies that the family $\{ \bars_{_{\!\! (\r)}} \}_{_{\! \r}}$ converges uniformly in $\{ [ \, \r, d \, ] \}_{_{\! \r}}$ to a smooth function $\bs_{_{+}} : ( 0 , d \, ] \to \mbbR \,$, and $\detau \, \bs_{_{+}} = \detau \, \s_{_{\! +}} \,$. Moreover, by proposition \ref{proposition Qr} it follows that
\beq
\lim_{\t \to 0^+} \bs_{_{+}}(\t) = \lim_{\r \to 0^+} \bars_{_{\!\! (\r)}} (\r) \quad ,
\eeq
which using (\ref{Noel Monique}) and (\ref{Julie Kay}) in turn implies
\begin{equation}
\lim_{\t \to 0^+} \bs_{_{+}} - \lim_{\t \to 0^{-}} \bars_{_{\! -}} = \arcsin \big\langle \mbf{z}(0) , \mbfscn_{_{+}} \big\rangle - \arcsin \big\langle \mbf{z}(0) , \mbfscn_{_{-}} \big\rangle \quad .
\end{equation}
On the other hand, applying the $\arcsin$ to the relations on the right of (\ref{Layla Rose}) we get 
\begin{equation}
\lim_{\t \to 0^+} \s_{_{\! +}} - \lim_{\t \to 0^{-}} \bars_{_{\! -}} + 2 \pi \, k = \arcsin \big\langle \mbf{z}(0) , \mbfscn_{_{+}} \big\rangle - \arcsin \big\langle \mbf{z}(0) , \mbfscn_{_{-}} \big\rangle
\end{equation}
for some $k \in \mbbZ \,$, from which we obtain
\begin{equation}
\lim_{\t \to 0^+} \bs_{_{+}} = \lim_{\t \to 0^+} \s_{_{\! +}} + 2 \pi \, k \quad .\footnote{In general $\arcsin(\sin x) \neq x$, since by convention $\arcsin : [-1 , 1] \to [-\pi/2 , \pi/2 \,]$ while $x$ may lie outside of $[-\pi/2 , \pi/2 \,]\,$.}
\end{equation}
Remind now that, if $\ti{\z}_{_{+}}$ is a solution of (\ref{cos S sin S}) in $\t = d \,$, then all and only the other solutions takes the form $\z_{_{+}} = \ti{\z}_{_{+}} + 2 \pi \, \bk\,$, with $\bk \in \mbbZ \,$. Since by (\ref{Becca Bali}) the choice of $\bk$ does not inf\mbox{}luence $\detau \, \s_{_{\! +}} \,$, it follows that all the solutions for $\s_{_{\! +}}$ are of the form $\s_{_{\! +}}(\t) = \ti{\s}_{_{\! +}}(\t) + 2 \pi \, \bk\,$, where $\ti{\s}_{_{\! +}}$ is the solution associated to $\ti{\z}_{_{+}}$. Therefore, calling $\bar{\s}_{_{\! +}}$ the smooth solution of (\ref{cos S sin S}) on $( 0 , d \, ]$ correspondent to the choice $\bk = - k \,$, we have
\begin{equation}
\lim_{\t \to 0^+} \bar{\s}_{_{\! +}} = \lim_{\t \to 0^+} \bs_{_{+}}
\end{equation}
and f\mbox{}inally
\begin{equation}
\label{Noel Monique 2}
\lim_{\t \to 0^+} \bar{\s}_{_{\! +}}(\t) = \lim_{\r \to 0^+} \bars_{_{\!\! (\r)}} (\r) \quad .
\end{equation}
We conclude that it is indeed possible to choose $\bz_{_{+}} = \bar{\s}_{_{\! +}}(d)$ such that the family $\{ \bars_{_{\!\! (\r)}} \}_{_{\! \r}}$ converges uniformly in $\{ [ \, \r , d \, ] \}_{_{\! \r}}$ to $\bars_{_{\! +}} \,$. Moreover, from (\ref{Noel Monique}), (\ref{Julie Kay}) and (\ref{Noel Monique 2}) it follows that
\beq
\label{rugby}
\barS_{_{+}} - \barS_{_{-}} = \arcsin \big\langle \mbf{z}(0) , \mbfscn_{_{+}} \big\rangle - \arcsin \big\langle \mbf{z}(0) , \mbfscn_{_{-}} \big\rangle \quad .
\eeq

\subsection{Self-consistency of the analysis}

Let us f\mbox{}inally def\mbox{}ine
\begin{gather}
\barS_{_{\! (\r)}}(\t) = \bars_{_{\!\! (\r)}} (\t) \\[3mm]
\barS(\t) =
\begin{cases}
\bars_{_{\! -}}(\t)& \text{for} \,\, \t \in [-d, 0) \\[2mm]
\bars_{_{\! +}}(\t)& \text{for} \,\, \t \in (0, d \,]
\end{cases} \\[3mm]
\barS_{_{\!-}} = \lim_{\t \to 0^{-}} \bars_{_{\! -}}(\t) \qquad \qquad \barS_{_{\!+}} = \lim_{\t \to 0^{+}} \bars_{_{\! +}}(\t) \quad .
\end{gather}
It follows from the discussion above that
\begin{enumerate}
 \item $\barS_{_{\! (\r)}}$ is smooth on $[-d , d\,] \,$, $\barS$ is smooth on $[-d , d\,] \setminus \{ 0 \}$ and $\barS_{_{\!-}} \,$, $\barS_{_{\!+}}$ are f\mbox{}inite ;
 \item the family $\big\{ \barS_{_{\! (\r)}} \big\}_{_{\!\! \r}}$ converges to $\barS$ uniformly in $\{ [-d , -\r \,] \cup [\, \r , d \,] \}_{_{\! \r}}\,$.
\end{enumerate}
We therefore proved the existence of solutions of (\ref{cos S sin S}) and (\ref{cos Sk sin Sk}) which satisfy 1. and 2., as claimed in section \ref{The slope function}.

Some considerations are appropriate at this point. In sections \ref{section The codimension-2 pillbox integration} and \ref{section Thin limit of the Cascading DGP model} we saw that from the points 1. and 2. it follows that
\beq
\mscrI = S_{_{\!+}} - S_{_{\!-}} \quad ,
\eeq
where $\mscrI$ is def\mbox{}ined in (\ref{integral N}), and that
\beq
S_{_{\!+}} - S_{_{\!-}} = \sgn \Big( \big\langle \mbfscn_{_{+}}, \mbf{n}_{_{-}} \big\rangle \Big) \,\, \arccos \Big( \big\langle \mbf{n}_{_{+}}, \mbf{n}_{_{-}} \big\rangle \Big) \quad .
\eeq
Since the def\mbox{}inition of $\mscrI$ does not involve the slope functions, this indirectly implies that every choice of $\big\{ S_{_{\! (\r)}} \big\}_{_{\!\! \r}}$ and $S$ which satisfy 1. and 2. gives the same result for $S_{_{\!+}} - S_{_{\!-}}\,$. To conf\mbox{}irm the self-consistency of the analysis, it is useful to show directly that this holds.

Let us consider a generic choice of $\big\{ S_{_{\! (\r)}} \big\}_{_{\!\! \r}}$ and $S$ which satisfy the points 1. and 2. . As we mentioned above, the point 1. implies that the functions $\big\{ S_{_{\! (\r)}} \big\}_{_{\!\! \r}}$ and $S$ are completely parametrized by the numbers
\beq
\big\{ \z_{_{-}}^{_{(\r)}} = S_{_{\! (\r)}}(-d) \big\}_{_{\!\! \r}} \qquad , \qquad \z_{_{-}} = S(-d) \qquad , \qquad \z_{_{+}} = S(d) \quad .
\eeq
Moreover, the point 2. implies that $\lim_{\r \to 0^+} \z_{_{-}}^{_{(\r)}} = \z_{_{-}} \,$. Note that this allows $\abs{\z_{_{-}}^{_{(\r)}} - \z_{_{-}}} > \pi$ to hold for a f\mbox{}inite subset of values of $\r \,$, therefore the bare validity of the points 1. and 2. allows a class of solutions of (\ref{cos S sin S}) and (\ref{cos Sk sin Sk}) which is wider that the one we constructed above. Along with $\big\{ S_{_{\! (\r)}} \big\}_{_{\!\! \r}}$ and $S \,$, let us consider a set of solutions $\bars_{_{\! -}}$, $\{ \bars_{_{\!\! (\r)}} \}_{_{\! \r}}$, $\bars_{_{\! +}}$ of the type we constructed above, and call $m$ the integer such that $\z_{_{-}} - \bz_{_{-}} = 2 \pi \, m \,$. By the relation (\ref{Becca Bali}), we have that
\beq
\label{baton}
S(\t) - \bars_{_{\! -}}(\t) = 2 \pi \, m \qquad \text{for every} \,\, \t \in [-d, 0) \quad ,
\eeq
which in particular implies that $S_{_{\!-}} = \lim_{\t \to 0^-} S$ exists f\mbox{}inite and $S_{_{\!-}} = \barS_{_{\!-}} + 2 \pi \, m \,$. Moreover, since the family $\{ S_{_{\! (\r)}} \}_{_{\! \r}}$ converges uniformly in $\{ [-d , -\r] \}_{_{\! \r}}$ to $S \,$, by the proposition \ref{proposition Qr} we have 
\beq
\label{Noel Moniques}
\lim_{\r \to 0^+} S_{_{\! (\r)}}(-\r) = \barS_{_{\!-}} + 2 \pi \, m \quad .
\eeq
On the other hand, analogously to the analysis of section \ref{section Absence of self-intersections} we have
\begin{equation}
\label{Julie Kays}
\lim_{\r \to 0^+} S_{_{\! (\r)}} (\r)  = \lim_{\r \to 0^+} S_{_{\! (\r)}} (-\r) + \arcsin \big\langle \mbf{z}(0) , \mbfscn_{_{+}} \big\rangle - \arcsin \big\langle \mbf{z}(0) , \mbfscn_{_{-}} \big\rangle \quad ,
\end{equation}
and therefore using (\ref{rugby}) and (\ref{Noel Moniques}) we get
\begin{equation}
\label{Desiree West}
\lim_{\r \to 0^+} S_{_{\! (\r)}} (\r) = \barS_{_{\!+}} + 2 \pi \, m \quad .
\end{equation}
Now, since the family $\{ S_{_{\! (\r)}} \}_{_{\! \r}}$ converges uniformly in $\{ [\, \r, d \,] \}_{_{\! \r}}$ to $S$ and $\lim_{\r \to 0^+} S_{_{\! (\r)}} (\r)$ exists f\mbox{}inite, again by the proposition \ref{proposition Qr} we have 
\begin{equation}
S_{_{\!+}} = \lim_{\r \to 0^+} S_{_{\! (\r)}} (\r) = \barS_{_{\!+}} + 2 \pi \, m \quad ,
\end{equation}
and so $S(\t) - \bars_{_{\! +}}(\t) = 2 \pi \, m$ for every $\t \in (0 , d \,] \,$. Remembering (\ref{baton}), it follows that 
\beq
S(\t) = \barS(\t) + 2 \pi \, m
\eeq
for every $\t \in [-d, d \,] \setminus \{ 0 \} \,$, and in particular
\beq
S_{_{\!+}} - S_{_{\!-}} = \barS_{_{\!+}} - \barS_{_{\!-}} \quad .
\eeq
We conclude that, provided the points 1. and 2. hold, changing the set of solutions $S$ and $\big\{ S_{_{\! (\r)}} \big\}_{_{\!\! \r}}$ merely results in a shift of $S$ by a constant phase, multiple of $2 \pi \,$.

\section{De Sitter vacuum cosmological solutions}
\label{appendix de Sitter}

Let us now show how the general equations (\ref{Bulkeq Cascading})--(\ref{cod2 jceq Cascading}) can be used to derive the vacuum solutions found in \cite{Minamitsuji:2008fz}, where an empty cod-2 brane, with a de Sitter induced metric, is embedded in a Riemann-f\mbox{}lat bulk.

\subsection{The ansatz for the metric and the embeddings}

To look for solutions where the metric induced on the cod-2 brane is of cosmological form, it is very convenient to foliate the bulk into 4D Robertson-Walker leaves. This is indeed possible when the bulk has vanishing Riemann curvature, since generalizing the analysis of \cite{Deruelle:2000ge} we can choose the coordinate system such that the line element takes the form
\beq
\label{bulk line cartesian}
ds^2 = -N^{2}(t,z) \, dt^2 + A^{2}(t,z) \, \d_{ij} \, dx^i dx^j + dy^2 + dz^2 \quad ,
\eeq
where 
\begin{align}
\label{foliation}
N(t,z) &= 1 + \ep \,  z \, \frac{\ddot{a}(t)}{\dot{a}(t)} & A(t,z) &= a(t) + \ep \, z \, \dot{a}(t) \quad ,
\end{align}
and $\ep = \pm 1\,$. The bulk metric is therefore completely determined by the function
\beq
\label{scale factor}
a(t) = A(t,0) \quad ,
\eeq
apart from the choice of $\ep\,$. To make sure that the induced metric on the cod-2 brane is of the Robertson-Walker form, it is then enough to identify the brane with one of the leaves, thereby gauge-f\mbox{}ixing its position. We choose to identify the cod-2 brane with the $y = z = 0$ leaf. Parametrizing the brane with the coordinates $(t, \vec{x})\,$, the cod-2 induced line element reads
\beq
d \bar{s}^{2} = - dt^2 + a^{2}(t) \, \d_{ij} \, dx^i dx^j \quad ,
\eeq
so $t$ is the proper time and $a(t)$ is the scale factor on the cod-2 brane.

Regarding the codimension-1 brane, we use the cod-1 GNC and ask the cod-1 embedding to respect the translational invariance with respect to $\vec{x}\,$. In particular we assume
\beq
\label{cod1 embedding}
\hvf^{_{A}}(t\, , \vec{x}\, , \z) = \Big( t\, , \vec{x}\, , \hvf^{y}(t, \z) , \hvf^{z}(t, \z) \Big) \quad ,
\eeq
where the cod-2 brane is identif\mbox{}ied by $\z = 0\,$. Since the position of the latter is f\mbox{}ixed we have
\beq
\label{oppaaa}
\hvf^{y}(t\, , 0) = \hvf^{z}(t\, , 0) = 0 \quad ,
\eeq
and $\hvf^{y}$ , $\hvf^{z}$ are taken to be respectively odd and even with respect to $\z$ to enforce the $\mbbZ_2$ symmetry inside the cod-1 brane. It can be verif\mbox{}ied that the non-vanishing components of the cod-1 induced metric read
\begin{align}
\hg_{\t\t} &= \big( \hvf^{y \, \p} \big)^{2} + \big( \hvf^{z \, \p} \big)^{2} & \hg_{\t t} &= \dot{\hvf}^{y} \, \hvf^{y \, \p} + \dot{\hvf}^{z} \, \hvf^{z \, \p} \\[2mm]
\hg_{tt} &= \big( \dot{\hvf}^{y} \big)^{2} + \big( \dot{\hvf}^{z} \big)^{2} - \bigg( 1 + \ep \, \frac{\ddot{a}}{\dot{a}} \, \hvf^{z} \bigg)^{\!\! 2} & \hg_{ij} &= \Big( a + \ep \, \dot{a} \, \hvf^{z} \Big)^{\! 2} \, \d_{ij} \quad ,
\end{align}
so the condition of the coordinate system on the cod-1 brane being Gaussian Normal with respect to the cod-2 brane implies
\begin{align}
\label{Sally}
\big( \hvf^{y \, \p} \big)^{2} + \big( \hvf^{z \, \p} \big)^{2} &= 1 & \dot{\hvf}^{y} \, \hvf^{y \, \p} + \dot{\hvf}^{z} \, \hvf^{z \, \p} &= 0 \quad .
\end{align}

\subsubsection{The auxiliary vector f\mbox{}ields and the slope function}

To construct the auxiliary vector f\mbox{}ields, we note that for an embedding of the form (\ref{cod1 embedding}) the tangent vector f\mbox{}ields $\mbf{u}_{_{[\m]}}$ (see def\mbox{}inition (\ref{def N and u})) read
\begin{align}
u_{^{[t]}}^{_{A}}(t , \t) &= \Big( 1\, , \vec{0}\, , \dot{\hvf}^{y} , \dot{\hvf}^{z} \Big) & u_{^{[i]}}^{_{A}}(t , \t) &= \d_{i}^{\, A} \quad ,
\end{align}
where $i = 1, 2, 3\,$. By def\mbox{}inition, the auxiliary f\mbox{}ields $\mbf{y}(t, \t)$ and $\mbf{z}(t, \t)$ satisfy the condition of orthogonality to the tangent vector f\mbox{}ields
\beq
y^{i} = z^{i} = 0 \quad \,\, , \,\, \quad - \hat{N}^{2} \, y^{t} + \dot{\hvf}^{y} \, y^{y} + \dot{\hvf}^{z} \, y^{z} = 0 \quad \,\, , \,\, \quad -\hat{N}^{2} \, z^{t} + \dot{\hvf}^{y} \, z^{y} + \dot{\hvf}^{z} \, z^{z} = 0 \quad ,
\eeq
where we indicated $\hat{N}(t, \t) = N \big( t, \hvf^{z}(t, \t) \big)\,$, and the condition of orthonormality between themselves. Among the (inf\mbox{}inite) possible choices for $\mbf{y}$ and $\mbf{z}\,$, a convenient one is
\begin{align}
y^{_{A}}(t , \t) &= \frac{1}{\hN \sqrt{\hN^{2} - \big(\dhvf^{y}\big)^{\! 2}}} \,\, \bigg( \dhvf^{y} \, , \, \vec{0} \, , \, \hN^{2} \, , \, 0 \bigg) \label{Vivi} \\[2mm]
z^{_{A}}(t , \t) &= \frac{1}{\sqrt{\Big( \hN^{2} - \big(\dhvf^{y}\big)^{\! 2} \Big) \Big( \hN^{2} - \big(\dhvf^{y}\big)^{\! 2} - \big(\dhvf^{z}\big)^{\! 2} \Big)}} \,\, \bigg( \dhvf^{z} \, , \, \vec{0} \, , \, \dhvf^{y} \, \dhvf^{z} \, , \, \hN^{2} - \big(\dhvf^{y}\big)^{\! 2} \bigg) \label{ana} \quad ,
\end{align}
which is indeed continuous at $\z = 0$ where it reads
\begin{align}
y^{_{A}}_{_{+}}(t) &= y^{_{A}}_{_{-}}(t) = \Big( 0 \, , \, \vec{0} \, , \, 1 \, , \, 0 \Big) & z^{_{A}}_{_{+}}(t) &= z^{_{A}}_{_{-}}(t) = \Big( 0 \, , \, \vec{0} \, , \, 0 \, , \, 1 \Big) \label{Rebeca} \quad .
\end{align}
It follows from the relations (\ref{auxiliary 1 lim}) that
\begin{align}
\hvf^{y \, \p}_{_{\! +}}(t) &= \cos S_{_{\! +}}(t) & \hvf^{z \, \p}_{_{\! +}}(t) &= \sin S_{_{\! +}}(t) \label{Mandzukic 1} \\[3mm]
n^{y}_{_{\! +}}(t) &= - \sin S_{_{\! +}}(t) & n^{z}_{_{\! +}}(t) &= \cos S_{_{\! +}}(t) \label{Mandzukic 2} \quad .
\end{align}

In light of the symmetry properties of $\hvf^{y}$ and $\hvf^{z}\,$, it would be very convenient to impose the slope function to be odd with respect to $\z\,$. However, taking into account the bound $\abs{S_{_{\! +}} - S_{_{\! -}}} < \pi$ discussed in section \ref{section Absence of self-intersections}, this would artif\mbox{}icially force $\hvf^{y \, \p}_{_{\! +}}$ and $n^{z}_{_{\! +}}$ to be positive. While we can always achieve $\hvf^{y \, \p}_{_{\! +}} > 0\,$, in case by changing $\z \to - \z\,$, it is a priori not clear if we can take $n^{z}_{_{\! +}} > 0$ as well. To see this, let us impose $S$ to be odd with respect to $\z$ and introduce a sign parameter $\vep = \pm 1$ to allow both the cases $n^{z}_{_{\! +}} > 0$ and $n^{z}_{_{\! +}} < 0\,$, so that
\begin{align}
n^{y}_{_{\! +}}(t) &= - \vep \, \sin S_{_{\! +}}(t) & n^{z}_{_{\! +}}(t) &= \vep \, \cos S_{_{\! +}}(t)
\end{align}
and
\beq
\label{S domain}
- \frac{\pi}{2} < S_{_{\! +}} < \frac{\pi}{2} \quad .
\eeq
Note that, regarding the line element (\ref{bulk line cartesian})--(\ref{foliation}), switching $\ep \to - \ep$ is equivalent to transforming $z \to - z\,$, which is a bulk gauge transformation. Taking into account the presence of the branes, a transformation equivalent to $z \to - z$ is obtained by switching $\ep \to - \ep$ and at the same time performing a mirror ref\mbox{}lection of the branes with respect to the plane $z = 0\,$, i.e.~transforming $S \to - S\,$. Finally, recall that the full 6D conf\mbox{}iguration is constructed by cutting the part of the bulk the normal vector $\mbf{n}$ points into, copying it and pasting together the two mirror copies in a $\mbbZ_2$ way with respect to the cod-1 brane. The 6-dimensional spacetime is left unaltered if, in addition to $\ep \to - \ep$ and $S \to - S\,$, we also change orientation to the cod-1 brane. We therefore conclude that the transformation
\begin{align}
\ep &\to - \ep & S(t, \z) &\to - S(t, \z) & \vep &\to - \vep
\end{align}
is a gauge transformation for the full 6-dimensional conf\mbox{}iguration. It follows that we can indeed f\mbox{}ix $n^{z}_{_{\! +}} > 0$ while leaving $\ep = \pm 1$ unconstrained and $S_{_{\! +}}$ def\mbox{}ined on the interval (\ref{S domain}). Using (\ref{Mandzukic 1}) and (\ref{Mandzukic 2}) it is then easy to show that
\beq
\D(t) = 4 \, S_{_{\! +}}(t) \quad ,
\eeq
so a positive $S_{_{\! +}}$ corresponds to a def\mbox{}icit angle conf\mbox{}iguration while a negative $S_{_{\! +}}$ corresponds to an excess angle conf\mbox{}iguration.

\subsubsection{The equations of motion}

Turning now to the equations of motion, it is clear that the bulk equations (\ref{Bulkeq Cascading}) are identically satisf\mbox{}ied, since the bulk is f\mbox{}lat.

Regarding the cod-2 junction conditions (\ref{cod2 jceq Cascading}), it can be verif\mbox{}ied that the non-vanishing components of $\bar{\mbf{K}} - \bar{\mbf{g}} \,\, tr \big( \bar{\mbf{K}} \big)$ at the ``+'' side of the cod-2 brane read
\begin{align}
\Big[ \bK_{tt} - \bK \, \bg_{tt} \, \Big]_{_{+}} &= - 3 \, \ep \, \frac{\dot{a}}{a} \, \sin S_{_{\! +}} \label{cod2 extrinsic curvature} \\[2mm]
\Big[ \bK_{ij} - \bK \, \bg_{ij} \, \Big]_{_{+}} &= \ep \, a^2 \, \bigg( \frac{\ddot{a}}{\dot{a}} + 2 \, \frac{\dot{a}}{a} \bigg) \, \sin S_{_{\! +}} \, \d_{ij} \nn \quad ,
\end{align}
where dependence of $a$ and $S_{_{\! +}}$ on $t$ is implicitly understood. Reminding that we are considering vacuum solutions, and taking suitable linear combinations of the independent components of the cod-2 junction conditions, we obtain the modif\mbox{}ied Friedmann equations
\begin{align}
\frac{\dot{a}^2}{a^2} - 2 \, \ep \,  m_{5} \, \frac{\dot{a}}{a} \, \sin \frac{\D}{4} + \frac{m_{6} \, m_{5}}{3} \, \D &= 0 \label{mod Friedmann eq 1} \\[3mm]
\frac{\ddot{a}}{a} - \ep \, m_{5} \, \bigg( \frac{\ddot{a}}{\dot{a}} + \frac{\dot{a}}{a} \bigg) \, \sin \frac{\D}{4} + \frac{m_{6} \, m_{5}}{3}\, \D &= 0 \quad . \label{mod Friedmann eq 2}
\end{align}

Let us now turn to the cod-1 junction conditions (\ref{junctionconditionseq Cascading}). It is easy to recognise that the $\m\n$ components contain the second derivative with respect to $\z$ of the cod-1 embedding functions, while the $\z\z$ and $\z\m$ components do not. Therefore, from the point of view of the evolution with respect to $\z\,$, the $\z\z$ and $\z\m$ components play the role of constraint equations, while the $\m\n$ components describe how the embedding evolves when moving away from the cod-2 brane. Since we are mainly interested in the solution for the metric on the cod-2 brane, we concentrate on the $\z\z$ and $\z\m$ components of the cod-1 junction conditions evaluated at $\z = 0^{+}$. It is straightforward to show that
\begin{align}
\Big[ \, \hK_{\z\z} - \hg_{\z\z} \, \hK \, \Big]_{_{+}} &= \ep \, \bigg( \frac{\ddot{a}}{\dot{a}} + 3 \, \frac{\dot{a}}{a} \bigg) \, \cos S_{_{\! +}} \nn \\[2mm]
\Big[ \, \hK_{\z t} - \hg_{\z t} \, \hK \, \Big]_{_{+}} &= \dot{S}_{_{\! +}} \quad , \label{cod1 extr curv at cod2}
\end{align}
and that
\begin{align}
\hG_{\z\z}\big\rvert_{_{+}} &= - 3 \, \frac{\dot{a}}{a} \, \bigg( \frac{\ddot{a}}{\dot{a}} + \frac{\dot{a}}{a} \bigg) \, \cos^{2}\! S_{_{\! +}} \nn \\[1mm]
\hG_{\z t}\big\rvert_{_{+}} &= - 3 \, \ep \, \frac{\dot{a}}{a} \, \dot{S}_{_{\! +}} \cos S_{_{\! +}} \label{cod1 intr curv at cod2} \quad .
\end{align}
Using these results, the $\z\z$ and $\z t$ components of the cod-1 junction conditions evaluated at the ``$+$'' side of the cod-2 brane read respectively
\begin{align}
\cos \frac{\D}{4} \,\, \Bigg[ \, 2 \, \ep \, m_{6} \, \bigg( \frac{\ddot{a}}{\dot{a}} + 3 \, \frac{\dot{a}}{a} \bigg) - 3 \, \frac{\dot{a}}{a} \, \bigg( \frac{\ddot{a}}{\dot{a}} + \frac{\dot{a}}{a} \bigg) \, \cos \frac{\D}{4} \, \Bigg] &= 0 
\label{cod1jc side zz} \\[3mm]
\dot{\D} \, \bigg( 2 \, m_{6} - 3 \, \ep \, \frac{\dot{a}}{a} \, \cos \frac{\D}{4} \bigg) &= 0 \label{cod1jc side zt} \quad ,
\end{align}
while the $\z i$ components vanish identically.

\subsection{Vacuum solutions}

Note f\mbox{}irst of all that the cod-2 junction conditions and the constraint cod-1 junction conditions give a system of four equations for two variables, the scale factor of the cod-2 brane and the local def\mbox{}icit angle. Therefore, to know what happens at the cod-2 brane, we don't really need to solve for the behaviour outside of it.

On the other hand, the system seems over-constrained. It can however be seen that (\ref{mod Friedmann eq 2}) is implied by (\ref{mod Friedmann eq 1}) and (\ref{cod1jc side zt}), so the equations which describe the behaviour of the geometry at the cod-2 brane are (\ref{mod Friedmann eq 1}), (\ref{cod1jc side zz}) and (\ref{cod1jc side zt}). The system is clearly still over-constrained, but this may not be a general property of the equations (\ref{Bulkeq Cascading})--(\ref{cod2 jceq Cascading}) since it may be a consequence of the overly restrictive assumptions of Riemann-f\mbox{}latness in the bulk, and of freezing the position of the cod-2 brane.

\subsubsection{The branches of solutions}

Despite this, the system of equations (\ref{mod Friedmann eq 1}), (\ref{cod1jc side zz}) and (\ref{cod1jc side zt}) admits solutions. Since the equation (\ref{cod1jc side zt}) displays a factorized form, we expect two branches of solutions to exist, one where the def\mbox{}icit angle is independent of time ($\dot{\D} = 0$) and one where its time evolution is kinematically linked to that of the scale factor. However, inserting the relation
\beq
2 \, m_{6} - 3 \, \ep \, \frac{\dot{a}}{a} \, \cos \frac{\D}{4} = 0
\eeq
into (\ref{cod1jc side zz}) we f\mbox{}ind
\beq
4 \, \ep \, m_{6} \, \frac{\dot{a}}{a} \, \cos \frac{\D}{4} = 0 \quad ,
\eeq
which together with (\ref{mod Friedmann eq 1}) implies that $\dot{a} = \D = 0\,$. Therefore, within our ansatz the $\dot{\D} \neq 0$ branch does not contain (vacuum) cosmological solutions.

Considering now the $\dot{\D} = 0$ branch, using the equation (\ref{cod1jc side zz}) we can express $\cos \D/4$ in terms of $a$ and its derivatives, and introducing the Hubble factor $H = \dot{a}/a$ we get
\beq
\label{eq hey}
\cos \frac{\D}{4} = \ep \, \frac{2}{3} \, \frac{m_{6}}{H} \,\, \frac{\dot{H} + 4 \, H^{2}}{\dot{H} + 2 \, H^{2}} \quad .
\eeq
It is apparent that, if
\beq
\label{cha de ppk}
\ep H > 0 \hspace{2cm} \mathrm{and} \hspace{2cm} \frac{4}{3} \, \frac{m_{6}}{\abs{H}} \leq 1 \quad ,
\eeq
the equation (\ref{eq hey}) admits solutions where $H$ is constant and
\beq
\label{gostosa}
\abs{\D} = 4 \, \arccos \bigg( \frac{4}{3} \, \frac{m_{6}}{\ep H} \bigg) = 4 \, \arccos \frac{4 \, m_{6}}{3 \, \abs{H}} \quad .
\eeq
To f\mbox{}ind the possible values for $H$, we insert (\ref{gostosa}) into the f\mbox{}irst modif\mbox{}ied Friedmann equation (\ref{mod Friedmann eq 1}) obtaining
\beq
\label{mod Friedmann 1 de Sitter ok}
H^2 \, \sgn \, \D - 2 \, \ep \,  m_{5} \, H \, \sqrt{1 - \frac{16 \, m_{6}^{2}}{9 H^{2}}} + \frac{4}{3} \, m_{6} \, m_{5} \, \arccos \bigg( \frac{4\, m_{6}}{3 \, \ep H} \bigg) = 0 \quad ,
\eeq
and remembering that $\arccos x = \arctan \big( \sqrt{1 - x^2}/x \big)$ we get
\beq
\label{mod Friedmann 1 de Sitter Minami}
2 \, m_{5} \, H \, \Bigg( \frac{H}{2 \, m_5} \, \sgn \, \D - \ep \, \sqrt{1 - \frac{16 \, m_{6}^{2}}{9 H^{2}} \, } \, + \frac{2 \, m_{6}}{3 H} \, \arctan \sqrt{\frac{9 H^{2}}{16 \, m_{6}^{2}} - 1} \,\, \Bigg) = 0 \quad . 
\eeq

\subsubsection{Comments}
\label{De Sitter comments}

The solutions where $H$ is constant and solves (\ref{mod Friedmann 1 de Sitter Minami}) indeed correspond, from the intrinsic point of view of the cod-2 brane, to a de Sitter universe. In principle there are four distinct cases, since for each of the two choices for the bulk metric ($\ep = 1$ or $-1$) we can have a def\mbox{}icit angle conf\mbox{}iguration ($\sgn \, \D > 0$) or an excess angle conf\mbox{}iguration ($\sgn \, \D < 0$). As a consequence of the f\mbox{}irst inequality in (\ref{cha de ppk}), the $\ep  = 1$ conf\mbox{}igurations correspond to expanding universes while the $\ep  = -1$ conf\mbox{}igurations correspond to contracting universes. Considering the case $\ep  = 1$ and $\sgn \, \D > 0$ (expanding universe, def\mbox{}icit angle) the equation (\ref{mod Friedmann 1 de Sitter Minami}) reproduces the equation (14) of \cite{Minamitsuji:2008fz}. For a discussion of the solutions of (\ref{mod Friedmann 1 de Sitter Minami}) we remand to that paper, and for a more general analysis to a forthcoming paper we are preparing.

The analysis above helps to appreciate better the role of the auxiliary vector f\mbox{}ields. Note that, despite $\bK_{\m\n}\,$, $\ti{K}_{ab}$ and $\ti{G}_{ab}$ are def\mbox{}ined in terms of the embedding functions $\vf^{y}$, $\vf^{z}$ and their derivatives, their value at $\z = 0^{+}$ can be expressed in terms a specif\mbox{}ic combination of $\vf^{y \, \p}_{_{\! +}}$ and $\vf^{z \, \p}_{_{\! +}}\,$, the def\mbox{}icit angle $\D$ (and its derivatives). In fact we could have obtained the expressions (\ref{mod Friedmann eq 1}), (\ref{mod Friedmann eq 2}), (\ref{cod1jc side zz}) and (\ref{cod1jc side zt}) of the equations of motion working directly with the embedding functions, without introducing the auxiliary f\mbox{}ields. This would however be a bit cumbersome since the gauge freedom obscures the fact that, from a geometrical point of view, the conf\mbox{}iguration of the cod-1 brane at the position of the cod-2 brane is described in terms of one quantity only, the def\mbox{}icit angle. This is revealed using the cod-1 GNC, since the condition (\ref{Sally}) links $\hvf^{y \, \p}_{_{+}}$ and $\hvf^{z \, \p}_{_{+}}$. Although we could, for example, have solved for $\hvf^{y \, \p}_{_{+}}$ in terms of $\hvf^{z \, \p}_{_{+}}$ and expressed everything in terms of $\hvf^{z \, \p}_{_{+}}$, introducing the auxiliary f\mbox{}ields and the slope function provides a geometrically more transparent way to solve the constraint $(\hvf^{y \, \p})^{2} + (\hvf^{z \, \p})^{2} = 1\,$. Of course, since the def\mbox{}icit angle and the curvature tensors are def\mbox{}ined without any reference to the auxiliary vector f\mbox{}ields, the f\mbox{}inal result (i.e.~the expressions (\ref{mod Friedmann eq 1}), (\ref{mod Friedmann eq 2}), (\ref{cod1jc side zz}) and (\ref{cod1jc side zt}) of the equations of motion in terms of $\D$) would have been exactly the same, independently of the way used to link the curvature tensors to $\D$ (and in particular of the specif\mbox{}ic realization of the auxiliary f\mbox{}ields chosen).

Finally let's remark that, although for completeness we wrote the explicit form (\ref{Vivi}) and (\ref{ana}) of the auxiliary vector f\mbox{}ields, in our analysis we actually used only their expression (\ref{Rebeca}) at $\z = 0$. That expression can be proposed doing barely any calculation, thanks to the gauge choice where the embedding of the cod-2 brane in the bulk is f\mbox{}ixed and straight. Therefore, the knowledge of the cod-2 junction conditions, the main reasult of our paper, allows to f\mbox{}ind the de Sitter solutions with a fair economy of calculations.

%%%%%%%%%%%%%%%%%%%%%%%%%%%%%%%%%%%%%%%%%%%%%%%%%%%%%%%%%%%%%%%%%%%%%%%%%%%%%%%%%%%%%%%%%%%%%%%%%%%%%%%%%%%%%%%%

\end{document}